%% file: main.tex
\documentclass[twoside,leqno]{article}

\usepackage[letterpaper]{geometry}

\usepackage{siamproceedings}

\usepackage[T1]{fontenc}
\usepackage{amsfonts}
\usepackage{graphicx}
\usepackage[caption=false]{subfig}
\usepackage{epstopdf}
\usepackage{enumitem}
\usepackage{algorithmic}
\ifpdf
  \DeclareGraphicsExtensions{.eps,.pdf,.png,.jpg}
\else
  \DeclareGraphicsExtensions{.eps}
\fi

\usepackage{mathtools} 
\usepackage[ruled,noend, algo2e]{algorithm2e}

\newsiamremark{remark}{Remark}
\newsiamremark{hypothesis}{Hypothesis}
\crefname{hypothesis}{Hypothesis}{Hypotheses}
\newsiamthm{claim}{Claim}

\usepackage{amsopn}

\newtheorem{invariant}[theorem]{Invariant}

\newtheorem{problem}[theorem]{Problem}
\newcommand{\eps}{\varepsilon}
\newcommand{\paren}[1]{\mathopen{}\left(#1\right)\mathclose{}}
\newcommand{\poly}{\operatorname{\text{{\rm poly}}}}
\newcommand{\ellval}{10^{4}(10\Delta_{\max}+1)^{50 \sqrt{\log n}}}
\newcommand{\aval}{\tfrac{\ell}{(2^{3\sqrt{\log n}}(\Delta_{\max}+1)^{6 \sqrt{\log n}})}}
\newcommand{\avallb}{\tfrac{\ell}{(2^{3\sqrt{\log n}}(\Delta_{\max}+1)^{8 \sqrt{\log n}})}}
\newcommand{\gammaval}{\tfrac{1}{(\Delta_{\max}+1)^{3}}}

\DeclarePairedDelimiter{\ceil}{\lceil}{\rceil}
\DeclarePairedDelimiter{\floor}{\lfloor}{\rfloor}

\DeclarePairedDelimiter{\set}{\lbrace}{\rbrace}

\begin{document}

\newcommand\relatedversion{}
\renewcommand\relatedversion{\thanks{This is the full version of the similarly named SODA26 paper.}} 

\title{\Large Deterministic Dynamic Edge Colouring\relatedversion}
    \author{Aleksander B. G. Christiansen\thanks{Department of Compute, Technical University of Denmark. This work was supported by VILLUM FONDEN grant 37507 ``Efficient Recomputations for Changeful Problems''.}}

\date{}

\maketitle


\fancyfoot[R]{\scriptsize{Copyright \textcopyright\ 2026\\
Copyright for this paper is retained by authors.}}





\begin{abstract} 
 Given a dynamic graph $G$ with $n$ vertices and $m$ edges subject to insertions and deletions of edges, we show how to maintain a $(1+\eps)\Delta$-edge-colouring of $G$ without the use of randomisation. 
 More specifically, we show a deterministic dynamic algorithm with an amortised update time of $2^{\tilde{O}_{\log \eps^{-1}}(\sqrt{\log n})}$ using $(1+\eps)\Delta$ colours. If $\eps^{-1} \in 2^{O(\log^{0.49} n)}$, then our update time is sub-polynomial in $n$. 

 While there exists randomised algorithms maintaining colourings with the same number of colours [Bhattacharya, Costa, Panski, Solomon SODA'24, Christiansen STOC'23, Duan, He, Zhang SODA'19] in polylogarithmic and even constant update time, this is the first efficient deterministic algorithm to go below the greedy threshold of $2\Delta-1$ colours for all input graphs. 
 On the way to our main result, we show how to dynamically maintain a shallow hierarchy of degree-splitters with both recourse and update time in $n^{o(1)}$. We believe that this algorithm might be of independent interest. 
\end{abstract}

\section{Introduction \& Related Work}
Edge-colouring is a classic and fundamental problem studied extensively in both graph theory as well as many different models of computation~\cite{Gabow,2BERNSHTEYN,Christiansen, duan, DBLP:conf/soda/BhattacharyaGW21}. 
Given a simple graph $G = (V,E)$, a $k$-edge-colouring of $G$ is a function $c: E(G) \rightarrow [k]$ that maps every edge to a colour in some colour set $[k] = \{1,2, \dots, k\}$, 
An edge-colouring is \emph{proper} if $e \cap e' \neq \emptyset$ implies that $c(e) \neq c(e')$ for any pair of distinct edges $e, e'$. 
The \emph{chromatic index} $\chi'(G)$ of $G$ is then the smallest $k$ for which a proper $k$-edge-colouring of $G$ exists. 
It is not difficult to check that $\chi'(G) \geq \Delta$ if $G$ has maximum degree $\Delta$. 
A classic result of Vizing~\cite{vizing1964estimate} shows that in fact $\chi'(G) \leq \Delta + 1$. 
Deciding whether $\chi'(G) = \Delta$ or $\chi'(G) = \Delta +1$ is, however, not easy, and it was shown by Holyer that determining $\chi'(G)$ exactly is NP-complete~\cite{Ian}. 

Recently, the edge-colouring problem has received a lot of attention in the dynamic algorithms community. Here one studies \emph{dynamic graphs}, where edges are continually inserted and deleted, and the challenge is to maintain some proper edge-colouring of the graph as it is subjected to these updates. 
There has been some exciting progress on developing efficient randomised dynamic algorithms using only $(1+\eps)\Delta$ colours. 
This number of colours go significantly below the greedy-threshold of $2\Delta - 1$ colours.
Duan, He, and Zhang~\cite{duan} showed how to maintain such a $(1+\eps)\Delta$-edge-colouring with an update time of $O(\tfrac{\log^{7} n}{\eps^2})$ provided that $\Delta = \Omega(\tfrac{\log^2 n}{\eps^2})$. 
Christiansen got rid of all requirements on $\Delta$ while maintaining a worst-case update time in $O(\poly(\eps^{-1},\log n))$. 
Very recently Bhattacharya, Costa, Panski, and Solomon~\cite{CostaEtAl} pushed the update time all the way down to $\tilde{O}(\poly(\eps^{-1}))$ under some additional assumptions on $\Delta$ and $\eps$. 

However, all of the above algorithms rely on randomness in several and crucial ways, and the landscape of deterministic dynamic algorithms is a lot less colourful. 
Barenboim and Maimon~\cite{barenboim2017fully} gave a dynamic algorithm with $O(\sqrt{\Delta})$ worst-case update time using $O(\Delta)$ colours. This was subsequently improved by Bhattacharya, Chakrabarty, Henzinger and Nanongkai~\cite{bhattacharya2018dynamic} who gave a fully dynamic algorithm for computing a $(2\Delta-1)$-edge-colouring with worst-case $O(\log \Delta)$ update time. 
Bhattacharya, Costa, Panski, and Solomon~\cite{DBLP:journals/corr/abs-2311-08367} and independently Christiansen, Rotenberg, and Vlieghe~\cite{Vlieghe} showed how to maintain a  $(\Delta+O(\alpha))$-edge-colouring in amortised $O(\poly(\log n))$ update-time, where $\alpha$ is the \emph{arboricity} of $G$. 
However, if $\Delta \sim \alpha$, this algorithm will not always use a number of colours below the greedy threshold as $\alpha$ can be as big as $\tfrac{\Delta}{2}$. 
In particular, until now no deterministic and dynamic algorithm always using less than $2\Delta - 1$ colours with sub-linear update time was known.

In general it is an interesting and challenging task to design deterministic algorithms with guarantees similar to the randomised ones. 
In this paper, we will address this gap between randomised and deterministic dynamic algorithm for $(1+\eps)\Delta$-edge-colouring by providing a deterministic and fully-dynamic algorithm  for maintaining a $(1+\eps)\Delta$-edge-colouring in $n^{o(1)}$ amortised update-time:
\begin{theorem}[Informal]
    Let $G$ be a dynamic graph and $\eps^{-1} \leq 2^{O(\log^{0.49} n)}$ a given parameter. 
    Then, we can maintain a proper $((1+\eps)\Delta)$-edge-colouring of $G$ in $n^{o(1)}$ amortised update time per operation.
\end{theorem}

\subsection{Main technical challenges} \label{sct:techchal}
In this section, we will pinpoint two issues towards designing deterministic dynamic edge-colouring algorithms with few colours. 
These are barriers we know how to overcome in either the static setting or by using randomisation, but so far, to the best of our knowledge, no current deterministic and dynamic techniques can overcome them.  

Many static algorithms, distributed algorithms, and dynamic randomised algorithms produce edge-colourings with few colours via a two-step approach. The first step is a form for \emph{dimensionality reduction}; one attempts to reduce the problem to a case where $\Delta$ is small. The second step is then to apply an algorithm that is fast and efficient for graphs with small maximum degree. 

Many static algorithms~\cite{Gabow,Costa2,Sinnamon} and many distributed algorithms~\cite{ghaffari2018deterministic,split1,split2} address the first step by computing a hierarchy of \emph{degree-splitters}. 
A degree-splitter splits a graph $G$ into two covering subgraphs $G_1, G_2$ such that the maximum degree of $G_i$ is roughly half of the maximum degree of $G$. By applying this procedure recursively one can reduce the problem to $\tfrac{\Delta}{\poly(\eps^{-1},\log n)}$ sub-problems each of maximum-degree $\poly(\eps^{-1},\log n)$. 
However, these static approaches for splitting the degree are very precise, and therefore seem prone to very high recourse in the dynamic setting. 
Hence, we shall take an approach much more resilient to local changes in the graph. 

If the maximum degree is large enough, it is easy to realise something similar in the randomised setting, even dynamically against an oblivious adversary. 
Indeed, one can place each edge into one of $\tfrac{\Delta}{\poly(\eps^{-1},\log n)}$ subgraphs uniformly at random. 
It then follows from a standard application of Chernoff-bounds that the maximum degree of each such subgraph is split nicely with high probability. 
As soon as the adversary becomes adaptive, however, this approach immediately breaks down. 
In the randomised adaptive setting, one can instead rely on a different approach for reducing $\Delta$: one can sample uniformly at random a \emph{palette} $S \subset [t]$ of colours such that the graph induced by all edges coloured with a colour from $S$, say $G[S]$, exhibits useful properties. Indeed, one can show that with high-probability the maximum degree of $G[S]$ is small and every vertex in $G$ has a free colour in $S$. 
One can then colour a new edge, by viewing it as being inserted into $G[S]$. 
This is the approach taken in~\cite{Christiansen,duan}. 
It is, however, not clear how to perform such a reduction in the deterministic dynamic setting; naively building a hierarchy can easily result in an insurmountable recourse (see Section~\ref{sec:overview}), and sampling colours or edges relies crucially on the power of randomisation. 

The second challenge is to efficiently colour a low-degree graph. 
In the static setting one could rely on parallelism~\cite{Gabow}, and in the dynamic \& distributed settings recent progress has shown that the recourse for extending colourings on low-degree graphs is small~\cite{2BERNSHTEYN, BernDhawan, Christiansen, duan} and, more importantly, can be efficiently realised using randomisation~\cite{2BERNSHTEYN, BernDhawan, Christiansen, duan}. However, naively searching for such an extending sequence of recolourings can be expensive, and so it is unclear how to locate it efficiently without using randomisation. 

\subsection{Further Related Work}
Gabow, Nishizeki, Kariv, Leven, and Tereda~\cite{Gabow} showed how to compute a $(\Delta + 1)$-edge-colouring in $\tilde{O}(m \sqrt{n})$ or $\tilde{O}(m \Delta)$ time. This improved over the $O(mn)$ time algorithm implied by the work of Vizing~\cite{vizing1964estimate}. 
Recently this was improved by Sinnamon~\cite{Sinnamon}, and in special settings by Bernshteyn and Dhawan~\cite{BernDhawan} and by Bhattacharya, Costa, Panski, and Solomon~\cite{Costa2}.
Alon~\cite{alon2003simple} and Cole, Ost and Schirra~\cite{cole2001edge} gave a near linear-time algorithm for $\Delta$-edge-colouring bipartite graphs. By applying a reduction due to Karloff and Shmoys~\cite{karloff1987efficient}, these algorithms yield $3 \lceil \frac{\Delta}{2} \rceil$-edge-colouring algorithms for general graphs. 
When $\Delta = \Omega{(\varepsilon^{-1} \log n)}$, Duan, He and Zhang~\cite{duan} gave a $O(\varepsilon^{-2}m \log^6 n)$ time algorithm for computing a $(1+\varepsilon)\Delta$-edge-colouring.
This has recently been improved by Bhattacharya, Costa, Panski, and Solomon~\cite{CostaEtAl} to $\tilde{O}(m\poly(\eps^{-1}))$ time for many values of $\Delta$ and $\eps$.
A very recent line of work~\cite{New1,New3,New2} has managed to bring the time needed for computing a $(\Delta+1)$-edge-colouring in a static graph all the way down to $O(m \log \Delta)$.

The edge-colouring problem has also been extensively studied in the LOCAL model~\cite{2BERNSHTEYN,BernDhawan,pettie, Christiansen,ghaffari2018deterministic,split1,split2,su2019towards} (see for instance~\cite{pettie,ghaffari2018deterministic} for extensive surveys). 
Some of these algorithms either rely on or directly consider the degree-splitting problem~\cite{ghaffari2018deterministic,split1,split2}. 
There has also been a lot of recent progress in online models, see for example~\cite{ONew1,ONew2,ONew3}. 

\subsection{Results}
Our main result it the following theorem: 
\begin{theorem}\label{thm:FullColour}
    Let $G$ be a dynamic graph, and $0<\eps<1$ a given parameter.
    We can maintain a proper $(1+\eps)\Delta$-edge-colouring of $G$ in $2^{\tilde{O}_{\log \eps^{-1}}(\sqrt{\log n})}$ amortised\footnote{Here $\tilde{O}_{\log \eps^{-1}}$ hides $\log \log n$ and $\log \eps^{-1}$ factors.} update time per operation.
\end{theorem}
On the way to this result, we present several algorithms which we believe might be of independent interest. 
First of all, we design what we call a dynamic $t$-splitter; a dynamic algorithm maintaining a $t$-edge-colouring such that the graph induced by each colour has low-degree.
\begin{theorem}\label{thm:tsplit}
    Let $G$ be a dynamic graph, and suppose $\eps > 0$, $t \in \mathbb{N}$, and an upper bound $\Delta_{\max}$ on the maximum degree of $G$ throughout the entire update sequence is given. 
    Assume that $\tfrac{\Delta_{\max}}{t} \geq \tfrac{10^4 \log^2 n}{\eps^{2}}$. Then there is a dynamic algorithm that computes a $t$-edge-colouring of $G$ such that the maximum degree of $G[i]$ is at most $(1+\eps)\tfrac{\Delta_{\max}}{t}$. The algorithm has an amortised update-time in $O(t^3\eps^{-4} \log^3 m)$ and an amortised recourse in $O(\tfrac{\log n}{\eps^2})$.
\end{theorem}
Second of all, we show how to build a shallow hierarchy of $t$-splitters that splits the degree all the way down to $\poly(\log n, \eps^{-1})$. We believe that this might be a useful algorithmic tool for developing efficient dynamic algorithms for other problems as well.
\begin{theorem} \label{thm:splitHier}
    Let $G$ be a dynamic graph, and let $1 > \eps > 0$ and an upper bound $\Delta_{\max} \geq \tfrac{10^7 \log^5 n}{\eps^{3}}$ on the maximum degree of $G$ through the entire update sequence be given. 
    Set $t_1=\floor{2^{10\sqrt{\log n}}}$ and $t_2 = \floor{\tfrac{\log n}{\eps}}$.
    Then there exists a dynamic algorithm maintaining parameters $h,i,j$ with $h = i+j \leq 4 \sqrt{\log n}$ and a set of graphs $\mathcal{G}_{h}$ satisfying that:
    \begin{enumerate}
        \item $\{E(G_i)\}_{G_i \in \mathcal{G}_{h}}$ partitions $E(G)$
        \item $|\mathcal{G}_{i+j}| \leq t_1^{i}t_2^{j} $
        \item $\hat{\Delta}_{i+j}  \leq (1+\tfrac{\eps}{16})\tfrac{\Delta_{\max}}{t_1^{i}t_2^{j}} \leq \poly{\paren{\tfrac{\log n}{\eps}}}$
    \end{enumerate} 
    The algorithm has amortised recourse and update time both in $2^{\tilde{O}_{\log \eps^{-1}}(\sqrt{\log n})}$.
\end{theorem}
Finally, we show how to efficiently $(\Delta_{\max}+1)$-edge-colour graphs with low maximum degree:
\begin{theorem}\label{thm:colAlg}
    Let $G$ be a dynamic graph, and let $\Delta_{\max}$ be a known upper-bound on the maximum degree throughout the entire update sequence. 
    Then, we can maintain a proper $(\Delta_{\max}+1)$-edge-colouring of $G$ in $(\Delta_{\max}+1)^{\tilde{O}(\sqrt{\log n})}$ worst-case update time per operation.
\end{theorem}

\section{Informal Overview} \label{sec:overview}
In this section, we will give a high-level overview of the main ideas behind the algorithm. 
We first introduce the overall approach:

\subsection{The overall approach}
From now on, unless otherwise specified, we will assume that we have a known upper bound $\Delta_{\max}$ on the maximum degree $\Delta$. This assumption will later be removed by using a degree-scheduler identical to the one used by Duan, He, and Zhang~\cite{duan}. 

We will take a two-step approach that, at a high-level, is inspired by the distributed algorithm of Ghaffari, Kuhn, Maus, and Uitto~\cite{ghaffari2018deterministic}. 
The idea is to partition the graph into a set of at most $\tfrac{\Delta}{t}$ subgraphs $\mathcal{G}_{h}$ such that $\{E(G_i)\}_{G_i \in \mathcal{G}_{h}}$ partitions $E(G)$ and such that $\Delta(G_{i}) \leq (1+\tfrac{\eps}{2})\cdot{} t$ for all $G_{i} \in \mathcal{G}_{h}$. 
Then one can colour each $G_{i} \in \mathcal{G}_{h}$ with a separate palette using some algorithm that is efficient in low-degree graphs, and then combine all of these colourings into a $(1+\eps)\Delta$-edge-colouring of the original graph $G$. 

As a warm-up, we briefly sketch how to boot-strap a simple static algorithm for computing a $(1+\eps)\Delta$-edge-colouring in near-linear time using the algorithms of Gabow et Al.~\cite{Gabow}. 
This algorithm was recently, and independently, obtained by Elkin and Khuzman~\cite{DBLP:journals/corr/abs-2401-10538}, while this manuscript was under preparation, and we therefore refer to their manuscript for a formal description of the algorithm.

The idea is the following: a \emph{degree-splitter} with discrepancy $\delta$ is a partition of the edge set into two sets $E_1$ and $E_2$, such that for all vertices $v$, we have that $d_{G[E_1]}(v),d_{G[E_2]}(v) \leq \tfrac{\Delta}{2} + \delta $, where $G[E_i]$ is the graph induced in $G$ by the set $E_i$ i.e.\ the subgraph of $G$ containing all edges of $E_i$.
This is a relaxed definition of degree-splitting similar to the one used in~\cite{split2}, but other variants have also been studied (see for example~\cite{split1}).
Given an efficient degree-splitter with small enough discrepancy, one can easily compute a partition $\mathcal{G}_{h}$ of the desired form by recursively applying the degree splitter to $G_1$ and $G_2$ to build a hierarchy whose lowest level form $\mathcal{G}_{h}$. 

Gabow et Al.~\cite{Gabow} showed how to compute a degree-splitter with a discrepancy of only $2$ in the static setting. 
By recursively applying this degree splitter roughly $h = \log \Delta - \log \tfrac{\log n}{\eps}$ times (or 0 times if $\Delta = O(\tfrac{\log n}{\eps})$) one can achieve a split into roughly $\tfrac{\Delta}{\tfrac{\log n}{\eps}}$ subgraphs, each with maximum degree at most $\tfrac{\log n}{\eps} + O(\log n + \sum_{i=0}^{\infty} \tfrac{2}{2^{i}})$.
By rescaling $\eps$ a bit, this yields the desired partition. 
Finally, one can colour each graph in $\mathcal{G}_{h}$ with a distinct palette efficiently using another algorithm of Gabow et Al.~\cite{Gabow}. 

As we have already briefly touched upon in Section~\ref{sct:techchal}, the two main challenges with making such an approach dynamic are the following: $1)$ how to maintain the degree splitters, and $2)$ how to efficiently colour the low-degree graphs. 

As for $1)$, we note that first of all, as far as we are aware, no deterministic dynamic degree-splitting algorithm is known. 
While it is easy to devise a randomised such algorithm against an oblivious adversary for graphs with not too small maximum degree (simply partition the edges u.a.r\ into $2$-edge sets), it is unclear how to do this deterministically. 
Second of all, if a deterministic degree-splitter has recourse at least $10$, then naively maintaining a hierarchy of degree-splitters becomes infeasible as the recourse at depth $h$ of the hierarchy can be as large as $10^h$. 
In order to split the degree from $n^{0.75}$ to $n^{0.25}$, we would require $h \geq \tfrac{1}{2}\log n$, thus resulting in super-linear recourse, and therefore a super-linear update time. 

We overcome this as follows: instead of designing a dynamic degree-splitter, we consider the dynamic $t$-splitting problem, where we are interested in splitting the graph into $t$-subgraphs, each with some restrictions on the maximum degree: 
\begin{problem}
    Given a dynamic graph $G$ subject to insertions and deletions, maintain a colouring $c: E(G) \rightarrow [t]$ such that for all vertices $v \in V(G)$ and all colours $i \in [t]$, we have that the number of edges coloured $i$ incident to $v$, $d_i(v)$, is upper bounded by $d_i(v) \leq (1+\eps) \frac{\Delta}{t}$.
\end{problem}
We show in Section~\ref{sct:tsplits} that we can, in fact, maintain a dynamic $t$-splitter in $\tilde{O}(\poly(t, \log(n), \eps^{-1}))$ update time, but crucially with only $\poly(\log(n), \eps^{-1})$ recourse. 
This data structure is far from trivial, and requires many ideas of its own. We will discuss it in Section~\ref{sct:discussTsplits}. 

Since the $t$-splitter has low recourse, we can tune $t$ very aggressively. In turn, we can get a dynamic $\frac{\Delta}{\poly (\log n, \eps^{-1})}$-splitter by maintaining a very shallow hierarchy of correctly tuned $t$-splitters. 
Indeed, we can set $t \sim 2^{\sqrt{\log n}}$ and split the degree down to $\sim 2^{\sqrt{\log n}}$ using a hierarchy of depth only $h \sim \sqrt{\log n}$. This allows us to bound the recourse $R(h)$ at depth $h$ as: 
\[
R(h) \leq \poly(\log(n), \eps^{-1})^{h} = 2^{\tilde{O}_{\log \eps^{-1}}(\sqrt{\log n})}
\]
Then, we can set $t' \sim \poly(\log n, \eps^{-1})$ to further split the degree down to $\poly(\log n, \eps^{-1})$, while only increasing the recursive depth by an extra additive factor of $\sqrt{\log n}$. 
Since each update takes only $\tilde{O}(\poly(t, \log(n), \eps))$ time, this allows us to split the degree down to $\poly(\log n, \eps^{-1})$ with both recourse and update time in $2^{\tilde{O}_{\log \eps^{-1}}(\sqrt{\log n})}$. 

In order to address $2)$, we will make use of \emph{multi-Step Vizing chains}, a term coined by Bernshteyn~\cite{2BERNSHTEYN} and a topic that has seen a lot of recent development~\cite{2BERNSHTEYN, BernDhawan,Christiansen,grebik2020measurable}. 
Here, one attempts to find a sequence of recolourings that creates an uncoloured edge with a free colour.
The idea is to construct such a sequence of recolourings by first constructing a $1$-step Vizing chain, which is a construction that goes back all the way to Vizing's original paper~\cite{vizing1964estimate}: a Vizing chain  consists of two parts: a) a \emph{fan} $F$ of edges and b) a bichromatic path $P$ of edges coloured with two colours $\kappa_1$ and $\kappa_2$. 
The fan $F$ consists of a center vertex $u$ together with edges $uw_1, \dots, uw_k$ such that the colour of $uw_{i+1}$ is available at $w_i$, meaning that no edges incident to $w_i$ has the colour $c(uw_{i+1})$.
The bichromatic path $P$ then has to begin at $w_k$ and consist only of edges coloured $\kappa_1$ and $\kappa_2$, where $\kappa_1$ is available at $u$ and $\kappa_2$ is available at $w_k$. 
A Vizing chain can be \emph{shifted} by recolouring edges in $F$ and subsequently in $P$ as shown in Figure~\ref{fig:shiftVC} to the left. 
\begin{figure}%
    \centering
    \subfloat{{\includegraphics[width=4cm]{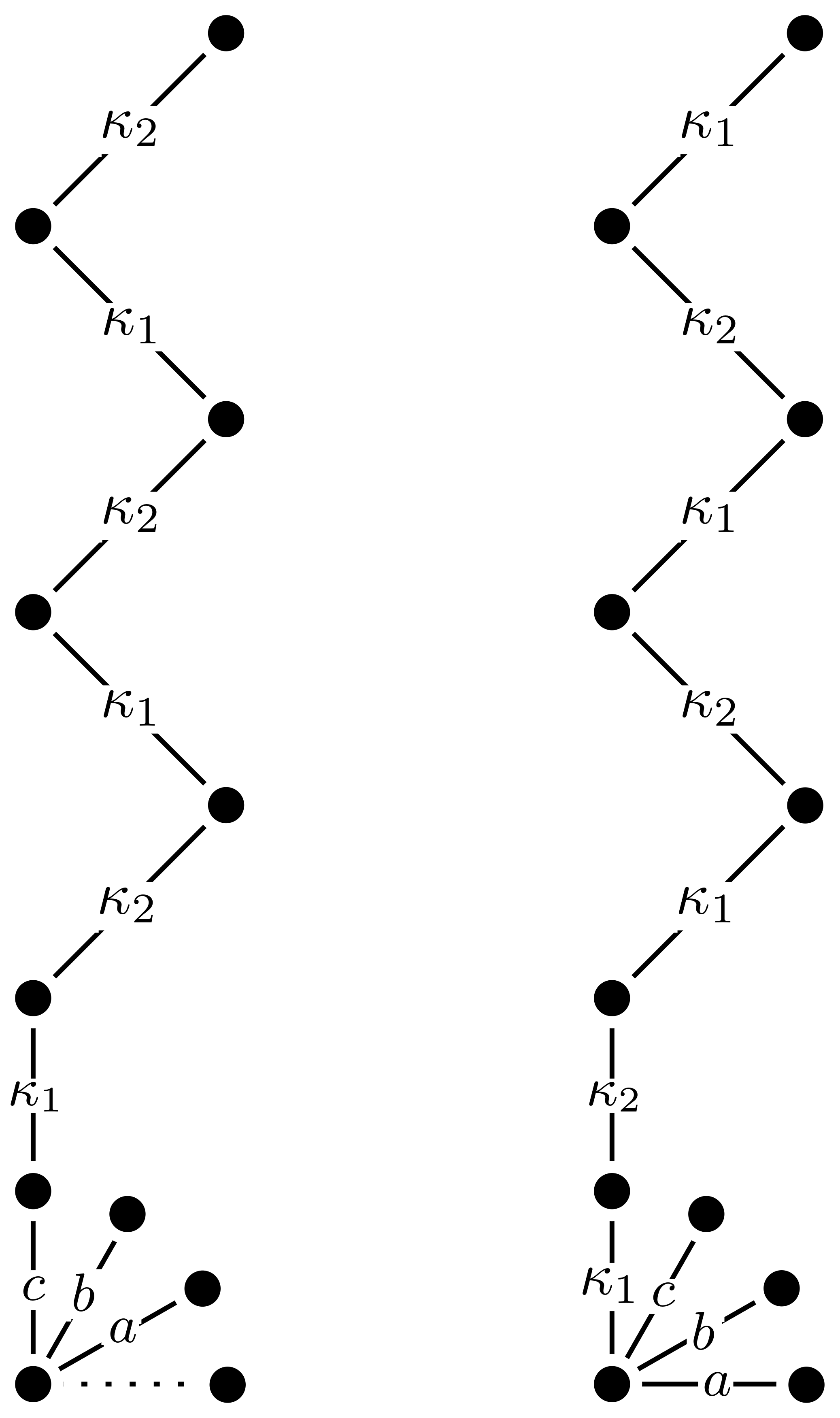} }}%
    \hspace{30mm}%
    \subfloat{{\includegraphics[width=4cm]{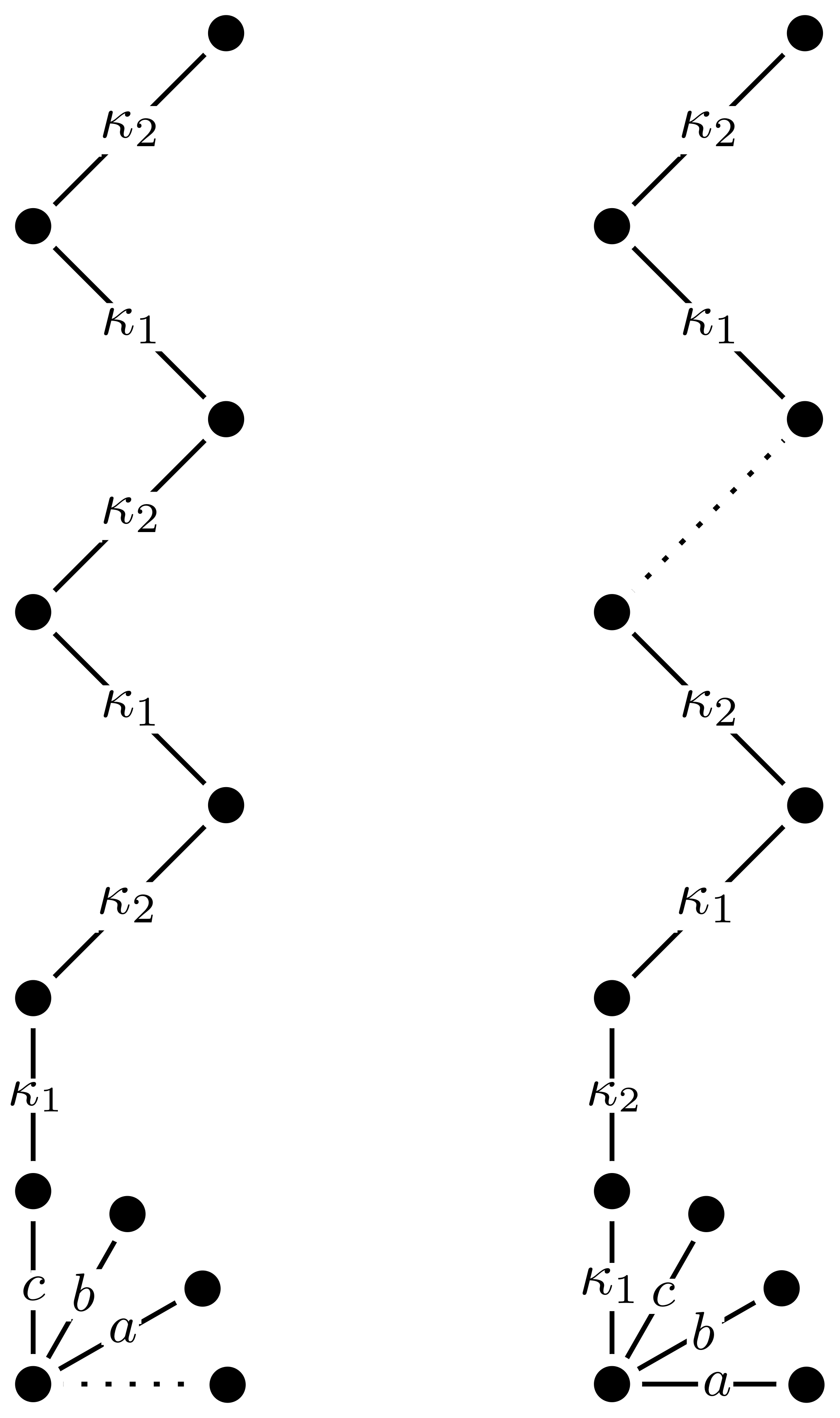} }}%
    \caption{An augmenting Vizing chain is shifted on the left, and a truncated Vizing chain is shifted on the right. The above illustration is originally due to Christiansen~\cite{Christiansen}.}
    \label{fig:shiftVC}%
\end{figure}
Observe that if $P$ is not too long, we are already happy with the construction, as we need only recolour few edges. However, $P$ could potentially have length $\Omega(n)$.
In this case, we can \emph{truncate} it by uncolouring an edge along $P$ and only shifting the first part of $P$ as shown in Figure~\ref{fig:shiftVC} to the right. 

Then we can build a new Vizing chain on top of the new uncoloured edge in the hope that this chain has a shorter length. We can iterate this to construct \emph{multi-step Vizing chains} (see Figure~\ref{fig:multistep}), as has been done in similar fashion in~\cite{2BERNSHTEYN,BernDhawan,Christiansen,grebik2020measurable}. 
The goal is to get a short Vizing chain which is \emph{augmenting} in the sense that after shifting it, the remaining uncoloured edge has a free colour.
In order to more readily reason about these chains, we will require them to be \emph{non-overlapping} in the sense that we impose restrictions on how the $i^{\text{th}}$ extension can overlap with earlier chains. 
\begin{figure}%
    \centering
    \subfloat{{\includegraphics[width=14cm]{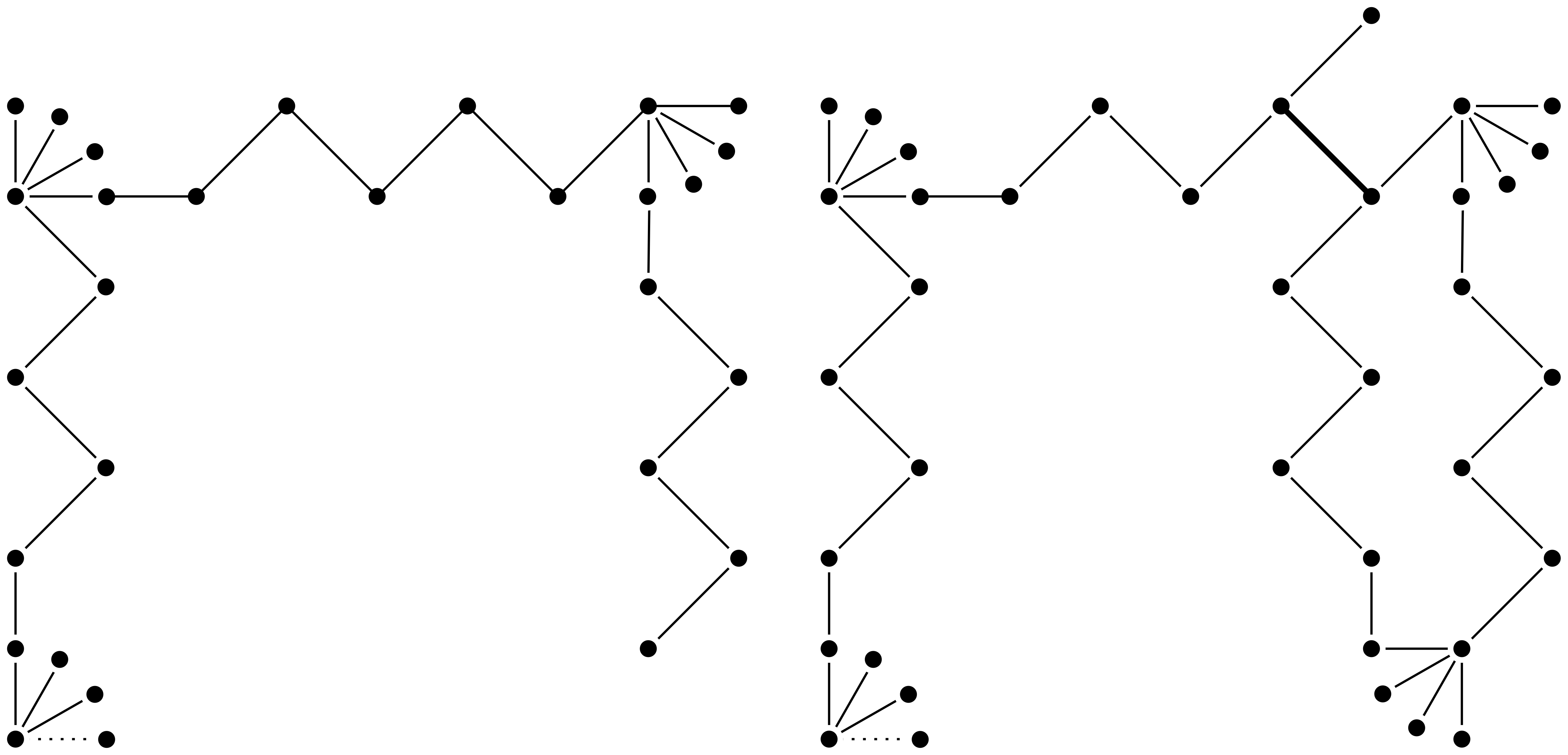} }}%
    \caption{On the left, we see a non-overlapping multi-step Vizing chain. On the right, we see a multi-step Vizing chain overlapping in a single edge. The above illustration is originally due to Christiansen~\cite{Christiansen}.}
    \label{fig:multistep}%
\end{figure}

The main issue is that all known results about Vizing chains are either existential or obtained via randomised constructions. 
As noted by Bernshteyn \& Dhawan~\cite{BernDhawan} naively applying the argument of Christiansen~\cite{Christiansen} to search for a Vizing chain can result in linear update time, and so it seems hopeless to perform some deterministic search for a short and augmenting Vizing chain. 
However, we will show that one can efficiently maintain data structures containing enough information to be able to locate certain classes of short and augmenting Vizing chains efficiently. 
The main observation is that, under the correct tuning of parameters, one has that while one uncoloured edge can be part of many different multi-step Vizing chains, not too many vertices can actually reach the edge via a multi-step Vizing chain. 
This means that assigning an edge a new colour, does not result in the need of updating too many different data structures. 
Maintaining these data structures is quite technical. Therefore, we will give a separate overview of how to maintain and use these data structures in Section~\ref{sct:overviewlowcol}. 

\subsection{Dynamic $t$-splitters} \label{sct:discussTsplits}
In this section, we will give an overview of how to maintain a dynamic $t$-splitter with low recourse. 
We will assume that $\Delta_{\max}, t, \eps$ satisfy that $\tfrac{\Delta_{\max}}{t} = \Omega(\poly(\log n, \eps^{-1}))$, since we only need to apply $t$-splitters when this is the case.
We begin by sketching an algorithm for the case $t = 2$. We note already now that this case is considerable easier than the case $t > 2$. For the case $t = 2$ , the algorithm has worst-case guarantees, but for $t > 2$ all guarantees are amortised.
After sketching the approach, we discuss the technical difficulties with generalising this approach, and sketch how we overcome them. 

We will colour the graph with $2$ colours, colour $0$ and colour $1$. For a vertex $v$, we define its surplus of colour $i$ to be $s_{i}(v) = \max \{d_{i}(v)-\frac{\Delta_{\max}}{t},0\}$.
We will maintain a colouring satisfying the following local invariant for some parameter $\eta$:
\begin{invariant}
    For all $uv \in E(G)$, we have that $c(uv) = i$ implies
    \[
    \max\{s_{i}(u), s_{i}(v) \} \leq \eta + \max\{s_{i+1}(u), s_{i+1}(v)\}
    \]
    where addition is performed modulo $2$.
\end{invariant}
This invariant is useful for two reasons: first of all, one can show that if it holds everywhere, then the maximum surplus is never too big, and second of all, whenever an edge $uv$ is inserted one can restore the invariant by only recolouring edges along a path with alternating colours. These paths bear some resemblance to the augmenting paths used in a distributed algorithm due to Ghaffari and Su~\cite{split2}. 

Indeed, to see that the invariant implies a low maximal surplus, observe that if $s_{i}(v)$ is large, then any vertex $w$ neighbouring $v$ via an edge coloured $i$ must have large surplus $s_{i+1}(w) \geq s_{i}(v) - \eta$. 
Applying this approach recursively yields sets $V_1, V_2, \dots, V_k$ where $u \in V_{j}$ implies that $s_{i+j-1 \mod 2}(u) \geq s_{i}(v) - (j-1)\eta$. 
In particular, one has:
\[
|E(G[\cup_{h=1}^{k} V_{h}])| \geq |\cup_{h=1}^{k-1} V_{h}| \cdot{} (s_{i}(v) - (k-1)\eta)
\]
Indeed, observe that one double counts no edges as vertices in $V_j$ cannot be connected by an edge coloured $i+j-1 \mod 2$, without violating the local invariant.

For a reasonable choice of $k,\eta$ and $s_{i}(v)$, we have that $s_{i}(v) - (k-1)\eta > (1+\tfrac{\eps}{16})\tfrac{\Delta}{2}$. 
In particular, since the maximum subgraph density $\rho(G)  = \max_{J \subset G} \{\frac{|E(J)|}{|V(J)|} \} \leq \frac{\Delta}{2}$ by the hand-shaking lemma, it follows by induction that $|\cup_{h=1}^{k} V_{h}| \geq (1+\tfrac{\eps}{16})|\cup_{h=1}^{k-1} V_{h}|$.
Hence, we conclude that $k = O(\tfrac{\log n}{\eps})$, which naturally bounds how big $s_{i}(v)$ can be. 

In order to see that one need only recolour edges along an alternating path, assume the edge $uv$ is inserted and receives colour $i$ minimizing $\max\{s_{i}(u), s_{i}(v) \}$. The choice of $i$ readily implies that at most one of $s_{i}(u)$ and $s_{i}(v)$ can be larger than $\frac{\Delta_{\max}}{2}$. If $s_{i}(w) \leq \frac{\Delta_{\max}}{2}$ after $uv$ is coloured by $i$, then no edge incident to $w$ will suddenly violate the invariant. So we need only focus on edges incident to the endpoint with the larger surplus in $i$. WLOG say $s_{i}(u) > \frac{\Delta_{\max}}{2}$. 
If so, then only edges coloured $i$ incident to $u$ can become violating. If an edge $ux$ becomes violating, then it must be that
\[
\max\{s_{i}(u), s_{i}(x) \} > \eta + \max\{s_{i+1}(u), s_{i+1}(x)\}
\]
Hence, recolouring $ux$ to colour $i+1 \mod 2$ restores the invariant at all edges incident to $u$, since $s_{i+1}(u) < \frac{\Delta_{\max}}{2}$. 
Now some edges coloured $i+1$ incident to $x$ might be violating. These can be fixed as above. Note that since $s_{i+1}(x) \leq s_{i}(u)-\eta$, this process cannot be iterated too many times before one visits a vertex with less than $\eta$ surplus, and one restores the invariant. 
In order to implement the above efficiently, one can locate violating in-edges by updating information lazily via a round-robin scheme. 
The use of the round-robin scheme is inspired by a sequence of recent work on dynamic algorithms for maintaining an estimate of the maximum subgraph density $\rho$, see~\cite{sawlani2020near,Chekurietal}, where a lazy round robin scheme is used to locate a chain of directed edges to re-orient. 
We will not speak more about this scheme, and instead focus on a more general scheme needed for the $t > 2$ case later on. 

In order to generalise the above to $t>2$, we will work with the following local invariant instead: 
\begin{invariant}
    For all $uv \in E(G)$ and all $i \in [t]$, we have that 
    \[
    L(uv) = s_{c(uv)}(u) + s_{c(uv)}(v) \leq \eta + s_{i}(u) + s_{i}(v)
    \]
\end{invariant}
Many of the difficulties that arise when generalising the above approach to $t > 2$ stem from the fact that the problem becomes a lot less symmetric. The generalisation of the local invariant is chosen carefully to account for this asymmetry.
Again we wish to show that $1)$ the local invariant implies that the colouring is $t$-splitter, and $2)$ that we can restore the invariant in some efficient manner. 

As for $1)$ the asymmetry is reflected in the fact that one can no longer extract dense subgraphs as in the argument above. 
Indeed, for $t > 2$ it is unclear how to extend a density based argument. Before there was a natural way of assigning a high-surplus edge to the corresponding high-surplus endpoint. This allowed us to extract a dense subgraph, which we could use to show that the maximum surplus could not be too large.
However, for $t>2$ it is unclear which endpoint one should assign a high-surplus edge to; indeed assigning it to the high-surplus vertex as before risks that some vertices do not get assigned more than $\tfrac{\Delta_{\max}}{2}$ edges. Furthermore, there is the additional caveat that it might not even be well-defined which endpoint has high-surplus, as both endpoints could contribute equally to $s_{c(uv)}(u) + s_{c(uv)}(v)$. 

Instead, we show in Section~\ref{sct:tsplits} that if $s_{c(uv)}(u) + s_{c(uv)}(v)$ is large, then one can extract a sequence of bipartite subgraphs $\{\mathbb{B}_j\}_{j = 1}^{10\log(n)\eps^{-1}}$, possibly with low density, containing an exponentially increasing number of vertices. 
This again allows us to bound $s_{c(uv)}(u) + s_{c(uv)}(v)$, and show that we indeed maintain a $t$-splitter.

The asymmetry is also reflected in $2)$. Before one could restore the invariant by recolouring a path, since fixing violating edges coloured $i$ did not cause edges coloured $i+1$ at the same vertex to become violating. 
However, if one has multiple colours to choose from, one might have to recolour multiple edges incident to each vertex. 
Hence, one has to analyse the recourse completely differently. We show that each recolouring of a violating edge cause the following potential to drop:
\[
\Phi(G,c) = \sum \limits_{i=1}^{t} \sum \limits_{v \in V(G)} \sum \limits_{j=1}^{s_{i}(v)} j
\]
thus allowing us to bound the recourse of the algorithm in an amortised fashion. 

Finally, we arrive at the dynamic $t$-splitter by carefully adjusting the lazy round-robin schemes and the data structures used to account for the asymmetry and the fact that $t>2$. 

\subsection{$(\Delta_{\max}+1)$-edge-colouring of low degree graphs} \label{sct:overviewlowcol}
In the following section, we sketch our approach for dynamically colouring low degree graphs. We mostly try to convey some intuition, so some parts of the descriptions are purposefully left a bit vague.
The starting point for our algorithm is the following: suppose that given any proper partial $(\Delta_{\max}+1)$-edge-colouring and an uncoloured edge $e$, there exists an augmenting, non-overlapping $\sqrt{\log n}$-step Vizing chain with length $\Delta_{\max}^{O(\sqrt{\log n})}$. Indeed, this is the case and follows from the work of Christiansen~\cite{Christiansen}. 

While a naive search for such a chain might take $O(\poly(\Delta_{\max})n)$ time, one can readily observe that any single edge can only be part of a $\sqrt{\log n}$-step Vizing chains beginning at one of $\Delta_{\max}^{O(\sqrt{\log n})}$ different points. 
Indeed, any coloured edge is part of at most $\Delta_{\max}+1$ different bichromatic paths and at most 2 different fans. 
Since any bichromatic path can be a part of at most $2(\Delta_{\max}+1)$ different 1-step Vizing chains, the statement follows for $1$-step Vizing chains. 
To see that the statement holds for more steps, observe that in order to extend the chain through $e$, one has to be able to reach an edge incident to an edge sharing an endpoint with $e$. Since there are at most $2(\Delta_{\max}+1)$ such edges, one can apply the above argument again to see that there are at most 
$O(\poly(\Delta_{\max}))$ choices for the starting point of $2$-step Vizing chains. Iterating this argument then gives a similar bound for $i$-step Vizing chains.

The above argument relies crucially on the fact that we require the Vizing chain to be non-overlapping (see Definition~\ref{def:nonO} for a formal definition). 
Indeed, an issue could be that shifting the first step of a Vizing chain could create many new bichromatic paths not accounted for by the above argument. 
But for that to happen, the extension must circle back and \emph{overlap} or \emph{intersect} an ealier Vizing chain. 
This in turn means that the Vizing chain becomes \emph{overlapping}, and thus is not considered for the above argument. 
This type of argument is not new, and has been used for instance by Bernshteyn~\cite{2BERNSHTEYN} and Christiansen~\cite{Christiansen} to analyse the recourse of certain multi-step Vizing chain constructions. 

In this paper, we use this fact very explicitly in our algorithm. The idea is to explicitly maintain enough information about bichromatic paths in the graph to be able to construct an augmenting and non-overlapping $\sqrt{\log n}$-step Vizing chain with small step-length fast: if  we know for each vertex if there is an augmenting $1$-step Vizing chains going through it, then one can update information about augmenting and non-overlapping $2$-step Vizing chains with small step lengths fast: simply consider all $\Delta_{\max}^{O(\sqrt{\log n})}$ different extension points, determine which extension have an augmenting $1$-step Vizing chain and extend via these points. 
We can then update information about $i$-step Vizing chains in a similar, iterative fashion.

Since any edge can be reached by at most $\Delta_{\max}^{O(\sqrt{\log n})}$ relevant Vizing chains, (re)-colouring or un-colouring an edge only affects $\Delta_{\max}^{O(\sqrt{\log n})}$ data structures. 
Hence, we can update all such data structures by first determining and updating all data structures containing affected $1$-step Vizing chains, and then subsequently use these to update the $2$-step structures efficiently, and so on. 

The above approach has some glaring issues: first of all, we need the extensions to be non-overlapping. This requires us to check whether each extension is overlapping. 
However, by extending in a controlled manner~\cite{2BERNSHTEYN,grebik2020measurable}, one can ensure that not too many extension are overlapping, and thus one can actually afford to check for each chain whether it is overlapping, in $\tilde{O}(\text{step-length})$ time per chain. 

Secondly, changing the colour of edges might cause some of the fan-constructions to become out-dated, and so explicitly storing the Vizing chains might suddenly require information about a different type of extension. 
To remedy this issue, we consider \emph{stepping processes} instead, which emulate the structure of Vizing chains: in each step the process is allowed to jump to a neighbouring vertex and then follow a new bichromatic path. 
This allows us to maintain information about processes for many different kind of extensions, which in turn allows us to quickly update the data structures whenever a fan-construction becomes out-dated. 

Thirdly, there could be too many augmenting extension for the algorithm to efficiently store. 
In this case, it is not clear which chains should be stored and used for extensions, as one might store the 'wrong' chains which overlap with some stepping process much later on trying to extend via the process. Especially, storing only one chain is not necessarily sufficient.
To avoid this issue, we store sets of stepping processes that are \emph{prefix-free}, meaning that the initial parts of the processes look different, even if the later parts could be identical. This allows us to ensure that not all of the stored processes will overlap with bichromatic path that we are currently trying to extend through. 
It is of course not clear that such a set of stepping processes even exists or can be efficiently computed, but we show that not only does such a set exist for \emph{any} starting point, it is in fact sufficient to only store 'small' sets of processes for each point in order to ensure that the algorithm has sufficiently many choices to iteratively extend. 
Since small sets are sufficient, the algorithm can actually afford to go through the stored processes and check them individually in order to construct a new small set of processes to store.

By carefully designing data structures, we show that the above approach yields a dynamic algorithm with $(\Delta_{\max}+1)^{\tilde{O}(\sqrt{\log n})}$ worst-case update time.

\subsection{Outline of the paper}
In Section~\ref{sct:prelim}, we define notation and discuss some preliminaries. Then in Section~\ref{sct:tsplits}, we present our dynamic $t$-splitter, before we show how to build a hierarchy of such splitters in Section~\ref{sct:hier}. In Section~\ref{sec:lowCol}, we present our algorithm for colouring low-degree graphs. 
Finally, in Section~\ref{sct:overall} we combine all of the above to get our final algorithm. 

\section{Preliminaries \& Notation} \label{sct:prelim}
The first parts of this section is very similar to Section 2 of the full-version of~\cite{Christiansen}, as we need similar terminology. 
For a positive integer $t$, we let $[t] = \{1, 2, \dots, t \}$ be the set of numbers in $\mathbb{N}$ that are greater than $0$, but at most $t$.
We let $G = (V,E)$ be a graph on $n$ vertices and $m$ edges. A subgraph $H \subset G$ of $G$ is a graph such that $V(H) \subset V(G)$ and $E(H) \subset E(G)$. The \emph{$t$-hop neighbourhood} of $H$ in $G$, $N^{t}(H)$ consists of all of the vertices in $G$ of distance at most $t$ (in $G$) to a vertex in $H$. In particular, $H \subset N(H)$.
A \emph{proper (partial) $k$-edge-colouring} $c$ of $G$ is a function $c:E(G) \mapsto [k] \cup \{\neg\}$ such that if $c(e) \neq \neg$ and $c(e') \neq \neg$, then $e \cap e'$ implies that $c(e) \neq c(e')$. An edge $e$ such that $c(e) = \neg$ is said to be \emph{uncoloured}. Two uncoloured edges may share an endpoint in a proper partial colouring.
If no edges are uncoloured, we say that the colouring is a proper $k$-edge-colouring. 
For a vertex set $U \subset V(G)$ (resp.\ edge set $U' \subset E(G)$), we let $G[U]$ (resp.\ $G[U']$) be the induced graph of $U$ in $G$ i.e.\ the subgraph of $G$ on vertex set $V(G[U]) = U$ containing the maximum number of edges (resp.\ the induced graph of $U'$ in $G$ i.e.\ the subgraph with edge set $E(G[U']) = U'$ and vertex set $V(G[U']) = V(G)$). 
Given a \emph{proper (partial) $k$-edge-colouring} $c$ and a set of colours $S \subset [k]$, we let $G[S]$ be the graph containing all edges coloured with a colour from $S$, and all of the vertices incident to an edge coloured with a colour from $S$. Given a proper (partial) $t$-edge-colouring of a graph $G$, we define the set $A(v)$ of \emph{available} colours at a vertex $v$ to consist of precisely the colours of $[t]$ that no edge incident to $v$ has received.

For a directed graph $G$, any edge denoted by $(uv) \in E(G)$ is oriented from $u$ towards $v$. We denote by $d_{\text{in}}(v)$ (resp.\ $d_{\text{out}}(v)$) the \emph{in-degree} of $v$ (resp.\ \emph{out-degree} of $v$). 
Similarly, if $E' \subset E(G)$ is a subset of edges and $H$ any subgraph, we let $N^{j}_{E',\text{in}}(H)$ denote the $j$-hop in-neighbourhood of $H$ in the graph induced by $E'$. 
Similarly, we let $E_{E', \text{in}}(H)$ be the edge set containing all in-edges of $H$ also belonging to $E'$.

\subsection{Chains and shifts}
In order to extend a partial edge-colouring, we will have to assign some coloured edges new colours. 
In order to describe such a process formally, we briefly recall the chain and shift terminology as it was used by for example Bernshteyn~\cite{2BERNSHTEYN} (see~\cite{2BERNSHTEYN} for a more in-depth treatment). 

Let $e_1, e_2 \in E(G)$ be two adjacent edges in a graph $G$, and let $c$ be a proper partial colouring of $G$. Then we define a colouring $\operatorname{Shift}(c,e_1,e_2)$ by setting: 
\[
\operatorname{Shift}(c,e_1,e_2)(e) = \begin{cases}
     c(e_2) &\quad\text{if } e = e_1 \\
     \neg &\quad\text{if } e = e_2 \\
    c(e) &\quad\text{if } e \notin \{e_1, e_2\} \\
     \end{cases}
\]
We say that such a pair of adjacent edges $e_1, e_2$ are $c$-\emph{shiftable} if $c(e_1) = \neg$ and $c(e_2) \neq \neg$ and the colouring $\operatorname{Shift}(c,e_1,e_2)$ defined above is a proper partial colouring. If the colouring $c$ is clear from the context, we will sometimes leave out this argument. We will only consider shiftable chains. 

A \emph{chain} $C$ of size (or length) $k$ is then a set of edges $C = (e_1, \dots, e_k)$ such that $e_{i}$ and $e_{i+1}$ are adjacent for all $i$ and $c(e_i) = \neg$ if and only if $i = 1$. 
For $0\leq j \leq k-1$, we can $j$-shift such a chain by performing $\operatorname{Shift}_{j}(c,C)$ defined as:
\begin{align*}
    \operatorname{Shift}_{0}(c,C) &= c \\
    \operatorname{Shift}_{i}(c,C) &= \operatorname{Shift}(\operatorname{Shift}_{i-1}(c,C),e_{i},e_{i+1})
\end{align*}
We say that $C$ is $c$-shiftable if every pair of edges $e_{i}, e_{i+1}$ is $\operatorname{Shift}_{i-1}(c,C)$-shiftable. 
It is straightforward to check that $j$-shifting such a chain $C$ yields a proper colouring, where the unique uncoloured edge in $C$ is the edge $e_{j+1}$. 
We say that the chain ends at the edge $e_k$ and at the vertex $v$ that is shared by $e_k$ and $e_{k-1}$. 
Note that in the context of simple graphs, this vertex is well defined. 
We will refer to the act of $k$-shifting a chain of size $k$ as simply \emph{shifting} the chain. We let $c$ be the \emph{pre-shift} colouring of $C$ and $\operatorname{Shift}(c,C)$ the \emph{post-shift} colouring of $C$. 
We briefly note that from an algorithmic point of view, we can view $\operatorname{Shift}(c,e_1,e_2)(e)$ as first uncolouring $e_2$, and then assigning $e_1$ the colour $c(e_2)$. 
This ensures that during all steps of the recolouring, the current colouring of the graph is proper. 
 
Often we wish to find a chain $C$ such that shifting $C$ yields an uncoloured edge $uv$ such that $A(u) \cap A(v) \neq \emptyset$ so that we may colour $uv$ to extend our partial colouring.
If this is the case, we say $C$ is an \emph{augmenting chain}, otherwise we refer to it is a \emph{truncated chain}.  We say that $\kappa$ is \emph{available} at $uv$ if $\kappa \in A(u) \cap A(v)$.

The \emph{initial segment} of length $s$ of a chain $C$ is then the chain $C|s = (e_1, \dots, e_s)$. 
One can concatenate two chains $C_1 = (e_1, \dots, e_{s_1})$ and $C_2 = (e_{s_1},f_2, \dots, f_{s_2})$ to get the chain $C_1 + C_2 = (e_1, \dots, e_{s_1}, f_2, \dots, f_{s_2})$. 
The chain $C_1 + C_2$ is $c$-shiftable if $C_1$ is $c$-shiftable and $C_2$ is  $\operatorname{Shift}(c, C_1)$-shiftable.
In order to shift such a chain, we use $\operatorname{Shift}(c, C_1+C_2) = \operatorname{Shift}(\operatorname{Shift}(c, C_1), C_2)$.

A \emph{fan} chain with center $u$ is a chain $F$ of the form $F = (uw_1, \dots, uw_k)$. Here $u$ is the \emph{center} of the fan, denoted by $\text{center}(F)$. 
Bernshteyn~\cite{2BERNSHTEYN} showed that for such fans to be shiftable chains, we require $c(uw_{i+1}) \in A(w_i)$. 
We will let the colour of $c(uw_{i+1})$ be the \emph{representative available colour} at $w_i$. 
We will also pick a representative available colour at $w_k$ in $A(w_k)$. If this representative available colour is also either available at $u$ or if it is the representative available colour for some $w_j$ with $j < k$, then we say that the fan is \emph{maximal}. 

A \emph{path chain} is a shiftable chain of the form $P = (e_1, \dots, e_k)$, such that the edges in the chain form a path in the graph. The \emph{length} of a path chain $P = (e_1, \dots, e_k)$ is $k$.
A \emph{bichromatic path chain} is then a $c$-shiftable chain $P = (e_1, \dots, e_k)$ that forms a path in $G$ such that $c(e_1) = \neg$ and the colour of $c(e_{2i}) = \kappa_1$ for all $1\leq i \leq \floor{\tfrac{k}{2}}$ and $c(e_{2i+1}) = \kappa_2$ for all $1\leq i \leq \floor{\tfrac{k-1}{2}}$, and furthermore such that $e_2$ is either the first or the last edge in some maximal $(\kappa_1, \kappa_2)$-bichromatic path. 
In order to specify the colours present in the chain, we will sometimes refer to such a chain as a $(\kappa_1, \kappa_2)$-bichromatic path chain.

Vizing originally showed how to choose both a fan chain $F$ and a bichromatic path chain $P$ that together form a bigger chain $F+P$ such that $\operatorname{Shift}(c,F+P)$ leaves some edge $uv$ with an available colour.
This type of chain has been referred to as a \emph{Vizing chain} in the literature~\cite{2BERNSHTEYN,BernDhawan,Christiansen,grebik2020measurable}. 
\begin{definition} \label{def:truncVC}
Given a proper partial colouring $c$ of $G$ and an edge $e=uv\in E(G)$ with $c(e) = \neg$, a \emph{Vizing chain}
on $e$ with center $u$ is a chain $F+P$, where $F = (uw_1, \dots, uw_k)$ is a fan chain with $\text{center}(F) = u$ and $P = (p_0p_1, \dots, p_{t-1}p_t) $ is a $(\kappa_1,\kappa_2)$-bichromatic path chain with $p_0 = u$, $p_1 = w_k$, $\kappa_1 \in A(u)$, and $\kappa_2 \in A(w_k)$.
\end{definition}
We say that $F+P$ \emph{ends} at the edge $p_{t-1}p_{t}$ and at the vertex $p_{t-1}$. 
We let $\operatorname{Shift}(c, F+P)$ be the \emph{post-shift} colouring. 
We let $|E(P)|$ be the length of $F+P$. 
\begin{definition}
Given a proper partial colouring $c$ of $G$, an $i$-step Vizing chain is a $c$-shiftable chain of the form $F_1+P_1 + F_2 + P_2 + \dots F_i + P_i$ where $F_j+P_j$ is a Vizing chain for all $j$.   
\end{definition}
The length of a multi-step Vizing chain is the sum of the lengths of the chains $F_j+P_j$. 
We will only be considering Vizing chains that are \emph{non-overlapping}:
\begin{definition} \label{def:nonO}
An $i$-step Vizing chain is \emph{non-overlapping} if every pair of Vizing chains $F_j+P_j$ on edge $e_j$ and $F_k+P_k$ on edge $e_k$ share an edge exactly when $j = k-1$ and the shared edge is $e_k$, and if, furthermore, for all $k \leq i$ no edge in the chain $F_k+P_k = (e_1, \dots, e_s)$ is repeated i.e.\ $e_{j} = e_{h}$ iff $j = h$. 
\end{definition} 

\subsection{Algorithms for constructing and extending Vizing chains}
In this subsection, we briefly recall some well-known algorithms for constructing Vizing chains. 
In some cases, one is allowed to make arbitrary choices in the constructions.
In order to perform the constructions in a consistent manner, we will order the colours $[s]$ according to the strict ordering $<$ on $\mathbb{N}$. 
Now, whenever an arbitrary choice between colours (possibly satisfying some conditions) can be made, we will always choose the smallest colour according to $<$. 
We say that a fan is \emph{consistent} if the representative available colour $\alpha_i$ of each vertex $w_{i}$ is the smallest colour in $A(w_{i})$ subject to the requirements that $1)$ if $A(u) \cap A(w_i) \neq \emptyset$, then $\alpha_{i} \in A(u) \cap A(w_i)$ and $2)$ else if $j$ is the smallest integer $j < i$ with $\alpha_j \in A(w_i)$ then $\alpha_{i} = \alpha_{j}$. 
Given a consistent fan-chain $F  = (uw_{1}, \dots, uw_{k})$, a $(\kappa_1, \kappa_2)$-bichromatic path chain $P$ is \emph{consistent with $F$} if: $\kappa_1 \in A(u)$ is chosen as small as possible according to $<$ and $\kappa_2 = \alpha_k$.

The following algorithm will be used to construct the first fan of every Vizing chain. 
It is a standard result that we have slightly adjusted to have a consistent output. 
\begin{lemma}[\cite{2BERNSHTEYN,BernDhawan, vizing1964estimate}] \label{lma:fan1}
    Let $G$ be a graph with maximum degree $\Delta$ satisfying $\Delta \leq \Delta_{\max}$.
    Let $c$ be a proper partial $(\Delta_{\max}+1)$-edge-colouring, and let $e = uv \in E(G)$ be an uncoloured edge. 
    Then there exists a unique consistent maximal fan chain $F = (uw_{1}, \dots, uw_{k})$ with center $u$. 

    Furthermore, $F$ may be computed in $\tilde{O}(\Delta_{\max}^{3})$ time.  
\end{lemma}
\begin{proof}
    We first describe how to construct the fan: we set $w_1 = v$. We pick an available representative colour $\alpha_1$ in $A(v)$ for $v$ as follows: if $A(v) \cap A(u) \neq \emptyset$, we let $\alpha_1$ be the smallest element of $A(v) \cap A(u)$. 
    In this case, the construction terminates, and we conclude that $F$ is unique.
    Otherwise, we let $\alpha_1$ be the smallest element of $A(v)$. We let $w_2$ be the unique neighbour of $u$ with $c(uw_2) = \alpha_1$. In this case, we note that any consistent maximal fan chain must have $(uw_1,uw_2)$ as an initial segment.
    
    Next assume that we have inductively constructed a fan chain $(uw_1, uw_2, \dots, uw_{j})$ with available representative colours $\alpha_1, \dots, \alpha_{j-1}$, and that any consistent maximal fan chain on $e$ must have $(uw_1, uw_2, \dots, uw_{j})$ as an initial segment. 
    We extend the chain as follows: we pick an available representative colour $\alpha_j$ in $A(w_j)$ for $w_j$ as follows: 
    \begin{enumerate}
        \item if $A(w_j) \cap A(u) \neq \emptyset$, we let $\alpha_j$ be the smallest element of $A(w_j) \cap A(u)$. In this case, the construction terminates, and we conclude that $F$ is unique by the induction hypothesis.
        \item else if there exists a $h < i$ such that $\alpha_h \in A(w_{j})$, we let $h'$ be the smallest such $h$ and set $\alpha_j = \alpha_{h'}$. In this case, the construction terminates, and we conclude that $F$ is unique by the induction hypothesis.
        \item Otherwise, we let $\alpha_j$ be the smallest element of $A(w_j)$. We let $w_{j+1}$ be the unique neighbour of $u$ with $c(uw_{j+1}) = \alpha_j$. In this case, we note that any consistent maximal fan chain must have $(uw_1, uw_2, \dots, uw_{j+1})$ as an initial segment.
    \end{enumerate}
    The above construction must terminate. Indeed, when $j = d(u)-1$ the third case cannot occur, as $u$ has at most $d(u)-1$ coloured edges incident to it. 

    As for the furthermore part, we can check $1)$ in $\tilde{O}(\Delta_{\max})$ by just checking for each colour in $A(w_{j})$ if it belongs to $A(u)$, we can check $2)$ in $O(\Delta_{\max}^2)$ time, since we can check if $\alpha_k \in A(w_j)$ in $O(\Delta_{\max})$ time. Finally, we can do $3)$ in $\tilde{O}(\Delta_{\max})$ time by simply computing $A(w_j)$ and extracting the minimum. 
\end{proof}
Given a consistent maximal fan chain $F = (uw_{1}, \dots, uw_{k})$ and a parameter $t \in \mathbb{Z}_{\geq 0}$, we can extend $F$ to a Vizing chain $(F'+P)_t$, or simply $F'+P$ if $t$ is clear from the context, in a consistent manner (note $F'$ might not be equal to $F$), as follows:
if $\alpha_{k} \in A(u)$, we set $P = (uw_{k})$ and $F' = F$.
Otherwise, we let $\kappa_1$ be the smallest colour in $A(u)$ and $\kappa_2 = \alpha_{k}$. 
Let $P' = p_1,p_2, \dots, p_{t'}$ be the unique maximal $(\kappa_1, \kappa_2)$-bichromatic path with $p_1 = w_k$. 
If $t < t'-1$, we let $F' = F$, and we let $P = (uw_k,w_kp_2, \dots p_{t}p_{t+1})$ (except for the case $t = t'-2$ in Case 2 below, where we do something else, as explained later). 
Otherwise, if $t \geq t'-1$, there are $3$ cases:  
\paragraph{Case 1:} if $p_{t'} = w_{i}$ with $i < k$, we let $F'  = (uw_1, \dots, uw_i)$ and we let $P = (uw_i, w_{i}p_{t'-1}, p_{t'-1}p_{t'-2}, \dots, p_{2}p_{1})$. Note that in this case $F'$ is not a maximal fan.
\paragraph{Case 2:} if $p_{t'-1} = w_{i+1}$ and $p_{t'} = u$ with $i +1 < k$, we let $F' = (uw_1, \dots, uw_{i+1})$ and we let $P = (uw_{i+1},w_{i+1}p_{t'-2},p_{t'-2}p_{t'-3}, \dots, p_{2}p_{1})$. Note that in this case $F'$ is not a maximal fan. 
In the case where $t = t'-2$, we pick $F' = (uw_1, \dots, uw_{i+1})$ and we let $P = (uw_{i+1},w_{i+1}p_{t'-2},p_{t'-2}p_{t'-3}, \dots, p_{3}p_{2})$
\paragraph{Case 3:} if $p_{t'} \notin \{w_1, \dots, w_{k}\}$, we let $F' = F$ and we let $P = (uw_k,w_{k}p_{2},p_{2}p_{3}, \dots, p_{t'-1}p_{t'})$. \\\\
\noindent 
It is straight-forward to check that we can construct $F'+P$ in $\tilde{O}(\Delta_{\max}^3 + \Delta_{\max}\cdot{}t)$ time, and that the construction is unique.
If $e = uv$ is any uncoloured edge and $F$ is the unique maximal fan chain guaranteed by Lemma~\ref{lma:fan1} centered on $u$, then for any given $t$, we say that $(F'+P)_t$ is a consistent Vizing chain with center $u$. Note that $(F'+P)_t$ is unique. 

Next we will see a different algorithm needed to extend any $i$-step Vizing chain to an $(i+1)$-step Vizing chain in a consistent way. 
Again we will first construct a fan chain, and then explain how to compute the accompanying bichromatic path chain. 
We will use the so-called 'second fan lemma'~\cite{2BERNSHTEYN, grebik2020measurable}. 
\begin{lemma}[\cite{2BERNSHTEYN,grebik2020measurable}] \label{lma:fan2}
 Let $G$ be a graph with maximum degree $\Delta$ satisfying $\Delta \leq \Delta_{\max}$.
 Let $c$ be a proper partial $(\Delta_{\max}+1)$-edge-colouring, and let $e = uv \in E(G)$ be an uncoloured edge. 
 Let $\kappa_{1} \in [\Delta_{\max}+1]$ be an available colour at $u$ and let $\kappa_{2} \in [\Delta_{\max}+1]$ be an available colour at $v$. Then there exists a fan chain $uw_{1}, \dots, uw_{k}$ with $w_1 = v$ such that either 
\begin{enumerate}
    \item the representative available colour at $w_{k}$ is also available at $u$. 
    \item $\kappa_2 \notin A(u)$ and the fan is maximal subject to the constraint that the colours $\kappa_1$ or $\kappa_2$ are considered unavailable to every vertex.
    \item $w_{k} \neq v$, the representative available colour at $w_{k}$ is $\kappa_2$ and no edge $uw_{i}$ is coloured $\kappa_2$.
\end{enumerate}
If the representative available colours are chosen as small as possible subject to first $1)$, then $2)$, and finally $3)$ above, then $F$ is unique.

Furthermore, we can compute $F$ in $\tilde{O}(\Delta_{\max}^{3})$ time.
\end{lemma}
\begin{proof}
We construct a maximal fan almost exactly as in the proof of Lemma~\ref{lma:fan1}. Again we begin by setting $w_1 = v$.
We pick an available representative colour $\alpha_1$ in $A(v)$ for $v$ as follows: if $A(v) \cap A(u) \neq \emptyset$, we let $\alpha_1$ be the smallest element of $A(v) \cap A(u)$. 
In this case, the construction terminates, and we conclude that $F$ is unique.
Otherwise, we let $\alpha_1$ be the smallest element of $A(v)\setminus \{\kappa_2\}$. Note that since $uv$ is uncoloured, $|A(v)| \geq 2$. We let $w_2$ be the unique neighbour of $u$ with $c(uw_2) = \alpha_1$. In this case, we note that any consistent maximal fan chain must have $(uw_1,uw_2)$ as an initial segment.
    
Next assume that we have inductively constructed a fan chain $(uw_1, uw_2, \dots, uw_{j})$ with available representative colours $\alpha_1, \dots, \alpha_{j-1}$, and that any consistent maximal fan chain on $e$ must have $(uw_1, uw_2, \dots, uw_{j})$ as an initial segment. 
We extend the chain as follows: we pick an available representative colour $\alpha_j$ in $A(w_j)$ for $w_j$ as follows: 
    \begin{enumerate}
        \item if $A(w_j) \cap A(u) \neq \emptyset$, we let $\alpha_j$ be the smallest element of $A(w_j) \cap A(u)$. In this case, the construction terminates, and we conclude that $F$ is unique by the induction hypothesis.
        \item else if there exists a $h < i$ such that $\alpha_h \in A(w_{j})$, we let $h'$ be the smallest such $h$ and set $\alpha_j = \alpha_{h'}$. In this case, the construction terminates, and we conclude that $F$ is unique by the induction hypothesis.
        \item else if $|A(w_j)\setminus \{\kappa_2\}| \geq 1$, we let $\alpha_j$ be the smallest element of $A(w_j)\setminus \{\kappa_2\}$. We let $w_{j+1}$ be the unique neighbour of $u$ with $c(uw_{j+1}) = \alpha_j$. In this case, we note that any consistent maximal fan chain must have $(uw_1, uw_2, \dots, uw_{j+1})$ as an initial segment.
        \item Otherwise, we set $\alpha_j = \kappa_2$ and conclude the construction.
    \end{enumerate}
    The above construction must terminate. Indeed, when $j = d(u)-1$ the third case cannot occur, as $u$ has at most $d(u)-1$ coloured edges incident to it. 

    Similarly to Lemma~\ref{lma:fan1}, we can construct the fan in $\tilde{O}(\Delta_{\max}^{3})$ time.
\end{proof}
Now given some $i$-step Vizing chain $F_1+P_1+F_2+P_2+ \dots F_i + P_i$ and some parameter $t \in \mathbb{Z}_{\geq 0}$, we can extend it in a consistent manner by applying Lemma~\ref{lma:fan2} to the colouring $\operatorname{Shift}(c, F_1+P_1+F_2+P_2+ \dots F_i + P_i)$ and the final edge $p_{t_{i}-1}p_{t_{i}} = u'v'$ to get a unique fan $F = (u'w_1', u'w_2', \dots, u'w_k')$. 
Let $P_{i}$ be $(\kappa_1,\kappa_2)$-bichromatic with $c(u'v') = \kappa_{2}$, then we apply Lemma~\ref{lma:fan2} with $\kappa_1$ and $\kappa_2$.  
If we are in Case $1.$, we set $F_{i+1} = F$ and $P_{i+1} = u'w_k'$. 
If we are in Case $2.$, we let $\tau_1 \in A(u')\setminus \{\kappa_1, \kappa_2\}$ be the smallest available colour and 
$\tau_2 = \alpha_{k}$. 
Note that $|A(u')| \geq 2$ in the shifted colouring, and that $\kappa_2 \notin A(u)$ by assumption.
Then we let
$P' = p'_1p'_2\dots p'_{t'}$ be the maximal $(\tau_1,\tau_2)$-bichromatic path beginning at $p'_1 = w_k'$.

If $p'_t = u$ and $p'_{t-1}  = w_s'$ for some $s$, there are three sub-cases. 
If $t < t'-2$, we set $F_{i+1} = F$ and $P_{i+1} = (u'w_k',w_kp'_2, \dots p_{t-1}p_{t})$. Else if $t = t'-1$, we set $F = (uw_1', \dots, uw_s')$ and $P_{i+1} = (u'w_s',w_s'p'_{t'-2}, \dots p_{3}p_{2})$.  
If none of the above hold, it must be the case that $t > t'-1$, and then we set $F = (uw_1', \dots, uw_s')$ and $P_{i+1} = (u'w_s',w_s'p'_{t'-2}, \dots p_{2}p_{1})$. 

If $p'_t =  w_s'$ for some $s$, there are two sub-cases. 
If $t < t'-1$, we set $F_{i+1} = F$ and $P_{i+1} = (u'w_k',w_kp'_2, \dots p_{t-1}p_{t})$. 
Else, we set $F = (uw_1', \dots, uw_s')$ and $P_{i+1} = (u'w_s',w_s'p'_{t'-1}, \dots p_{2}p_{1})$.

Else $p'_t$ is disjoint from $F-\{w'_k\}$, in which case there are $2$ sub-cases: 
If $t < t'-1$, we set $F_{i+1} = F$ and $P_{i+1} = (u'w_k',w_kp'_2, \dots p_{t-1}p_{t})$. 
Else, we set $F_{i+1} = F$ and $P_{i+1} = (u'w_k',w_kp'_2, \dots p_{t'-1}p_{t'})$.

Otherwise, we are in Case $3)$. Again we set $F_{i+1} = F$, but this time we let $P' = p'_1p'_2\dots p'_{t'}$ be the maximal $(\kappa_1,\kappa_2)$-bichromatic path beginning at $p'_1 = w_k'$.
If $t < t'-1$, we set $P_{i+1} = (u'w_k',w'_kp_2, \dots p_{t}p_{t+1})$, otherwise we set $P_{i+1} = (u'w_k',w_kp'_2, \dots p_{t'-1}p_{t'})$. 

Similarly to before, we can perform the above extension in $\tilde{O}(\Delta_{\max}^3 + \Delta_{\max}\cdot{} t)$ time. 
We will say that we extend the Vizing chain via Lemma~\ref{lma:fan2} in a \emph{consistent} manner. 
Note that if every extension in a Vizing chain $F_1+P_1+F_2+P_2+ \dots F_i + P_i$ is consistent, then $P_{j+1}$ will never overlap $P_j$ in any edges for all $j$ (except possibly in the first edge). 
Indeed, the colours of $P_{j+1}$ are either disjoint from those of $P_{j}$ or identical to those of $P_{j}$. In the latter case, $P_{j+1}$ and $P_{j}$ belong to different components in the shifted colouring.

In this paper, we will only be working with Vizing chains satisfying certain restrictions. The restrictions we will impose are as follows: 
\begin{definition} \label{def:cVc}
    Let $VC = F_1 + P_1 + \dots + F_{j}+P_{j}$ be a multi-step Vizing chain. We say that $VC$ is \emph{consistent} if the following conditions hold:
    \begin{itemize}
        \item $F_1 + P_1$ is constructed in a consistent manner (Lemma~\ref{lma:fan1} and following discussion).
        \item For $i\geq 2$ we have that $F_i+P_i$ is constructed in a consistent manner (Lemma~\ref{lma:fan2} and following discussion).
        \item $\text{center}(F_j)$ is disjoint from the $2$-hop neighbourhood of $\text{center}(F_1) + P_1 + \dots + \text{center}(F_{j-2}) + P_{j-2} + \text{center}(F_{j-1})$ for all $j$.
        \item For all $j$, we have that $P_{j}$ is disjoint from the $2$-hop neighbourhood of $\text{center}(F_{k})$ for all $k \leq j-1$ and from $P_{k}$ for all $k \leq j-2$.
        \item For all $j$, we have that the maximal bichromatic path containing $P_j$ under $c$ is different from the maximal bichromatic path containing $P_{j+1}$ under $c$.
    \end{itemize}
\end{definition}

\subsection{Path-stepping processes}
We will now consider a looser process than the one used to construct Vizing chains. Beginning from some vertex $w$, we let one step be the process of first moving to a neighbour of $w$, say $y$, and then following some bichromatic path $P$ having $y$ as an endpoint. As with the multi-step Vizing chains, we are allowed to truncate the bichromatic path at any point to end up at some new vertex $w'$. 
We may then iteratively repeat this process at $w'$. 
If we do so $i$ times, we say that we run an $i$-step process starting from $w_1$. 
We denote by $w_1+y_1+\hat{P}_1+w_2+y_2+\hat{P}_2 + \dots w_{i}+ y_{i}+\hat{P}_{i}$ the $i$-step process beginning at $w_1$, that first jumps to the neighbour $y_1$ of $w_1$, then follows a truncated maximal bichromatic path $P_1$ to end up at $w_2$, and then jumps to the neighbour $y_2$ of $w_2$ and follows a truncated maximal bichromatic path $P_2$ and so on and so forth. 
We will say that the process is \emph{non-overlapping} if for all $k, j$ we have that $V(P_{k}) \cap V(P_{j}) \neq \emptyset$ implies that $|k-j| \leq 1$, if $E(P_{k}) \cap E(P_{j}) \neq \emptyset$ implies $k = j$, if $V(P_{k}) \cap N^{2}(w_j) \neq \emptyset$ for $k\geq j$ implies that $k=j$, and, finally, if the maximal bichromatic path $Q_j$ containing $P_j$ is equal to the maximal bichromatic path $Q_k$ containing $P_k$, then $j = k$. 
In other words, we will only allow a bichromatic path to intersect the 2-hop-neighbourhood of the immediately preceding path, and we will not allow it to intersect the edge set of any previous paths. 
If $|V(P_{j})| \leq \ell$ for all $j$, we say that the process has step length at most $\ell$.
If $P_{i}$ is a maximal bichromatic path, we say that the $i$-step process is \emph{augmenting}.

Maintaining information about this type of process is useful, as it allows us to accurately represent where consistent Vizing chains might go, while at the same time avoiding the need to worry about exactly which extensions are possible at what times. 
In particular, we observe that any consistent and non-overlapping $i$-step Vizing chain $F_1+P_1 + \dots + F_{i} + P_{i}$ where $F_1$ is centered at $w$ naturally induces an $i$-step process on $w$, by letting $w_j = \text{center}(F_j)$, letting $y_j$ be the second point of $P_{j}$, and letting $\hat{P}_{j}$ be the the bichromatic sub-path of the path $P_j$ excluding only the first and the last edge of the chain $P_j$. 
To make this observation more precise, we will define and consider \emph{semi-consistent} processes. 
Here, we will put some further restrictions on the processes so that they model consistent Vizing chains exactly. 

Given an $i$-step process $w_1+y_1+P_1+w_2+y_2+P_2 + \dots w_{i}+ y_{i}+P_{i}$, we say that the extension $w_j + y_j + P_j$ for $j\geq 2$ is \emph{consistent} with Lemma~\ref{lma:fan2}, if the following is the case: consider the colouring $c'$ obtained from $c$ as follows: suppose $P_{j-1} = p_1p_2\dots p_t$ is $(\kappa_1, \kappa_2)$-bichromatic with $c(p_{t-1}p_{t}) = \kappa_1$. 
Suppose $P = p_1p_2\dots p_t p_{t+1} \dots p_{s}$ is a maximal bichromatic path containing $P_{j-1}$. Note that if $s = t$, then $P_{j-1}$ is augmenting, and it cannot be extended consistently.

We obtain $c'$ by first un-colouring the edge $p_{t}p_{t+1}$ i.e.\ $c'(p_{t}p_{t+1}) = \neg$, and then for all $k < t-1$ setting $c'(p_k p_{k+1}) = \kappa_2$ if $c(p_k p_{k+1}) = \kappa_1$ and vice versa setting $c'(p_k p_{k+1}) = \kappa_1$ if $c(p_k p_{k+1}) = \kappa_2$. 
It is straight-forward to check that if $c$ is a proper $(\Delta_{\max}+1)$-edge-colouring, then so is $c'$. 
Now, apply Lemma~\ref{lma:fan2} to $p_{t}p_{t+1}$ with the colouring $c'$, $u = p_{t}$, $v = p_{t+1}$, and the colours $\kappa_1$ and $\kappa_2$ as before. 
By fixing a length $t' = |P_{i}|$, the lemma together with the discussion that follows it produces a unique fan chain $F$ and a bichromatic-path chain $Q = (q_0q_1, \dots, q_{\hat{t}}q_{\hat{t}+1})$ with $\hat{t} \leq t'$ such that $F+Q$ is a $1$-step Vizing chain. 
We say the extension $w_j + y_j + P_j$ for $j\geq 2$ is \emph{consistent} with Lemma~\ref{lma:fan2} if $w_j = \text{center}(F_{j}) = q_0$, and if $y_j = q_1$, and finally $P_{j} = q_1,q_2, \dots, q_{\hat{t}+1}$ if $\hat{t}+1 \leq t'$ and $P_{j} = q_1,q_2, \dots, q_{\hat{t}}$ otherwise. Note that in the former case $P_j$ is augmenting, and that in all cases the extension is unique. 

\begin{remark}
    It is important to remark that in the case where $P_j$ is augmenting, we must have that $i = j$, and perhaps more importantly, that it could be the case that $F_j$ is \emph{not} a maximal fan.
\end{remark}

Next, we define a semi-consistent process: 
\begin{definition} \label{def:sc}
    Let $C = w_1+y_1+P_1+w_2+y_2+P_2 + \dots w_{i}+ y_{i}+P_{i}$ be a non-overlapping $i$-step process beginning at $w_1$. We say that $C$ is \emph{semi-consistent} if the following holds: 
    \begin{itemize}
        \item For $j\geq 2$ we have that $w_j, y_j$ and $P_{j}$ is an extension consistent with Lemma~\ref{lma:fan2}.
        \item We have that $w_{j}$ is disjoint from the $2$-hop neighbourhood of $w_1 + y_1 + P_1 + \dots w_{j-2} + y_{j-2} + P_{j-2} + w_{j-1}$ for all $j$. 
        \item For all $j$, we have that $P_{j}$ is disjoint from the $2$-hop neighbourhood of $w_k$ for all $k \leq j-1$ and from $P_{k}$ for all $k \leq j-2$.
    \end{itemize}
\end{definition}
We observe that for any $j < i$ the length of $P_j$ is at least $3$. Indeed, otherwise Item 2 in Definition~\ref{def:sc} fails.
Semi-consistent processes are in correspondence with consistent Vizing chains in the following sense: 
\begin{lemma} \label{lma:mapping}
    There is a mapping $\gamma$ from the set of consistent $i$-step Vizing chains to the set of semi-consistent $i$-step processes. 
    In particular, if a consistent $i$-step Vizing chain is augmenting, so is its image, and similarly if a semi-consistent $i$-step process is augmenting, so is every Vizing chain in its pre-image. 
\end{lemma}
\begin{proof}
    Let $VC = F_1 + Q_1 + \dots + F_{i}+Q_{i}$ be any consistent $i$-step Vizing chain. Then we will let $VC = F_1 + Q_1 + \dots + F_{i}+Q_{i}$ map to the $i$-step process $C = w_1 + y_1 + P_1 + \dots + w_{i} + y_{i}+P_{i}$, where for all $j$ we have $w_j = \text{center}(F_j)$, $y_j = q_1$, and $P_j = q_1, q_2 \dots, q_{t_j}$ with $Q_j = (q_0q_1, \dots, q_{t_j+1})$. 
    If $VC$ is \emph{not} augmenting, we let $w_i = \text{center}(F_i)$, $y_i = q_1$, and $P_i = q_1, q_2 \dots, q_{t_i}$, otherwise we set $w_i = \text{center}(F_i)$, $y_i = q_1$, and $P_i = q_1, q_2 \dots, q_{t_i+1}$. 

    Next we show that $C$ is semi-consistent. First of all, we need to check that $C$ is non-overlapping. This follows from the fact that $VC$ is consistent, which puts precisely the same constraints on the $P_j$'s in Definition~\ref{def:cVc} as is required for $C$ to be non-overlapping. 
    The first item of Definition~\ref{def:sc} is satisfied since the second item of Definition~\ref{def:cVc} holds.
    This is the case since no edge outside of $P_{j-1}$ in the two-hop neighbourhood of the center of $F_j$ has changed colour before Lemma~\ref{lma:fan2} is applied. This means that applying it directly to $P_{j-1}$ yields the same fan and path choice as using it to extend the Vizing chain. 
    Finally, items 3 and 4 of Definition~\ref{def:cVc} are equivalent to the last two items of Definition~\ref{def:sc}.

    Furthermore, observe that solely the length of the maximal path containing $P_{i}$ determines whether $C$ and $VC$ are augmenting and so $C$ is augmenting if and only if $VC$ is. 

    The other direction is symmetrical.
\end{proof}
If a semi-consistent process is the image of a consistent $i$-step Vizing chain, we say that the process is \emph{consistent}. 
Finally, we note that under the right conditions, we can use semi-consistent processes to build consistent Vizing chains: 
\begin{lemma} \label{lma:processToChain}
    Let $G$ be a graph with maximum degree $\Delta\leq \Delta_{\max}$, and let $c$ be a proper $(\Delta_{\max}+1)$-edge-colouring of $G$. 
    Let $C = w_1 + y_1 + P_1 + \dots + w_{i} + y_{i}+P_{i}$ be an augmenting semi-consistent $i$-step process with step-length $\leq \ell$. 
    Let $e = w_1 v$ be an un-coloured edge. Suppose that applying Lemma~\ref{lma:fan1} with $t = \text{length}(P_1)+1$ and the following discussion yields a consistent $1$-step Vizing chain $(F+Q_1)_{t}$ such that
    \begin{itemize}
        \item $\text{center}(F) = w_1$.
        \item $y_1 = x_k$ with $F = (w_1x_1, w_1 x_2, \dots w_1 x_k)$.
        \item $Q = (q_0q_1,q_1q_2, \dots,q_{t-1}q_{t})$ where $P_{1} = q_1q_2\dots q_{t-1}$.
    \end{itemize}
    Then there is an augmenting consistent $i$-step Vizing chain on $e$ with the same step length as $C$.
\end{lemma}
\begin{proof}
    We will construct the augmenting consistent Vizing chain on $e$, say $VC = F_1+Q_1+F_2+Q_2+\dots+F_i + Q_i$ as follows:
    we set $F_1 = F$ and $Q_1 = Q$. 
    Then we extend $F_1+Q_1+F_2+Q_2+\dots+F_j + Q_j$ by applying Lemma~\ref{lma:fan2} and the following discussion making sure that the length of $F_{j+1} + Q_{j}$ is equal to $t_j + 1$ with $t_j = |P_{j}| + 1$. 

    We now observe that $\gamma(VC) = C$. Indeed, by construction every extension results in $Q_{j} = (q_{0,j}q_{1,j}, \dots, q_{t_j+1,j}q_{t_j,j})$ such that $P_{j} = q_{1,j}q_{2,j}\dots q_{t_j-1,j}$. 
    This follows from the fact that since $C$ is semi-consistent, the extension performed on $P_j$ uses a colouring $c'$ which only differs from the colouring $\operatorname{Shift}(c,F_1+Q_1+F_2+Q_2+\dots+F_j + Q_j)$ outside of the $2$-hop neighbourhood of $w_j$. 
    Indeed, $c'$ and $\operatorname{Shift}(c,F_1+Q_1+F_2+Q_2+\dots+F_j + Q_j)$ agree on the colours along $P_j$ by construction, and $w_{j+1}$ is disjoint from the two-hop neighbourhood of $C = w_1 + y_1 + P_1 + \dots + w_{j-1} + y_{j-1}+P_{j-1} + w_j$ also by construction. 
    Since $C$ is semi-consistent, the extension ensures that $Q_{j} = (q_{0,j}q_{1,j}, \dots, q_{t_j+1,j}q_{t_j,j})$ where $P_{j} = q_{1,j}q_{2,j}\dots q_{t_j-1,j}$. 

    It now follows that $\gamma(VC) = C$, and therefore Lemma~\ref{lma:mapping} now guarantees that $VC$ is also augmenting. 
\end{proof}

\section{Dynamic t-splitters} \label{sct:tsplits}

Given a graph $G$ equipped with any (not necessarily proper) edge-colouring $c$, we let $d_{i}(u)$ be the degree of $u$ in the graph induced in $G$ by all edges coloured with the colour $i$. We let $\tilde{d} = \max_{\kappa} \max_{v} d_{\kappa}(v)$.
Given two disjoint subsets of vertices $A,B \subset V(G)$, we let $G[A,B]$, be the bipartite subgraph of $G$ on vertex set $A \cup B$ containing exactly the edges going from a vertex in $A$ to a vertex in $B$. 
We let $s_{i}(v) = \max \{d_{i}(v)-\frac{\Delta_{\max}}{t},0\}$ be the \emph{surplus} of $i$-coloured edges around $v$.
Finally, we let $L: E(G) \rightarrow \mathbb{Z}$ be a function that maps $uv \mapsto s_{c(uv)}(u) + s_{c(uv)}(v)$, and we set $\max \limits_{e \in E} L(e) = \tilde{L}$. 

In this section, we will show Theorem~\ref{thm:tsplit}, which we restate below for convenience:
\begin{theorem}[Identical to Theorem~\ref{thm:tsplit}]
   Let $G$ be a dynamic graph, and suppose $\eps > 0$, $t \in \mathbb{N}$, and an upper bound $\Delta_{\max}$ on the maximum degree of $G$ throughout the entire update sequence is given. 
    Assume that $\tfrac{\Delta_{\max}}{t} \geq \tfrac{10^4 \log^2 n}{\eps^{2}}$. Then there is a dynamic algorithm that computes a $t$-splitter of $G$ with an amortised update-time in $O(t^3\eps^{-4} \log^3 m)$ and an amortised recourse in $O(\tfrac{\log n}{\eps^2})$.
\end{theorem}
In this entire section, we will work with two parameters $\Delta_{\max}$, which we assume is a known upper-bound on the maximum degree of some dynamic graph $G$, as well as a parameter $\eta \geq 16$. 
Note that it is important that $\eta \geq 1$, but we will only apply the theory below, when this is the case.  We will analyse a colouring $c$ satisfying the following invariant: 
\begin{invariant} \label{inv:local}
    For all $uv \in E(G)$ and all $i \in [t]$, we have that 
    \[
    L(uv) = s_{c(uv)}(u) + s_{c(uv)}(v) \leq \eta + s_{i}(u) + s_{i}(v)
    \]
\end{invariant}
We first show that if $1 \leq \eta \leq \tfrac{\eps^2}{128\log n}\tfrac{\Delta_{\max}}{t}$, then $\tilde{d} \leq (1+\eps) \frac{\Delta_{\max}}{t}$. 
We note already now that we will only apply the above Theorem, when $\tfrac{\Delta_{\max}}{t} = \Omega(\poly(\eps^{-1},\log n))$, so the assumption that $\eta \geq 16$ can always be realised when necessary. 

\subsection{Analysing the Local Condition}
We first show the following structural theorem:
\begin{theorem} \label{thm:invGaran}
    Let $0 < \eps \leq \frac{1}{2}$ and let $G$ be equipped with a $t$-colouring $c$ satisfying Invariant~\ref{inv:local}. Let $1 \leq \eta \leq \tfrac{\eps^2}{128\log n}\tfrac{\Delta_{\max}}{t}$. Then it holds that $\tilde{d} \leq (1+\eps) \frac{\Delta_{\max}}{t}$.
\end{theorem}
Before we show this theorem, we first show some helpful lemmas. 
The first lemma shows that if an edge $e$ has a large surplus, then any endpoint of $e$ has a large surplus in some colour. 
\begin{lemma} \label{lma:bigN}
    Suppose Invariant~\ref{inv:local} holds. Suppose furthermore that $e = uv$ has $L(uv) = D$. Then there exists a colour $\kappa$ such that $s_{\kappa}(u) \geq D - \eta$.
\end{lemma}
\begin{proof}
    Consider $v$. Since $d(v) \leq \Delta_{\max}$, it follows by the pigeon-hole principle that there exists some colour $\kappa'$ for which $d_{\kappa'}(v) \leq \frac{\Delta_{\max}}{t}$. 
    Since we assumed Invariant~\ref{inv:local} holds, we have that 
    \[
    D \leq \eta + s_{\kappa'}(u) + s_{\kappa'}(v) = \eta + s_{\kappa'}(u)
    \]
    Hence, we deduce that $D-\eta \leq s_{\kappa'}(u)$, and we are done by setting $\kappa = \kappa'$. 
\end{proof}
We let $\tilde{L}:=\max_{e \in E(G)} L(e)$. The next lemma shows how to inductively determine a large set of bipartite subgraphs containing vertices with large surpluses. 
\begin{lemma} \label{lma:largeSurplus}
    Suppose Invariant~\ref{inv:local} holds, and let $\mu \in \mathbb{Z}$ be a parameter. Then there exists two colours $\kappa_x$ and $\kappa_y$, as well as two family of sets $(X_i)_{i = 1}^{\mu}$ and $(Y_i)_{i = 1}^{\mu}$ such that for all $i$, we have $X_i, Y_i \subset V(G)$ and that the following two conditions hold:
    \begin{align}
        \forall u \in X_i: s_{\kappa_{x}}(u) - s_{\kappa_{y}}(u) &\geq \tilde{L}-(1+2(i-1))\eta \\
        \forall v \in Y_i: s_{\kappa_{y}}(v) - s_{\kappa_{x}}(v) &\geq \tilde{L}-2i \eta 
    \end{align}
\end{lemma}
\begin{proof}
    We first determine $\kappa_x$ and $\kappa_y$. Then we construct the sets inductively. Pick an arbitrary edge $uv$ achieving $L(uv) = \max L(e) = \tilde{L}$. 
    Now we may apply Lemma~\ref{lma:bigN} to see that there exists a colour $\kappa_x$ such that $d_{\kappa_x}(u) \geq \frac{\Delta_{\max}}{t}+\tilde{L}-\eta$. By the pigeon-hole principle there is some colour $\kappa_y$ incident to $u$ such that $d_{\kappa_y}(u) \leq \frac{\Delta_{\max}}{t}$.

    With $\kappa_x$ and $\kappa_y$ defined, we construct the family of sets inductively. 
    Initially, we set $X_1 = \{u\}$ and $Y_1 = N_{G[\kappa_x]}(u) - \{u\}$.
    We confirm that the required conditions hold: $(1)$ is clear from our choice of $\kappa_x$ and $\kappa_y$.
    We get $(2)$ by observing that by Invariant~\ref{inv:local}, we find that for any $v \in Y_1$, we have: 
    \[
        s_{\kappa_y}(v) + s_{\kappa_{y}}(u) + \eta \geq s_{\kappa_x}(v) + s_{\kappa_x}(u) 
    \]
    Hence: 
    \begin{align*}
        s_{\kappa_y}(v) - s_{\kappa_{x}}(v) &\geq s_{\kappa_x}(u) - s_{\kappa_y}(u) - \eta \\
        &\geq \tilde{L}- 2\eta
    \end{align*}
    by above.

    For the inductive step, we now set $X_i = N_{G[\kappa_y]}(Y_{i-1}) - Y_{i-1}$, and $Y_i = N_{G[\kappa_x]}(X_i) - X_i$. Again we confirm the two conditions: 
    \begin{itemize}
        \item $(1)$: Consider $u \in X_{i}$. We have $\forall v \in Y_{i-1}$ that $s_{\kappa_{y}}(v) - s_{\kappa_{x}}(v) \geq \tilde{L}-2(i-1) \eta$ by the induction hypothesis. By applying Invariant~\ref{inv:local}, we find that 
        \[
        s_{\kappa_x}(v) + s_{\kappa_{x}}(u) + \eta \geq s_{\kappa_y}(v) + s_{\kappa_y}(u) 
        \]
        since $c(uv) = \kappa_y$. Hence: 
        \begin{align*}
            s_{\kappa_x}(u) - s_{\kappa_{y}}(u) &\geq s_{\kappa_y}(v) - s_{\kappa_x}(v) - \eta \\
            &\geq \tilde{L}- (1+2(i-1))\eta
        \end{align*}
         as before. 
        \item $(2)$ is  symmetric: consider any $v\in Y_i$. We have $\forall u \in X_{i}$ that $s_{\kappa_{x}}(u) - s_{\kappa_{y}}(u) \geq \tilde{L}-(1+2(i-1)) \eta$ by induction. By applying Invariant~\ref{inv:local}, we find that 
        \[
        s_{\kappa_y}(v) + s_{\kappa_{y}}(u) + \eta \geq s_{\kappa_x}(v) + s_{\kappa_x}(u) 
        \]
        since $c(uv) = \kappa_x$. Hence: 
        \begin{align*}
            s_{\kappa_y}(v) - s_{\kappa_{x}}(v) &\geq s_{\kappa_x}(u) - s_{\kappa_y}(u) - \eta \\
            &\geq \tilde{L}- 2i\eta
        \end{align*}
         as before. Note that the above argument can be repeated ad infinitum, and so in particular also $\mu$ times.
    \end{itemize}

\end{proof}
We are now ready to show Theorem~\ref{thm:invGaran}:
\begin{proof}[Proof of Theorem~\ref{thm:invGaran}]
    Suppose for contradiction that $\tilde{d} > (1+\eps) \frac{\Delta_{\max}}{t}$. We will show that this implies the existence of a family of bipartite subgraphs $\{\mathbb{B}_j\}_{j = 1}^{10\log(n)\eps^{-1}} = \{(P_j, Q_j)\}_{j = 1}^{10\log(n)\eps^{-1}}$ of $G$ satisfying that $|Q_j| \geq (1+\frac{\eps}{4})^{j}$.
    
    This implies that for some $j \leq 10\log(n)\eps^{-1}$, we have that
    \[
    |Q_j| \geq (1+\frac{\eps}{4})^{10\log(n)\eps^{-1}} \geq e^{\frac{10}{8} \log n} > n
    \]
    where we used the standard in-equality $1-x \leq e^{-x}$ which holds for all real $x$. In particular, it follows that for all $0 \leq x \leq 1/2$, we have that $e^{x} \leq \frac{1}{1-x} \leq 1+\frac{x}{1-x} \leq 1+2x$.
    This contradicts the fact that we also must have $\mathbb{B}_j \subset G$ and so $|V(\mathbb{B}_j)| \leq |V(G)| = n$. 

    In order to construct $\{\mathbb{B}_j\}_{j = 1}^{10\log(n)\eps^{-1}}$, we invoke Lemma~\ref{lma:largeSurplus} with $\mu = \ceil{10\log(n)\eps^{-1}}$ to get sets $(X_i)_{i = 1}^{\mu}$ and $(Y_i)_{i = 1}^{\mu}$ and colours $\kappa_x$ and $\kappa_y$ satisfying the condition of the lemma. Then for appropriate values of $i$, we set 
    \[
    \mathbb{B}_{2i-1} = (G[\kappa_{x}])[X_{i},Y_{i}]
    \]
    with $P_{2i-1} = X_i$ and  $Q_{2i-1} = Y_i$  as well as 
    \[
    \mathbb{B}_{2i} = (G[\kappa_{y}])[Y_{i},X_{i+1}]
    \]
    with $P_{2i} = Y_i$ and  $Q_{2i} = X_{i+1}$
    These graphs are clearly bipartite by construction, so all we need to show is that $Q_i$ must expand. We show this by induction on $i$. 
    The base case is is as follows: since we have tacitly assumed that $\eta \geq 1$ it follows that $\tilde{L} - \eta$ is large, say at least $2$. In particular, we have that $|Q_{1}| = |Y_1| \geq \tilde{L} - \eta \geq 2$.

    Instead of showing that $Q_2$ is also sufficiently large, we proceed to the induction step directly. We will first show that if $Q_{2i-1} \geq (1+\frac{\eps}{4})^{2i-1}$, then $Q_{2i} \geq (1+\frac{\eps}{4})^{2i}$. 
    Then we show that if $Q_{2i} \geq (1+\frac{\eps}{4})^{2i}$ , then $Q_{2i+1} \geq (1+\frac{\eps}{4})^{2i+1}$. 
    In this way, we only need to show the base case $|Q_{1}| = |Y_1| \geq \tilde{L} - \eta \geq 2$, as we have already done.

    So for the inductive step, we first consider $\mathbb{B}_{2i}$. We note that, by construction, we have that $\forall u \in X_{i+1}$
    \[
        s_{\kappa_{x}}(u) - s_{\kappa_{y}}(u) \geq \tilde{L}-(1+2i)\eta
    \]
    similarly for all $v \in Y_{i}$ we have 
    \[
    s_{\kappa_{y}}(v) - s_{\kappa_{x}}(v) \geq \tilde{L}-2i \eta
    \]
    This implies that every vertex $v \in Y_{i}$ has $s_{\kappa_{y}}(v) \geq \tilde{L}-2i \eta$ and that every vertex $u \in X_{i+1}$ has $s_{\kappa_{y}}(u) \leq (1+2i) \eta$. 
    Since the graph is bipartite, we have that 
    \[
    \sum \limits_{u \in X_{i+1}} d_{\kappa_y}(u) \geq \sum \limits_{v \in Y_{i}} d_{\kappa_y}(v)
    \]
    Applying the above, we find that 
    \[
     ((1+2i)\eta + \frac{\Delta_{\max}}{t})|X_{i+1}|\geq \sum \limits_{u \in X_{i+1}} d_{\kappa_y}(u) \geq \sum \limits_{v \in Y_{i}} d_{\kappa_y}(v) \geq (\tilde{L}-2i\eta+ \frac{\Delta_{\max}}{t})|Y_{i}|
    \]
    Since $\tilde{d} \leq \frac{\Delta_{\max}}{t} + \tilde{L}$ by definition and since  $\eta\frac{22\log(n)}{\eps} \leq \frac{\eps\Delta_{\max}}{4t}$ by choice of $\eta$, we have that 
    \[
    \tilde{L} + \frac{\Delta_{\max}}{t}-\eta\frac{22\log(n)}{\eps} \geq \tilde{d}-\eta\frac{22\log(n)}{\eps} \geq (1+\frac{3\eps}{4})\frac{\Delta_{\max}}{t}
    \]
    and that 
    \[
    \frac{\Delta_{\max}}{t}+\eta\frac{22\log(n)}{\eps} \leq (1+\frac{\eps}{4})\frac{\Delta_{\max}}{t}
    \]
    Hence we can write
    \[
     (1+\frac{\eps}{4})\frac{\Delta_{\max}}{t}|X_{i+1}|\geq (1+\frac{3\eps}{4})\frac{\Delta_{\max}}{t}|Y_{i}|
    \]
    Therefore, we find that 
    \begin{align*}
        |X_{i+1}| &\geq \frac{(1+\frac{3\eps}{4})}{(1+\frac{\eps}{4})}|Y_{i}| \\
        &\geq (1-\frac{\eps}{4})(1+\frac{3\eps}{4})|Y_{i}| \\
        &\geq (1+\frac{3\eps}{4}-2\frac{\eps}{4})) |Y_{i}| \\
        &\geq (1+\frac{\eps}{4}) |Y_{i}|
    \end{align*}
    since $Q_{2i} = X_{i+1}$ and $Q_{2i-1} = Y_{i}$, applying the induction hypothesis yields the first case. The second case is symmetric, but we include it for completeness.
    
    Now consider $\mathbb{B}_{2i+1}$. We note that, by construction, we have that $\forall u \in X_{i+1}$
    \[
        s_{\kappa_{x}}(u) - s_{\kappa_{y}}(u) \geq \tilde{L}-(1+2i)\eta
    \]
    similarly for all $v \in Y_{i+1}$ we have 
    \[
    s_{\kappa_{y}}(v) - s_{\kappa_{x}}(v) \geq \tilde{L}-2(i+1) \eta
    \]
    This implies that every vertex $v \in Y_{i+1}$ has $s_{\kappa_{x}}(v) \leq 2(i+1) \eta$ and that every vertex $u \in X_{i+1}$ has $s_{\kappa_{x}}(u) \geq \tilde{L}-(1+2i) \eta$. 
    Since the graph is bipartite, we have that 
    \[
    \sum \limits_{v \in Y_{i+1}} d_{\kappa_x}(v) \geq \sum \limits_{u \in X_{i+1}} d_{\kappa_x}(u)
    \]
    Applying the above, we find that 
    \[
     (2(i+1)\eta + \frac{\Delta_{\max}}{t})|Y_{i+1}|\geq \sum \limits_{u \in X_{i+1}} d_{\kappa_y}(u) \geq \sum \limits_{v \in Y_{i}} d_{\kappa_y}(v) \geq (\tilde{L}-(1+2i)\eta+ \frac{\Delta_{\max}}{t})|X_{i+1}|
    \]
    Since $\tilde{d} \leq \frac{\Delta_{\max}}{t} + \tilde{L}$ by definition and since  $\eta\frac{22\log(n)}{\eps} \leq \frac{\eps\Delta_{\max}}{4t}$ by choice of $\eta$, we have similarly to before that 
    \[
    \tilde{L} + \frac{\Delta_{\max}}{t}-\eta\frac{22\log(n)}{\eps} \geq \tilde{d}-\eta\frac{22\log(n)}{\eps} \geq (1+\frac{3\eps}{4})\frac{\Delta_{\max}}{t}
    \]
    and that 
    \[
    \frac{\Delta_{\max}}{t}+\eta\frac{22\log(n)}{\eps} \leq (1+\frac{\eps}{4})\frac{\Delta_{\max}}{t}
    \]
    So we can write
    \[
     (1+\frac{\eps}{4})\frac{\Delta_{\max}}{t}|Y_{i+1}|\geq (1+\frac{3\eps}{4})\frac{\Delta_{\max}}{t}|X_{i+1}|
    \]
    Therefore, we find, as before, that
    \begin{align*}
        |Y_{i+1}| &\geq (1+\frac{\eps}{4}) |X_{i+1}|
    \end{align*}
    since $Q_{2i+1} = Y_{i+1}$ and $Q_{2i} = X_{i+1}$, applying the induction hypothesis thus yields the second case.
\end{proof}
Note that Theorem~\ref{thm:invGaran} shows that any colouring satisfying the local condition is a $t$-degree splitting colouring. 

\subsection{Dynamically Maintaining Invariant~\ref{inv:local}}
Again, we assume we have access to some upper bound $\Delta_{\max}$ on the maximum degree over the entire update sequence, and some input parameter $t$. We further require that $\tfrac{\Delta_{\max}}{t} \geq \tfrac{10^4 \log^2 n}{\eps^{2}}$. We will set $\eta = \floor{\tfrac{\eps^2}{128\log n}}\tfrac{\Delta_{\max}}{t}$, and note that $\eta \geq 16$ always due to our assumptions.

We will maintain the following data structures. At each vertex $u$, we will store:
\begin{itemize}
    \item The surplus for each colour in an index-able array $SP(u)$.
    \item For each colour $\kappa$, we store a BBST, $MSD_{\kappa}(u)$, containing all estimated surplus differences of the form:
    \[
    \tilde{s}_{c(ux)}(u)+\tilde{s}_{c(ux)}(x) - (\tilde{s}_{\kappa}(u)+\tilde{s}_{\kappa}(x))
    \]
    \item For each colour $\kappa$, we store a doubly-linked-list $LL(u, \kappa)$ containing all edges coloured $\kappa$ incident to $u$.
    \item For each colour $\kappa$, we store a counter $\iota(u,\kappa) \in [d_{k}]$ and a pointer to the $\iota(u,\kappa)^{\text{th}}$ element of $LL(u, \kappa)$.
\end{itemize}
For each edge $e$, we store a BBST over the estimated surpluses for each colour: $T^{s}_{e}$. 

The algorithm now works as follows (see below for pseudo-code). Whenever an edge $e = uv$ is inserted, we calculate the surplus $s_{\kappa}(u) + s_{\kappa}(v)$ for all choices of $\kappa$, and assign $e$ an arbitrary colour $\kappa'$ with minimum surplus.
We also initialize all of $e$'s data structures, and update its contributions to $u$ and $v$'s data structures. 
Then, we insert $e$ into the linked-lists $LL(u, \kappa')$ and $LL(v, \kappa')$, ensuring that $e$ is inserted just before position $\iota(u,\kappa')$ in $LL(u, \kappa')$ resp.\  $\iota(v,\kappa')$ in $LL(v, \kappa')$.
Then we update data structures for other edges $f$ around $w \in \{u,v\}$ in a round-robin fashion: starting with the edge stored at position $\iota(w,\kappa)$ in $LL(w, \kappa)$, we recompute the contents of the next $\ceil{500 \tfrac{\Delta_{\max}}{\eta}}$ edges in $LL(w, \kappa)$ and update $\iota(w,\kappa)= \iota(w,\kappa) + \ceil{500 \tfrac{\Delta_{\max}}{\eta}} \mod d_{\kappa}(w)$. Note that this is done for all choices of colours $\kappa$. 

Then, we check for all colours if any violations have arisen. If so, we push the endpoint and the colour with too large a surplus difference, unto a queue $Q$. After this, we process $Q$ one pair at the time by recolouring with a similar procedure as above, possibly pushing new violating pairs unto $Q$. 
Note that to get an edge to recolour, we extract the observed maximum and check if its surplus difference is still large enough. If yes, we process it by recolouring it, if no, we conclude our processing of this element of $Q$. 
Deletions are symmetric in the sense that we simply remove the edge and update the affected data structures, before we update in a round-robin fashion and fix violations exactly as before. Here, we have to remember to remove $e$ from any doubly-linked-lists containing it, and to re-adjust the relevant $\iota(\cdot{}, \cdot{})$ by decrementing them by $1$, if $e$'s position in $LL(\cdot{}, \cdot{})$ was before the element pointed to by $\iota(\cdot{}, \cdot{})$. 

Note that these running times can probably be optimised with respect to $t$, but since any such optimisations are likely to only change constant factors hidden by big-Oh notation anyways in the final edge-colouring algorithm, we have not attempted them here. Such optimisation, however, would yield a better $t$-splitter. 

We update $T_{e}^{s}$ be calculating the surplus for each colour of $e$ and inserting it into a search tree. 
We then update (deleting the old, and inserting the new) the surplus difference for $e$ into $MSD_{\kappa}(w)$ for all colours $\kappa$ and all endpoints $w$ of $e$. 
These data structures always contain correct estimates immediately after being updated. Hereafter they might slowly deteriorate in quality, until they are updated again. 
We will later show that this deterioration is small enough to be handled efficiently. 
We will denote by $\tilde{s}_{\kappa}(u) + \tilde{s}_{\kappa}(v)$ the estimate of a surplus stored in a data structure and by $s_{\kappa}(u) + s_{\kappa}(v)$ the actual value of a surplus.
In particular, we will denote by $\Delta \tilde{s}_{\kappa}(w)$ the maximum estimated surplus difference stored in $MSD_{\kappa}(w)$: 
\[
\Delta \tilde{s}_{\kappa}(w) = \max_{x \in N(w)} \tilde{s}_{c(wx)}(w)+\tilde{s}_{c(wx)}(x) - (\tilde{s}_{\kappa}(w)+\tilde{s}_{\kappa}(x))
\]
Note that this value can be found by querying $MSD_{\kappa}(w)$ and returning the maximum. 

\noindent\begin{minipage}[t]{.5 \textwidth}
\null 
 \begin{algorithm2e}[H]
    \caption{\texttt{Insert}$(uv)$}
    \label{alg:insertion}
    \begin{algorithmic}
        \STATE $\texttt{add}(uv)$ to $G$
        \STATE $\texttt{Initialize-DS}(uv)$
        \STATE $Q \gets \texttt{ReColour}(uv, \emptyset)$
        \WHILE{$Q \neq \emptyset$}
            \STATE $(w, \kappa) \gets \texttt{pop}(Q)$
            \IF{$\Delta \tilde{s}_{\kappa}(w) > \frac{\eta}{2}$}
                \STATE $e \gets$ edge maximising the above expression.
                \STATE $Q \gets \texttt{ReColour}(e, Q)$
            \ENDIF
        \ENDWHILE
    \end{algorithmic}
  \end{algorithm2e}
\end{minipage}~%
\begin{minipage}[t]{.5\textwidth}
\null
 \begin{algorithm2e}[H]
    \caption{\texttt{Delete}$(\overrightarrow{uv})$}
    \label{alg:deletion}
    \begin{algorithmic}
        \STATE $\texttt{Remove}(uv)$ from $G$.
        \STATE Update $uv$'s influence on relevant $SP(\cdot{})$,$LL(\cdot{},\cdot{})$, $\iota(\cdot{}, \cdot{})$, $MSD_{\cdot{}}(\cdot{})$
        \STATE $\texttt{Round-Robin}(uv)$
        \FORALL{$\kappa \in [t]$ and $w \in \set{u,v}$}
            \IF{$ \Delta \tilde{s}_{\kappa}(w) > \frac{\eta}{2}$}
                \STATE $\texttt{Push}((w, \kappa))$ to $Q$
            \ENDIF
        \ENDFOR
        \WHILE{$Q \neq \emptyset$}
            \STATE $(w, \kappa) \gets \texttt{pop}(Q)$
            \IF{$\Delta \tilde{s}(w) > \frac{\eta}{2}$}
                \STATE $e \gets$ edge maximising the above expression.
                \STATE $Q \gets \texttt{ReColour}(e, Q)$
            \ENDIF
        \ENDWHILE

    \end{algorithmic}
  \end{algorithm2e}
\end{minipage}

\noindent\begin{minipage}[t]{.5 \textwidth}
\null 
 \begin{algorithm2e}[H]
    \caption{\texttt{ReColour}$(uv, Q)$}
    \label{alg:insertion}
    \begin{algorithmic}
            \IF{$c(uv) = \emptyset$}
                \STATE $c(uv) \gets \texttt{ExtractMinColour}(T^{s}_{uv})$
                \STATE Update $uv$'s influence on relevant $SP(\cdot{})$, $LL(\cdot{},\cdot{})$,  $\iota(\cdot{}, \cdot{})$, $MSD_{\cdot{}}(\cdot{})$
                \STATE $\texttt{Update-DS}(uv)$
            \ELSE
                \STATE $c(uv) \gets \texttt{ExtractMinColour}(T^{s}_{uv})$
                \STATE Update $uv$'s influence on relevant $SP(\cdot{})$, $LL(\cdot{},\cdot{})$, $\iota(\cdot{}, \cdot{})$, $MSD_{\cdot{}}(\cdot{})$
                \STATE $\texttt{Update-DS}(uv)$
            \ENDIF
            \STATE $\texttt{Round-Robin}(uv)$
            \FORALL{$\kappa \in [t]$ and $w \in \set{u,v}$}
                \IF{$\Delta \tilde{s}_{\kappa}(w) > \frac{\eta}{2}$}
                    \STATE $\texttt{Push}((w, \kappa))$ to $Q$
                \ENDIF
            \ENDFOR
            \STATE $\texttt{Return}$ $Q$
    \end{algorithmic}
  \end{algorithm2e}
\end{minipage}~%
\begin{minipage}[t]{.5\textwidth}
\null
 \begin{algorithm2e}[H]
    \caption{\texttt{Round-Robin}$(uv)$}
    \label{alg:RR}
    \begin{algorithmic}
        \FORALL{$\kappa \in [t]$ and $w \in \set{u,v}$}
            \FORALL{$j \in [\ceil{\tfrac{500 \Delta_{\max}}{\eta}}]$}
                \STATE $e \gets LL(w, \kappa)[\iota(w, \kappa)+(j-1) \mod d_{\kappa}(w)]$
                \STATE $\texttt{Update-DS}(e)$
            \ENDFOR 
            \STATE $\iota(w, \kappa) \gets \iota(w, \kappa)+\ceil{\tfrac{500 \Delta_{\max}}{\eta}} \mod d_{\kappa}(w)$
            \STATE Update pointer associated with $\iota(w, \kappa) $
        \ENDFOR
    \end{algorithmic}
  \end{algorithm2e}
\end{minipage}
\noindent\begin{minipage}[t]{.5 \textwidth}
\null 
 \begin{algorithm2e}[H]
    \caption{\texttt{Initialize-DS}$(uv)$}
    \label{alg:insertion}
    \begin{algorithmic}
            \IF{$d(w) = 1$ for $w \in \{u,v\}$}
                \STATE Initialize DS for $w$
            \ENDIF
            \STATE Initialize BBST $T^{s}_{uv}$. 
            \FORALL{$\kappa \in [t]$}
                \STATE \texttt{insert}($s_{\kappa}(u)+s_{\kappa}(v)$) into $T^{s}_{uv}$
            \ENDFOR
    \end{algorithmic}
  \end{algorithm2e}
\end{minipage}~%
\begin{minipage}[t]{.5\textwidth}
\null
 \begin{algorithm2e}[H]
    \caption{\texttt{Update-DS}$(uv)$}
    \label{alg:RR}
    \begin{algorithmic}
        \STATE Update $SP(u)$ and $SP(v)$
        \FORALL{$\kappa \in [t]$}
            \STATE \texttt{Update} ($s_{\kappa}(u)+s_{\kappa}(v)$) in  $T^{s}_{uv}$
        \ENDFOR
        \FORALL{$\kappa \in [t]$ and $w \in \set{u,v}$}
            \STATE Update $uv$'s contribution to $MSD_{\kappa}(w)$
        \ENDFOR
    \end{algorithmic}
  \end{algorithm2e}
\end{minipage}

\subsection{Analysing the Dynamic Algorithm}
In this subsection, we will analyse the dynamic algorithm above. We begin by analysing how precise the estimates $\tilde{s}_{\kappa}(u)+\tilde{s}_{\kappa}(v)$ in the relevant data structures are. 
Then we show that the algorithm actually maintains Invariant~\ref{inv:local}, thus establishing correctness via Theorem~\ref{thm:invGaran}. 
Then we analyse the recourse of the algorithm, before we pick explicit implementations of the data structures and analyse the final running time.

\subsubsection{Analysis of Round-Robin}
In this section, we will analyse the round-robin scheme applied by the algorithm. In particular, we will show that the algorithm maintains Invariant~\ref{inv:local}. We will do this by showing that even-though the algorithm sometimes works with (slightly) out-dated information, the information is still precise enough for the algorithm to detect necessary changes and perform them efficiently. 

To this end, we first show that the data structures used by the algorithm contain good estimates. Before doing this, we will introduce some notation. 
For any data structure $\mathcal{D}$, we will let $(\mathcal{D})_{\tau}$ be the state of the data structure after the $\tau^{th}$ call to \texttt{RoundRobin} concluded, and just before the $(\tau +1)^{th}$ call to \texttt{RoundRobin} commences. 
Similarly, for any value $p(x)$ we will let $(\tilde{p}(x))_{\mathcal{D},\tau}$ denote the estimate of this value stored in data structure $\mathcal{D}$ at time $\tau$, and we will let $(p(x))_{\tau}$ denote the actual value of the parameter at time $\tau$. 
We will say that something happens \emph{during} time step $\tau$ if it happens between the $(\tau-1)^{th}$ and $\tau^{th}$ call to \texttt{RoundRobin}.
If the data structure is clear from the context, we will leave out this argument to ease readability. 
\begin{lemma} \label{lma:rrbounds}
    At all times $\tau$, we have that: 
    \begin{enumerate}
        \item For any edge $e=uv$ and any colour $\kappa$, the estimates $(\tilde{S})_{\tau}:= (\tilde{s}_{\kappa}(u)+\tilde{s}_{\kappa}(v))_{\tau}$ in $T^{s}_{uv}$ satisfy that: 
        \[
        |(\tilde{S})_{\tau} - ((s_{\kappa}(u))_{\tau}+(s_{\kappa}(v))_{\tau})| \leq \tfrac{\eta}{16}
        \]
        \item For any vertex $u$, any neighbour $v$, and any colour $\kappa$, the estimates 
        \[
        (\tilde{E}(u,uv,\kappa))_{\tau} = ((\tilde{s}_{c(uv)}(u)+\tilde{s}_{c(uv)}(v)) - (\tilde{s}_{\kappa}(u)+\tilde{s}_{\kappa}(v)))_{\tau}
        \]
        stored in $MSD_{\kappa}(u)$ satisfy that: 
        \[
        |(E(u,uv,\kappa))_{\tau}-(\tilde{E}(u,uv,\kappa))_{\tau}| \leq \tfrac{\eta}{8}
        \]
        where 
        \[
        (E(u,uv,\kappa))_{\tau} = ((s_{c(uv)}(u)+s_{c(uv)}(v)) - (s_{\kappa}(u)+s_{\kappa}(v)))_{\tau}
        \]
    \end{enumerate}
\end{lemma}
\begin{proof}
    We will show item $1)$ first. Let $\tau'\leq \tau$ be the last time before (or at) time $\tau$, where $\tilde{S}$ was updated. We then have that 
    \[
    (\tilde{S})_{\tau} = (\tilde{S})_{\tau'} = (\tilde{s}_{\kappa}(u)+\tilde{s}_{\kappa}(v))_{\tau'} = (s_{\kappa}(u))_{\tau'}+(s_{\kappa}(v))_{\tau'}
    \]
    and so it is sufficient to estimate differences of the form $|(s_{\kappa}(w))_{\tau}-(s_{\kappa}(w))_{\tau'}|$. 
    To this end, we note that due to the round-robin scheme being run every time an edge incident to $w$ changes its colour, we can change $s_{\kappa}(w)$ by at most
    \[
    \frac{\Delta_{\max}}{\ceil{\tfrac{500\Delta_{\max}}{\eta}}} \leq \tfrac{\eta}{500}
    \]
    before it is updated again. As a result, we obtain that $|(s_{\kappa}(w))_{\tau}-(s_{\kappa}(w))_{\tau'}| \leq \tfrac{\eta}{32}$, and so the above observations yield: 
    \begin{align*}
        |(\tilde{S})_{\tau} - ((s_{\kappa}(u))_{\tau}+(s_{\kappa}(v))_{\tau})| & \leq  |((s_{\kappa}(u))_{\tau'}+(s_{\kappa}(v))_{\tau'}) - ((s_{\kappa}(u))_{\tau}+(s_{\kappa}(v))_{\tau})| \\
            &\leq |(s_{\kappa}(u))_{\tau'}-(s_{\kappa}(u))_{\tau}| + |(s_{\kappa}(v))_{\tau'}-(s_{\kappa}(v))_{\tau}| \\
            &\leq \tfrac{\eta}{16}
    \end{align*} 
    For the second item, first observe that every time $MSD_{\kappa}(u)$ is updated, it is updated with a difference of entrances from some BBST $T^{s}_{ux}$. This means that if $c(uv) = \kappa'$, then for all $\kappa$, we actually have that: 
    \[
        (\tilde{E}(u,uv,\kappa))_{\tau} = (\tilde{s}_{\kappa'}(u)+\tilde{s}_{\kappa'}(v))_{\tau'} - (\tilde{s}_{\kappa}(u)+\tilde{s}_{\kappa}(v))_{\tau'}
    \]
    for some $\tau' \leq \tau$ denoting the last time before (or at) $\tau$ that $MSD_{\kappa}(u)$ was updated. 
    Observe that any time we update the entry $\tilde{E}(u,uv,\kappa)$ in $MSD_{\kappa}(u)$, we do so right after updating both $\tilde{s}(u)_{\kappa'}$ and $\tilde{s}(v)_{\kappa'}$. 
    Therefore, we can in fact write: 
    \[
        (\tilde{E}(u,uv,\kappa))_{\tau} = (s_{\kappa'}(u)+s_{\kappa'}(v))_{\tau'} - (s_{\kappa}(u)+s_{\kappa}(v))_{\tau'}
    \]
    By construction, we also have that:
    \[
        (E(u,uv,\kappa))_{\tau} = ((s_{\kappa'}(u))_{\tau}+(s_{\kappa'}(y))_{\tau}) - ((s_{\kappa}(u))_{\tau}+(s_{\kappa}(y))_{\tau})
    \]
    Thus similarly to in case $1.$, we conclude that: 
    \begin{align*}
        |(E(u,uv,\kappa)_{\tau}-(\tilde{E}(u,uv,\kappa))_{\tau}|&\leq |(s_{\kappa'}(u)+s_{\kappa'}(v))_{\tau'}-((s_{\kappa'}(u))_{\tau}+(s_{\kappa'}(v))_{\tau})| \\
        &+ |((s_{\kappa}(u))_{\tau}+(s_{\kappa}(v))_{\tau})-(s_{\kappa}(u)+s_{\kappa}(v))_{\tau'}| \\
        &\leq \tfrac{\eta}{16} + \tfrac{\eta}{16} = \tfrac{\eta}{8}
    \end{align*}
\end{proof}

Next we will argue that this information is accurate enough to maintain Invariant~\ref{inv:local}, provided that the algorithm terminates. Afterwards, we show that the algorithm always terminates and in fact does so with small amortised recourse.
\begin{lemma} \label{lma:invHolds}
    Let $G$ be a dynamic graph, and let $c_i$ be the colouring in $G_i$ maintained by the algorithm after $i$ updates. Assume that the algorithm terminates after each update and that $ \tfrac{\Delta_{\max}}{t} \geq \tfrac{10^4 \log^2 n}{\eps^{2}}$. Then $c_i$ satisfies Invariant~\ref{inv:local} for all $i$.
\end{lemma}
\begin{proof}
    It suffices to show that after every update, we must have $\Delta s_{\kappa}(w) < \eta$ for all vertices $w$ and all colours $\kappa$. We first note that by Lemma~\ref{lma:rrbounds}, it follows that if $\Delta \tilde{s}_{\kappa}(w) \leq \tfrac{\eta}{2}$, then we have $\Delta s_{\kappa}(w) \leq \Delta \tilde{s}_{\kappa}(w) + \tfrac{\eta}{8} < \eta$. 
    
    Hence, if we can show that after the algorithm is done processing any update, we must have $\Delta \tilde{s}_{\kappa}(w) \leq \tfrac{\eta}{2}$ for all vertices $w \in V$ and all colours $\kappa \in [t]$, then we are done. 
    To see this, we apply induction on the number of updates processed. The base case is when no updates have been processed yet in which case what we claim clearly is true. Next assume that the above holds after $i$ updates have been processed, and consider the state of the data structures just after the $(i+1)^{\text{th}}$ update concluded. 

    Since the round-robin scheme, which updates the values stored in $MSD_{\kappa}(w)$, is only called on vertices $w$ either directly incident to edges that were inserted, deleted or recoloured or vertices neighbouring such vertices, the desired statement holds for any other vertex by induction. In particular, $MSD_{\kappa}(w)$ is only changed after deleting an edge, giving an edge a new colour, or during a call to $\texttt{Update-DS}$. The latter only happens during a call to $\texttt{Round-Robin}(uv)$ or imedditaly after recolouring an edge. 

    We first note that we need only worry about vertices directly incident to inserted, deleted or recoloured edges, and not neighbours of such vertices. Indeed, we always update an edge $uv$'s contribution to $MSD_{\kappa}(w)$ for $w \in \{u,v\}$ for both endpoints to the same value. This means that after a call to $\texttt{Update-DS}(uv)$ or a direct update terminates, we always have $\tilde{E}(u,uv,\kappa) = \tilde{E}(v,uv,\kappa)$ i.e.\ for all $\tau$ we have $(\tilde{E}(u,uv,\kappa))_{\tau} = (\tilde{E}(v,uv,\kappa))_{\tau}$.
    Hence, it is sufficient to verify that any endpoint $w$ directly incident to an insertion, deletion or recolouring satisfies $\Delta \tilde{s}_{\kappa}(w) \leq \tfrac{\eta}{2}$ after the update procedure ends. 
    
    After every $\texttt{Round-Robin}(uv)$ call or direct update, any vertex $w$ with $\Delta \tilde{s}_{\kappa}(w) > \tfrac{\eta}{2}$ is immediately added to a queue $Q$. The vertex $w$ can only leave $Q$ in two ways: either because it is verified at some time step $\tau$ that $(\Delta \tilde{s}_{\kappa}(w))_{\tau} \leq \tfrac{\eta}{2}$ or because some $\kappa$-coloured edge incident to $w$ is recoloured. 
    In the first case, $(w,\kappa)$ can only enter $Q$ again if some edge incident to it is recoloured at some time-step $\tau'>\tau$ in which case the exact same argument can be applied again. If no such future update occurs, we immediately have that $\Delta \tilde{s}_{\kappa}(w) \leq \tfrac{\eta}{2}$ after the update procedure ends.
    In the latter case, it is checked after the recolouring whether $\Delta \tilde{s}_{\kappa}(w) > \tfrac{\eta}{2}$ and if so, then $w$ is added back to $Q$ but at a later time-step. If  $\Delta \tilde{s}_{\kappa}(w) \leq \tfrac{\eta}{2}$, then $w$ immediately leaves $Q$ and similar arguments to the first case applies. 

    Since we have assumed that the algorithm terminates, we must eventually end up in the first case. 
\end{proof}
Next, we show that the algorithm actually terminates and analyse its recourse. 

\subsubsection{Analysing the Recourse}
In this subsection, we analyse the recourse of the algorithm and show that the it terminates. We first show that any edge re-coloured by the algorithm results in a significant lowering of the surplus of this edge. 
This allows us to show that a certain global potential of the graph drops enough to ensure that the amortised number of recolourings becomes bounded. 
\begin{lemma} \label{lma:goodflip}
    Let $e=uv$ by any edge recoloured by the algorithm during time-step $\tau$. Suppose that $(c(uv))_{\tau-1} = \kappa_1$ and that $(c(uv))_{\tau} = \kappa_2$ i.e.\ the colour of the edge is changed from $\kappa_1$ to $\kappa_2$. 
    Then at time $\tau-1$, we have that 
    \[
    \delta^{uv}_{\kappa_1 \rightarrow \kappa_2}:= ((s_{\kappa_1}(u))_{\tau-1}+(s_{\kappa_1}(v))_{\tau-1}) - ((s_{\kappa_2}(u))_{\tau-1}+(s_{\kappa_2}(v))_{\tau-1}) \geq \tfrac{\eta}{4}
    \]
\end{lemma}
\begin{proof}
    Any $\kappa$-coloured edge incident to $u$ or $v$ is only recoloured during time-step $\tau$ if $(\Delta \tilde{s}_{\kappa}(u))_{\tau-1} \geq \tfrac{\eta}{2}$ or $(\Delta \tilde{s}_{\kappa}(v))_{\tau-1} \geq \tfrac{\eta}{2}$. 
    Since $e$ was chosen as an edge maximising the surplus difference, it must be the case that 
    \begin{align*}
        \tilde{\delta}^{uv}_{\kappa_1 \rightarrow \kappa_2} &:=  (\tilde{s}_{\kappa_1}(u)+\tilde{s}_{\kappa_1}(v) - (\tilde{s}_{\kappa_2}(u)+\tilde{s}_{\kappa_2}(v)))_{\tau-1} \\
        &= (\tilde{s}_{\kappa_1}(u)+\tilde{s}_{\kappa_1}(v))_{\tau-1} - (\tilde{s}_{\kappa_2}(u)+\tilde{s}_{\kappa_2}(v))_{\tau-1} \\
        &\geq \tfrac{\eta}{2}
    \end{align*}
    where the second equality follows from the fact that $MSD_{\kappa}(u)$ always is updated by a difference of surpluses stored in BBSTs of of the form $T^{s}(uv)$ immediately after all of these data structures have been updated. 
    Hence, it follows that:
    \begin{align*}
        \delta^{uv}_{\kappa_1 \rightarrow \kappa_2}&+(\tilde{s}_{\kappa_1}(u)+\tilde{s}_{\kappa_1}(v))_{\tau-1} - ((s_{\kappa_1}(u))_{\tau-1}+(s_{\kappa_1}(v))_{\tau-1}) - ((\tilde{s}_{\kappa_2}(u)+\tilde{s}_{\kappa_2}(v))_{\tau-1}\\ &- ((s_{\kappa_2}(u))_{\tau-1}+(s_{\kappa_2}(v))_{\tau-1})) \geq \tilde{\delta}^{uv}_{\kappa_1 \rightarrow \kappa_2}\geq \tfrac{\eta}{2}
    \end{align*}
    Consequently, we can apply Lemma~\ref{lma:rrbounds} part 1) in order to observe that: 
    \begin{align*}
        \delta^{uv}_{\kappa_1 \rightarrow \kappa_2} &\geq \tilde{\delta}^{uv}_{\kappa_1 \rightarrow \kappa_2}-|(\tilde{s}_{\kappa_1}(u)+\tilde{s}_{\kappa_1}(v))_{\tau-1} - ((s_{\kappa_1}(u))_{\tau-1}+(s_{\kappa_1}(v))_{\tau-1})| -\\
        &|(\tilde{s}_{\kappa_2}(u)+\tilde{s}_{\kappa_2}(v))_{\tau-1}- ((s_{\kappa_2}(u))_{\tau-1}+(s_{\kappa_2}(v))_{\tau-1})| \\
        &\geq \tfrac{\eta}{2}-\tfrac{2\eta}{16} \geq \tfrac{\eta}{4}
    \end{align*}
    as claimed.
\end{proof}
The idea is now to consider the following potential: 
\[
\Phi(G,c) = \sum \limits_{i=1}^{t} \sum \limits_{v \in V(G)} \sum \limits_{j=1}^{s_{i}(v)} j
\]
As we will show below, Lemma~\ref{lma:goodflip} implies that every recolouring performed by the algorithm will lower the potential substantially. In particular, this ensures that the amortised recourse of the dynamic algorithm is low. This is shown in the following Lemma:
\begin{lemma} \label{lma:tsplitterRecourse}
    Let $\tilde{d} = \max_{c} \max_{v} \tilde s(v)$ be the maximum surplus of any colouring satisfying Invariant~\ref{inv:local} with $0<\eps < 1$. 
    Then the amortised recourse of the algorithm is upper bounded by $16\frac{\tilde{d}}{\eta} \leq \frac{2^{10}\log n}{\eps^2} = O(\tfrac{\log n}{\eps^2})$ provided that $\eta = \floor{\tfrac{\eps^2}{128\log n}}\tfrac{\Delta_{\max}}{t} \geq 16$ and $\tfrac{\Delta_{\max}}{t} \geq \tfrac{10^4 \log^2 n}{\eps^{2}}$.

    In particular, under these assumptions, the algorithm always terminates. 
\end{lemma}
\begin{proof}
    We first note that after inserting and colouring any edge $e = uv$, the potential can increase by at most $2(\tilde{d}+1)$, since between insertions the algorithm guarantees that Invariant~\ref{inv:local} holds, provided that the algorithm terminates, by Lemma~\ref{lma:invHolds}. Note that a deletion will always lower the potential. 

    The second observation is that each recolouring will lower the potential by $\Omega(\eta)$. 
    Indeed, if $e$ is recoloured from having colour $\kappa_1$ to having colour $\kappa_2$ during time-step $\tau$, then it follows from Lemma~\ref{lma:goodflip} that:
    \[
    \delta^{uv}_{\kappa_1 \rightarrow \kappa_2} = ((s_{\kappa_1}(u))_{\tau-1}+(s_{\kappa_1}(v))_{\tau-1}) - ((s_{\kappa_2}(u))_{\tau-1}+(s_{\kappa_2}(v))_{\tau-1}) \geq \tfrac{\eta}{4}
    \]
    Notice that if the recolouring performed in step $\tau$ transforms $(G,c)$ to $(G,c')$, then we have that:
    \[
    \Phi(G,c) - \Phi(G,c') = \delta^{uv}_{\kappa_1 \rightarrow \kappa_2} - 2 \geq \tfrac{\eta}{4} - 2 \geq \tfrac{\eta}{8}
    \]
     Indeed, the above holds since recolouring $e$ reduces $s_{\kappa_1}(u)$ and $s_{\kappa_1}(v)$ by 1, and hence $\Phi$ by $(s_{\kappa_1}(u))_{\tau-1}+(s_{\kappa_1}(v))_{\tau-1}$. On the other hand, $s_{\kappa_2}(u)$ and $s_{\kappa_2}(v)$ are both increased by one, which increases the potential by $(s_{\kappa_2}(u)))_{\tau-1}+(s_{\kappa_2}(v)))_{\tau-1}+2$. Note also that at most one edge can be recoloured during a time-step. 

    Since $\Phi$ is always non-negative, deletions do not lower $\Phi$, and since insertions increase $\Phi$ by at most $2(\tilde{d}+1)$, we find that we can witness at most:
    \[
    \frac{r\cdot{}2(\tilde{d}+1)}{\tfrac{\eta}{8}} = 16 \frac{r\cdot{}\tilde{d}}{\eta}
    \]
    recolourings over a sequence of updates containing $r$ insertions. Consequently, the algorithm always terminates and it does so with recourse bounded by $O(\frac{\tilde{d}}{\eta})$. 
    To recover the last part of the Lemma, note that Theorem~\ref{thm:invGaran} gives that $\tilde{d} \leq (1+\eps)\tfrac{\Delta_{\max}}{t}$ when $\eta = \floor{\tfrac{\eps^2}{128\log n}}\tfrac{\Delta_{\max}}{t}$.
\end{proof}

We briefly remark that it is quite important that $\eta \geq 16$. This ensures that any recolouring actually lowers the potential, and that one does not end up in a situation where the algorithm continually recolours the same edges. 
In the next section, we will analyse the update time as a function of the number of recolourings i.e.\ the recourse performed by the algorithm. We will then use Lemma~\ref{lma:tsplitterRecourse} to upper bound the total update time.

\subsubsection{Analysing the Update Time} 

We begin by implementing the basic procedures in the following lemma. The proof is straight-forward, but we include a proof for completeness.
\begin{lemma} \label{lma:primitiveUpdateTimes}
    For vertex $v$, we can support the following operations with deterministic and worst-case guarantees as listed below:
    \begin{itemize}
        \item For all edges $uv$ incident to $v$, we can perform $\emph{\texttt{Initialize-DS}}(uv)$ in $O(t \log m)$ time. 
        \item For all edges $uv$ incident to $v$, we can perform $\emph{\texttt{Update-DS}}(uv)$ in $O(t \log m)$ time.
        \item For all edges $uv$ incident to $v$, we can perform $\emph{\texttt{Round-Robin}}(uv)$ in $O(t^3 \eps^{-2} \log^2 m)$ time.
        \item We can perform one call to $\emph{\texttt{ReColour}}(e)$ in $O(t^3 \eps^{-2} \log^2 m)$ time.
        \item We can extract an edge (and the value) maximising $\Delta \tilde{s}_{\kappa}(v)$ in $O(\log m)$ time. 
        \item We can perform $\emph{\texttt{ExtractMinColour}}(T^{s}_{uv})$ in $O(\log m)$ time.
    \end{itemize}
\end{lemma}
\begin{proof}
    We will only require our BBST's to support very basic operations, so any standard BBST such as red-black trees~\cite{cormen2022introduction} will do. 

    For the first item, we note that initializing $MSD_{\kappa}(w)$ for any $\kappa \in [t]$ and any $w \in V(G)$ takes constant time since $MSD_{\kappa}$ is initialised as an empty BBST. 
    Initialising the doubly linked-lists over neighbouring edges can also be done trivially in $O(1)$ time for each $\kappa$.
    Since, we have to do the above for $t$ choices of $\kappa$ and up to $2$ choices of $w$, this part takes $O(t)$ time. 
    We can initialize $T^{s}_{uv}$ in $O(1)$-time and insert the $t$ surplus sums into $T^{s}_{uv}$ in $O(t\log t)$ time.
    Finally, we can insert the surplus differences into the relevant $MSD_{\cdot{}}(\cdot{})$'s in $O(t\cdot{}\log m)$ time, as $MSD_{\kappa}(w)$ contains $O(m)$ elements for all $\kappa, w$. 

    For the second item, we note that each update to $T^{s}_{uv}$ takes $O(\log t)$ time, since we can compute $s_{\kappa}(u) + s_{\kappa}(v)$ in constant time. 
    Each update to $MSD_{\kappa}(v)$ can also be performed in $O(\log n)$ time. Indeed, we can compute $(s_{c(uv)}(u) + s_{c(uv)}(v))-(s_{\kappa}(u) + s_{\kappa}(v))$ in $O(1)$ time, and insert it in $O(\log m)$ time. 
    Similarly, by storing pointers to the elements in a BBST at $w$, we can also remove the old entrance in $O(\log m)$ time. 
    All of the above has to be performed $t$ times, leaving us with the stated update time.  

    For the third item, we note that we perform $2t\cdot{}\ceil{500\tfrac{\Delta}{\eta}}$ calls to $\texttt{Update-DS}(e)$, which dominates the update time. Indeed, we can index into the doubly-linked-list in $O(1)$ time, provided that we have stored a pointer to the element, we wish to index, as well as its position in the linked list. 
    Applying the second item therefore yields the running time.

    For the fourth item, we note that we can perform $\texttt{ExtractMinColour}(T^{s}_{uv})$ in $\log t$ time by simply computing the minimum in the BBST. 
    Hence, if we disregard further recursive call, the update time is dominated by the call to $\texttt{Round-Robin}(uv)$, so the statement follows from the third item.

    For the fifth item, we note that by storing a pointer from each entrance in $MSD_{\kappa}(v)$ to the edge responsible for the entrance, we can, given access to an entrance, locate the corresponding edge in $O(1)$ time. 
    As such, we only need to compute the maximum entrance in $MSD_{\kappa}(v)$ which can be done  in $O(\log m)$ time.

    For the sixth and final item, we note that extracting the maximum entrance in a BBST can be performed in $O(\log n)$ time by a standard search always going to the largest child. 
\end{proof}

Everytime we call $\texttt{ReColour}$ on some edge, we perform a recolouring. Hence, by Lemma~\ref{lma:tsplitterRecourse}, after $i$ insertions, we have performed at most $16\cdot{}i\cdot{}\frac{\tilde{d}}{\eta}$ calls to $\texttt{ReColour}$. Thus, in an amortised sense, the number of calls to $\texttt{ReColour}$ is $O(\frac{\tilde{d}}{\eta})$ per update. 
If we ignore the recursive calls, then by Lemma~\ref{lma:primitiveUpdateTimes} the time spent for each call is dominated by the call to $\texttt{Round-Robin}$, and so each call to $\texttt{Recolour}$ takes only $O(t^3\eps^{-2} \log^2 m)$ time. 
Each recolouring puts at most $2t$ pairs unto $Q$. The time required to process such pairs which do not result in a recolouring can be, in an amortised sense, ascribed to the recolouring that pushed the pair unto $Q$. 
This time is also dominated by the time is dominated by the call to $\texttt{Round-Robin}$. 

Hence, we have shown the following bound on the amortised running time of $\texttt{insert}$ and $\texttt{delete}$.
\begin{lemma} \label{lma:insdel}
    The amortised running time of $\texttt{insert}$ and $\texttt{delete}$ is upper-bounded by $O(t^3\eps^{-4} \log^3 m)$.
\end{lemma}
Combining Lemma~\ref{lma:invHolds}, Lemma~\ref{lma:tsplitterRecourse}, and Lemma~\ref{lma:insdel} now gives Theorem~\ref{thm:tsplit}.

\section{Maintaining a hierarchy of $t$-splitters} \label{sct:hier}
In this section, we will show how to build a hierarchy of dynamic $t$-splitters in order to split the degree all the way down to $O(\poly(\eps^{-1}, \log n))$ while retaining an update time in $2^{\tilde{O}_{\log \eps^{-1}}(\sqrt{\log n})}$. Thus we will show Theorem~\ref{thm:splitHier}, which we restate below: 
\begin{theorem}[Identical to Theorem~\ref{thm:splitHier}]
    Let $G$ be a dynamic graph, and let $1 > \eps > 0$ and an upper bound $\Delta_{\max} \geq \tfrac{10^7 \log^5 n}{\eps^{3}}$ on the maximum degree of $G$ through the entire update sequence be given. 
    Set $t_1=\floor{2^{10\sqrt{\log n}}}$ and $t_2 = \floor{\tfrac{\log n}{\eps}}$.
    Then there exists a dynamic algorithm maintaining parameters $h,i,j$ with $h = i+j \leq 4 \sqrt{\log n}$ and a set of graphs $\mathcal{G}_{h}$ satisfying that:
    \begin{enumerate}
         \item $\{E(G_i)\}_{G_i \in \mathcal{G}_{h}}$ partitions $E(G)$
        \item $|\mathcal{G}_{i+j}| \leq t_1^{i}t_2^{j} $
        \item $\hat{\Delta}_{i+j}  \leq (1+\tfrac{\eps}{16})\tfrac{\Delta_{\max}}{t_1^{i}t_2^{j}} \leq \poly{\paren{\tfrac{\log n}{\eps}}}$
    \end{enumerate}
    The algorithm has amortised recourse and update time both in $2^{\tilde{O}_{\log \eps^{-1}}(\sqrt{\log n})}$.
\end{theorem}
In this section, we will again assume we are given an upper bound $\Delta_{\max}$ on $\Delta$. 
We first give a formal description of the hierarchy and an algorithm that maintains it.

\subsection{The hierarchy} 
We first describe the hierarchy in words.
The hierarchy will consist of two types of $t$-splitters. The first type of $t$-splitter, we will use, has $t_1 = \floor{2^{10\sqrt{\log n}}}$, and the second type has $t_2 = \floor{\tfrac{\log n}{\eps}}$. 
Given any graph $G$, any $\Delta_{\max}$, and any $0<\eps<1$, we will select some $h$. The hierarchy then consists of a sequence of sets of graphs $\{\mathcal{G}_{i}\}_{i = 0}^{h}$ and a sequence of values $\{\hat{\Delta}_{i}\}_{i = 0}^{h}$ such that for all $0 \leq i \leq h$ and all $H\in \mathcal{G}_{i}$, we have that $\max_{v \in V(H)} d(v) \leq \hat{\Delta}_{i}$ i.e.\ the sequence of $\hat{\Delta}_i$.

We note here that the algorithm from Theorem~\ref{thm:tsplit} is \emph{not} history independent so 'running' it on a static graph is ill-defined (one can think of the edges being processed in an arbitrary order). We will ignore this in the following informal description, and then provide a proper definition afterwards.
We will initialise $\mathcal{G}_0 = \{G\}$ and $\hat{\Delta}_0 = \Delta_{\max}$, and proceed as follows: given $\mathcal{G}_{i}$ and $\hat{\Delta}_i$, we let $\mathcal{G}_{i+1}$ be the set of graphs obtained in the following manner.
If $\Delta_{\max}\geq 2^{10i \sqrt{\log n}}$, we will run the algorithm from Theorem~\ref{thm:tsplit} with the parameter $\mu = \tfrac{\eps'^2}{\log n}$ and $t = t_1$ on all $H \in \mathcal{G}_{i}$. For each $H$, every colour-class $\kappa$ induces a subgraph $H_{\kappa}$ of $H$. For each choice of $H$ and $\kappa$, we include $H_{\kappa}$ in $\mathcal{G}_{i+1}$ after removing all degree $0$ vertices. Finally, by Theorem~\ref{thm:tsplit}, we can choose $\hat{\Delta}_{i+1} = (1+\mu)\tfrac{\hat{\Delta}_i}{t}$.
Otherwise we set $t = t_2$ and repeat the above construction. 
We will pick $h$ such that the construction terminates with $\hat{\Delta}_{h} = O(\poly(\eps^{-1}, \log n))$, but still large enough for any application of Theorem~\ref{thm:tsplit} to go through. 

The idea is that the edgesets of the graphs in $\mathcal{G}_{h}$ form a partition of the edges of $G$ such that each partition induces a low-degree subgraph of $G$. 
Since the $t$-splitters used are fully-dynamic, we can maintain the above hierarchy dynamically as well (see Algorithms~\ref{alg:InitHierarchy}~\ref{alg:insertionHierarchy}~\ref{alg:DeletionHierarchy}). 
To do so, we simply build the $t$-splitters as specified, and whenever an edge is inserted, deleted or recoloured, we remove it from the graphs it previously belonged to (if any) and instead insert it into the graphs specified by its new colour (if any).
The order of operations does not matter too much, as long as we guarantee that any $t$-splitter algorithm that is called is allowed to finish and that the $\hat{\Delta}_{i}$'s remain valid upper-bounds through-out the update sequence. In Algorithms~\ref{alg:InitHierarchy}~\ref{alg:insertionHierarchy}~\ref{alg:DeletionHierarchy} we will explicitly fix the order of operations.  

Next, we introduce some basic notation in order to make the above description formal. 
First of all, we define a $(G,\Delta_{\max}, \eps)$-valid-hierarchy. 
This is any hierarchy obtained by splitting via any $t$-splitter satisifying Invariant~\ref{inv:local}, and not necessarily only via the dynamic $t$-splitter from Theorem~\ref{thm:tsplit}.
This deals with the technicality that the data structure from Theorem~\ref{thm:tsplit} might not be history-independent, and so we cannot dynamically hope to maintain a very strictly defined hierarchy.
\begin{definition} \label{def:hierarchy}
    Given $G,\Delta_{\max}$, and $\eps$. A $(G,\Delta_{\max}, \eps)$-valid-hierarchy of depth $h$ is a sequence of sets of graphs $\{\mathcal{G}_{i}\}_{i = 0}^{h}$, a sequence of sets of colourings $\{\mathcal{C}_{i}\}_{i = 0}^{h}$, and two sequences of integers $\{\hat{\Delta}_{i}\}_{i = 0}^{h}$ and $\{t(i)\}_{i = 0}^{h-1}$ such that:
    \begin{itemize}
        \item $\mathcal{G}_0 = \{G\}$ and $\hat{\Delta}_0 = \Delta_{\max}$. 
        \item For all $i$: if $\Delta_{\max}\geq \floor{2^{10\sqrt{\log n}}}\floor{2^{\sqrt{\log n}}}^{i-1}$ then $t(i)=t_1$, else $t(i)=t_2$.
        \item  For each $H \in \mathcal{G}_{i}$ there is a corresponding $t(i)$-splitting colouring $c_H \in \mathcal{C}_{i}$.
        \item $\mathcal{G}_{i+1} := \{\hat{H}: \exists H \in \mathcal{G}_{i}, \kappa \in [t(i)] \text{ with } H[\kappa] = \hat{H} \text{ under } c_H\}$. 
        \item For all $i$: $\hat{\Delta}_{i+1} \leq (1+\tfrac{\eps}{128 \log n}) \tfrac{\hat{\Delta}_{i}}{t(i)}$
        \item For all $i$, all $ H \in \mathcal{G}_{i}$, and all $v \in V(H)$: $d_{H}(v) \leq \hat{\Delta}_{i}$
    \end{itemize}
\end{definition}
\begin{remark}
    We remark that the choice of the value $2^{(10+(i-1))\log n}$ is not extremely precise. One could have equally well chosen $2^{(9+(i-1))\log n}$ or $20\cdot{}2^{(10+(i-1))\log n}$. One then need only change the constant in Algorithms~\ref{alg:InitHierarchy}~\ref{alg:DeletionHierarchy}~\ref{alg:insertionHierarchy}, in Definition~\ref{def:hierarchy}, and in the subsequent proofs mutatis mutandi. 
    In particular, one can always discretize the values and slightly over- or under-estimate all parameters like $\Delta_{\max}, \eps$ or $2^{(10+(i-1))\log n}$ to ensure that one does not need access to real RAM in order to do fast comparisons.
\end{remark}
In particular, any hierarchy that uses a valid instantiation of Theorem~\ref{thm:tsplit} for every $t$-splitter yields a $(G,\Delta_{\max}, \eps)$-valid-hierarchy; regardless of the insertion order. We will also observe that for all $i$, the edge sets of the graphs in $\mathcal{G}_{i}$ form a valid partition of $E(G)$.
We will let a vector $\chi(e)$ describe the \emph{hierarchy-profile} of an edge:
\begin{definition}
    For any edge $e$ and any $(G,\Delta_{\max}, \eps)$-valid-hierarchy, we let the \emph{hierarchy-profile} of $e$ be the vector $\chi(e) \in \mathbb{Z}^{h}$ such that if $e \in E(H)$ for $H \in \mathcal{G}_{i}$, then $\chi(e)_{i} = c_{H}(e)$.
\end{definition} 
For any $i \leq h$ and any vector $v \in \mathbb{Z}^{h}$, we let the $i^{\text{th}}$ initial segment of $v$ be the vector $v^{(i)}  \in \mathbb{Z}^{i}$ whose first $i$ coordinates agree with the first $i$ coordinates of $v$. 
For any vector $v\in \mathbb{Z}^{i}$ and constant $k \in \mathbb{Z}$, we also let $[v^{(i)},k]$ be the vector in $\mathbb{Z}^{i+1}$ with $i^{th}$ initial segment equal to $v$ and $[v,k]_{i+1} = k$
We can index the graphs in the hierarchy by providing an index set $J$ containing every initial segment of the hierarchy-profile of every edge $e \in G$. 
Here, we let $G_\emptyset = G$, and we let $G_j$ for any $j \in J$ be the graph containing exactly the edges $e$ for which $j$ is an initial segment of $\chi(e)$. 
We let $|j|:=dim(j)$ denote the dimension of the vector $j$. Also $|j|$ is the \emph{level} of $G_j$. 

\subsection{An algorithm for maintaining a hierarchy}
The algorithm for dynamically maintaining the hierarchy is quite straight-forward. 
First of all, one initializes values $t(i)$ and $\hat{\Delta}_{i}$ so that they conform to the constraints of Definition~\ref{def:hierarchy} (see Algorithm~\ref{alg:InitHierarchy}). This only has to be done once. 

The hierarchy is then maintained by repeatedly applying Theorem~\ref{thm:tsplit} to further split graphs into ones with lower degree.
First, one performs the update on level $0$, which then propagates the update along with any extra recolourings to the next level of the hierarchy.
One has to be a bit careful with determining the order of deletions and insertions, in order to ensure that the $\hat{\Delta}_{i}$'s remain valid upper-bounds through-out the update sequence. 
The algorithm will delete in a bottom-up fashion, beginning at the bottom of the hierarchy, and insert in a top-down fashion.
That is every time an edge is deleted or recoloured, the algorithm first deletes every affected instance of the edge, beginning at the bottom. 
When an edge is inserted (or receives a new colour), the edge is inserted into the top levels of the structure first, in order to ensure that it does not enter a lower level structure before it is guaranteed that the maximum degree is preserved.
After processing an insertion or deletion at a level some edges at the level are (potentially) recoloured.
In order to process the recoloured edges, we ensure to first perform all deletions, and subsequently the insertions. 
Thus, we can ensure that the degree bounds hold at all times before an insertion in to a graph is processed.
This ensures that $\hat{\Delta}_{i}$'s remain valid upper-bounds after the updates. 
Finally, we note that the hierarchy consists of many graphs, which can either be constructed in a pre-processing phase or simply allocated online by only inserting vertices into a graph once the vertex actually has degree at least one. 

Next, we give pseudo-code for the above dynamic algorithm. 
If an edge $e$ has just been recoloured, we let $c(e)_{\text{old}}$ respectively $\chi(xy)_{\text{old}}$ denote the colours of $e$ just before the recolouring. The hierarchy is maintained via the following set of recursive procedures.
The first procedure shows how to initialize the parameters needed for the following procedures. We note already now that to actually perform an insertion or a deletion, one must call the respective recursive procedure with the correct parameter setting (see Remark~\ref{rmk:parset}): \\\\
\begin{algorithm2e}[H]
    \caption{\texttt{Initialize-Parameters}$(\Delta_{\max}, \eps)$}
    \label{alg:InitHierarchy}
    \begin{algorithmic}
        \STATE $i_1 = 0$
        \STATE $i_2 = 0$
        \STATE $\hat{\Delta}_{0} = \Delta_{\max}$
        \STATE $\mathcal{G}_{0} = \{G\}$
        \STATE $\mu = \floor{\tfrac{\eps}{128\log n}}$
         \WHILE{$\hat{\Delta}_{i_1+i_2} \geq \ceil{\tfrac{10^7 \log^5 n}{\eps^{3}}}$}
            \IF{$\tfrac{\Delta_{\max}}{\floor{2^{\sqrt{\log n}}}^{i_1-1}} \geq \floor{2^{10 \sqrt{\log n}}}$}
                \STATE $t(i_1+i_2) = t_1$
                \STATE $i_1 = i_1 + 1$
                \STATE $\hat{\Delta}_{i_1+i_2} = (1+\mu)^{i_1+i_2}\tfrac{\Delta_{\max}}{t_1^{i_1} t_2^{i_2}}$ 
            \ELSE
                \STATE $t(i_1+i_2) = t_2$
                \STATE $i_2 = i_2 + 1$
                \STATE $\hat{\Delta}_{i_1+i_2} = (1+\mu)^{i_1+i_2}\tfrac{\Delta_{\max}}{t_1^{i_1} t_2^{i_2}}$
            \ENDIF
        \ENDWHILE
        \STATE $h = i_1 + i_2$
    \end{algorithmic}
\end{algorithm2e} 
\begin{algorithm2e}[H]
    \caption{\texttt{Insert-Hierarchy}$(uv,j)$}
    \label{alg:insertionHierarchy}
    \begin{algorithmic}
        \STATE $\texttt{insert(uv)}$ into $G_j$ using Theorem~\ref{thm:tsplit} with $(\Delta_{\max},t,\eps') = (\hat{\Delta}_{|j|},t(|j|),\mu)$
        \FORALL{recoloured edges $xy \neq uv$ in $G_j$}
            \STATE \texttt{Delete-Hierarchy}$(xy,\chi(xy)^{(|j|)}_{\text{old}},\chi(xy)_{\text{old}})$
            \STATE \texttt{Insert-Hierarchy}$(xy,[j,c_{G_j}(xy)])$
        \ENDFOR
        \STATE \texttt{Insert-Hierarchy}$(uv,[j,c_{G_j}(uv)])$
    \end{algorithmic}
\end{algorithm2e}
\begin{algorithm2e}[H]
    \caption{\texttt{Delete-Hierarchy}$(uv,j,\zeta)$}
    \label{alg:DeletionHierarchy}
    \begin{algorithmic}
        \STATE $\texttt{delete(uv)}$ from $G_j$ using Theorem~\ref{thm:tsplit} with $(\Delta_{\max},t,\eps') = (\hat{\Delta}_{|j|},t(|j|),\mu)$
        \IF{$|\zeta| \geq |j|+1$}
            \STATE \texttt{Delete-Hierarchy}$(uv,\zeta^{(|j|+1)},\zeta)$
        \ENDIF
        \texttt{Delete-Hierarchy}$(uv,\zeta^{(|j|)},\zeta)$
        \FORALL{recoloured edges $xy \neq uv$ in $G_j$}
            \STATE \texttt{Delete-Hierarchy}$(xy,\chi(xy)^{(|j|)}_{\text{old}},\chi(xy)_{\text{old}})$
            \STATE \texttt{Insert-Hierarchy}$(xy,[j,c_{G_j}(xy)])$
        \ENDFOR
    \end{algorithmic}
\end{algorithm2e} 
\begin{remark}\label{rmk:parset}
    We remark that to actually perform an insertion, one calls $\texttt{Insert-Hierarchy}(uv,\emptyset)$, and to actually perform a deletion, one calls $\texttt{Delete-Hierarchy}(uv,\emptyset,\chi(uv))$. 
\end{remark}

\subsection{Correctness}
In this section, we show that the algorithms from above actually maintain a $(G,\Delta_{\max}, \eps)$-valid-hierarchy with parameters that give the guarantees stated in Theorem~\ref{thm:splitHier}. 
We first show that the algorithm maintains a $(G,\Delta_{\max}, \eps)$-valid-hierarchy:
\begin{lemma} \label{lma:correctnessHierarchy}
    Algorithms~\ref{alg:InitHierarchy}~\ref{alg:insertionHierarchy}~\ref{alg:DeletionHierarchy} implement a $(G,\Delta_{\max}, \eps)$-valid-hierarchy.
\end{lemma}
\begin{proof}
    The proof is a straight-forward verification that the hierarchy returned satisfies the items in Definition~\ref{def:hierarchy}. To that end, we note that $\texttt{Initialize-Parameters}(\Delta_{\max}, \eps)$ initialises in accordance with item 1. This part is never changed, and so through-out all updates Item 1 holds.
    
    Furthermore, $\texttt{Initialize-Parameters}(\Delta_{\max}, \eps)$ ensures that the sequence $\{t(i)\}_{i = 0}^{h-1}$ is in accordance with the second item of Definition~\ref{def:hierarchy}. 
    
    The hierarchy is also in accordance with the third, fifth, and sixth items. Indeed, if at all times the conditions of Theorem~\ref{thm:tsplit} is met, then it follows by induction and the theorem itself that the computed colourings are $t(i)$-splitting for all $i$, and furthermore satisfy the fifth and sixth items. 
    
    To check that this is the case, we check that the upper bounds on the degrees $\hat{\Delta}_{i}$ remain valid at all times during an update. This holds since the algorithms immediately delete all further copies of an edge upon a deletion or recolouring. 
    Indeed, suppose the graph $G_v$ is a graph with $|v|$ minimum that invalidates the upper-bound $\hat{\Delta}_{|v|}$. 
    Then it must be because the colouring on $G_{v^{(|v|-1)}}$ is not a $t(|v|-1)$-splitter, or because $G_{v}$ at some point is not a subgraph of $G_{v^{(|v|-1)}}[v_{|v|}]$. 
    The first case cannot happen, since the $t(|v|-1)$-splitter on $G_{v^{(|v|-1)}}$ always has all requirements satisfied by assumption. 
    The second case cannot happen, since any time an edge coloured $v_{|v|}$ is deleted or recoloured in $G_{v^{(|v|-1)}}$, it is immediately removed from $G_v$. 
    Hence, any time that $G_v$ is subjected to an insertion, it is a subgraph of $G_{v^{(|v|-1)}}[v_{|v|}]$ and thus satisfies the degree-bound. 

    Since we construct $\mathcal{G}_{i+1}$ by applying Theorem~\ref{thm:tsplit} to every graph in $\mathcal{G}_{i}$, and since Theorem~\ref{thm:tsplit} returns a valid $t(i)$-splitter, we thus also recover Item $4$. 
\end{proof}
We will use the following standard estimate repeatedly:
\begin{claim} \label{lma:est}
    For $0< \eps < 1$ and $0< \mu \leq \tfrac{\eps}{128\log n}$, we have that $(1+\mu)^{4\log n} \leq (1+\tfrac{\eps}{16})$
\end{claim}
\begin{proof}
    It follows from the well-known facts that $1+r \leq e^{r}$ for all $r \in \mathbb{R}$ and that $(1+\tfrac{x}{q})^{q} \leq e^{x}$ for all reals $x, q \neq 0$. Indeed, we can write;
    \[
     (1+\mu)^{4\log n} \leq (1+\tfrac{\eps}{128\log n})^{4\log n} \leq e^{\tfrac{\eps}{32}}
    \]
    Furthermore, for $r > -1$ it is well-known that $\tfrac{r}{1+r} \leq \log(1+r)$ and so we have
    \[
    (1+\mu)^{4\log n} \leq e^{\tfrac{\eps}{16(1+\tfrac{\eps}{16})}} \leq (1+\tfrac{\eps}{16})
    \]
    as claimed. 
\end{proof}
Next, we show that:
\begin{lemma} \label{lma:hbound}
    For any graph $G$ with at least one edge and any upper bound $\Delta_{\max} \leq n$, we have that any  $(G,\Delta_{\max}, \eps)$-valid-hierarchy of depth $h$ must have $h \leq 4\sqrt{\log n}$.
\end{lemma}
\begin{proof}
    Observe that we can perform at most $i \leq 2\sqrt{\log n}$ splits with $t = t_1$. Indeed, after $i$ splits, we have
    $\tfrac{\Delta_{\max}}{t_1^{i}} \leq \tfrac{n}{t_1^{i}}$, so setting $i > 2\sqrt{\log n}$ forces $\tfrac{\Delta_{\max}}{t_1^{i-1}} < \floor{2^{10 \sqrt{\log n}}}$ and so $\Delta_{\max} < \floor{2^{10 \sqrt{\log n}}} \cdot{} \floor{2^{\sqrt{\log n}}}^{i-1}$. 
    This means that $t(j) \neq t_1$ for all $j \geq i$.

    A similar argument shows that we perform at most $2\sqrt{\log n}$ splits with $t = t_2$ before $\tfrac{2^{10 \sqrt{\log n}}}{t_2^{i-1}}$ goes below $1/2$. By Lemma~\ref{lma:est}, we have $(1+\mu)^{4\sqrt{\log n}} < 2$, and so $\hat{\Delta}_{4\sqrt{\log n}} < 1$, contradicting the fact that $\hat{\Delta}_{4\sqrt{\log n}}$ must upper bound the maximum degree of all graphs in $\mathcal{G}_{4\sqrt{\log n}}$, one of which must have degree at least one, since $\mathcal{G}_{2\sqrt{\log n}}$ partitions $G$. 
\end{proof}
Finally, we show that any $(G,\Delta_{\max}, \eps)$-valid-hierarchy with $\Delta_{\max} \geq \tfrac{10^7 \log^5 n}{\eps^{3}}$ satisfies the conditions of Theorem~\ref{thm:splitHier}. We note that if $\Delta_{\max} < \tfrac{10^7 \log^5 n}{\eps^{3}}$, we do not need to apply the hierarchy, as we can simply apply the techniques developed in Section~\ref{sec:lowCol} directly on the graph. 
\begin{lemma} \label{lma:guaranteesHierarchy}
    Given any $(G,\Delta_{\max}, \eps)$-valid-hierarchy with $\Delta_{\max} \geq \tfrac{10^7 \log^5 n}{\eps^{3}}$. Let $i$ be the smallest integer $i$ for which $\tfrac{\Delta_{\max}}{t_1^{i}} < \floor{2^{10 \sqrt{\log n}}}$. 
    Similarly, let $j$ be the smallest integer $j$ for which $\tfrac{\Delta_{\max}}{t_1^{i}t_2^{j}} < \tfrac{10^7 \log^5 n}{\eps^{3}}$. Then, if the hierarchy has depth $i+j$, we have:
    \begin{enumerate}
        \item $|\mathcal{G}_{i+j}| \leq t_1^{i}t_2^{j} $
        \item $\hat{\Delta}_{i+j} = (1+\mu)^{i+j}\tfrac{\Delta_{\max}}{t_1^{i}t_2^{j}} \leq (1+\tfrac{\eps}{16})\tfrac{\Delta_{\max}}{t_1^{i}t_2^{j}}$
    \end{enumerate}
\end{lemma}
\begin{proof}
    We show the first item first. Observe that we perform $i$ splits into at most $t_1$ graphs first, meaning that $|\mathcal{G}_{i}| \leq t_1^{i}$ by an easy induction on $i$, since $|\mathcal{G}_{0}| = 1$. A similar induction on $j$ yields the result, since this time each split is into $t_2$ graphs. 
    The proof of the first part of the second item is a completely synchronous induction on first $i$ and then $j$ using item 5 of Definition~\ref{def:hierarchy}. 

    The last in-equality of item $2$ follows by observing that $i+j \leq 4\sqrt{\log n}$ by Lemma~\ref{lma:hbound} and then applying Claim~\ref{lma:est}. 
\end{proof}
Lemma~\ref{lma:correctnessHierarchy} and Lemma~\ref{lma:guaranteesHierarchy} together establish correctness of the Algorithm. In order to prove Theorem~\ref{thm:splitHier}, we need now only analyse the running time and the recourse of the hierarchy. We do so in the next sub-section. 

\subsection{Analysis}
The following Lemma analyses the recourse the algorithm for maintaining the hierarchy:
\begin{lemma} \label{lma:hierarchyRecourse}
    Starting from an empty graph, a sequence of $\sigma$ insertion and deletions to graphs at level $i$ triggers at most $\sigma \paren{2^{12} \eps^{-2} \log^{} n}^{j-i}$ insertions and deletions to graphs at level $j \geq i$. 
\end{lemma}
\begin{proof}
    The proof is by strong induction on $j-i$. The base case where $j-i = 0$ clearly holds. We proceed to the induction step. 
    Assume that for $j-i = 0$, $j-i = 1$, \dots, $j-i = k$ the lemma holds. Now assume $j-i = k + 1$. 
    A deletion or an insertion can happen in layer $j$ only if an insertion or deletion is performed at this or an earlier layer, or if a recolouring is performed at an earlier layer.
    
    By the induction hypothesis, the number of insertions and deletions, $U_1$, at level $j$ resulting from an insertion or an deletion at an earlier layer can be bounded as follows by applying the induction hypothesis:
    \[
    U_1 \leq \sigma \paren{2^{12} \eps^{-2} \log n}^{j-i-1} \leq \tfrac{\sigma}{2}\paren{2^{12}\eps^{-2} \log n}^{j-i}
    \]
    
    It remains to account for the number of insertions or deletions, $U_2$, caused by re-colourings at level $j-1$. 
    Note that any recolourings at level $j-2$ or higher results in deletion at level $j-1$, and are therefore accounted for in $U_1$.
    Observe that by Lemma~\ref{lma:tsplitterRecourse} it follows that over any sequence of $\sigma'$ insertions and deletions to a graph at level $j-1$, at most $\sigma' \cdot{} \tfrac{2^{10}\log n}{\eps^2}$ recolourings are performed at level $j-1$. Each recolouring corresponds to one insertion and one deletion on the next level, and so we find that at most 
    \[
    U_2 \leq 2\cdot{} \sigma \paren{2^{12}\eps^{-2} \log^{} n}^{j-1-i} \cdot{} \tfrac{2^{10}\log n}{\eps^2} \leq \tfrac{\sigma}{2}\paren{2^{12} \eps^{-2} \log^{} n}^{j-i}
    \]
    insertions or deletions are performed at level $j$ as a result of re-colourings at level $j-1$. Summing $U_1$ and $U_2$ finishes the induction. 
\end{proof}
The following lemma is then immediate:
\begin{lemma}
    The amortised recourse of an insertion or deletion to $G$ is in $2^{\tilde{O}_{\log \eps^{-1}}(\sqrt{\log n})}$.
    Furthermore, the amortised update time is also in $2^{\tilde{O}_{\log \eps^{-1}}(\sqrt{\log n})}$.
\end{lemma}
\begin{proof}
    Starting from any empty graph undergoing a sequence of $\sigma$ updates Lemma~\ref{lma:hbound} and Lemma~\ref{lma:hierarchyRecourse} together imply that at most $\sigma \cdot{} 2^{\tilde{O}_{\log \eps^{-1}}(\sqrt{\log n})}$ insertion and deletions are performed at level $\ell$ for any $\ell \leq h$, since the hierarchy maintained has depth $h \leq 4 \sqrt{\log n}$. An insertion or deletion triggers at most $O(\tfrac{\log n}{\eps^2})$ re-colourings, so at most $\sigma \cdot{} 2^{\tilde{O}_{\log \eps^{-1}}(\sqrt{\log n})}$ recourse happens at each level.
    Hence, over the sequence of $\sigma$ updates the total recourse is bounded by $\sigma \cdot{} 2^{\tilde{O}_{\log \eps^{-1}}(\sqrt{\log n})}$. 

    To recover the second part of the lemma, we know by above that, starting from an empty graph, any sequence of $\sigma$ operations trigger at most $\sigma \cdot{} 2^{\tilde{O}_{\log \eps^{-1}}(\sqrt{\log n})}$ insertion and deletions. 
    Hence, for the usual index set $I$, we know that if $G_v$ experience $\sigma_v$ insertion or deletions, then 
    \[
     \sigma \cdot{} 2^{\tilde{O}_{\log \eps^{-1}}(\sqrt{\log n})}= \sum_{v \in I} \sigma_{v} 
    \]
    By Theorem~\ref{thm:tsplit}, we can process a sequence of $\sigma_{v}$ insertions and deletions to a graph of level $|v|$ in $c\sigma_{v} \log(m) t_1^{3}$ update time for some constant $c$. 
    As a result, the total update time is 
    \[
    \sum_{v \in I} c\sigma_{v} \log(m) t_1^{3} = c \log(m) t_1^{3}  \sigma \cdot{} 2^{\tilde{O}_{\log \eps^{-1}}(\sqrt{\log n})} = \sigma \cdot{} 2^{\tilde{O}_{\log \eps^{-1}}(\sqrt{\log n})}
    \]
    and so the amortised update time is as claimed. 
\end{proof}
Finally, we note that Lemma~\ref{lma:hbound} shows that the $\texttt{while}$-loop of $\texttt{Initialize-Parameters}(\Delta_{\max}, \eps)$ has few iterations. 
Secondly, each iteration of the $\texttt{while}$-loop can be implemented efficiently provided that all values are properly discretised as noted in Remark~\ref{rmk:parset}.
The time spent on this step is therefore only $O(\sqrt{\log n})$. This step only has to be performed once upon initialization.
Combining all of the above now yields Theorem~\ref{thm:splitHier}.

\section{Deterministic Dynamic Edge-Colouring of Low-Degree Graphs} \label{sec:lowCol}

In this section, we show how to \emph{deterministically} maintain a $(\Delta_{\max} + 1)$-edge-colouring in a dynamic graph with maximum degree $\Delta \leq \Delta_{\max}$. 
In Section~\ref{sec:dynDelta}, we will show how to adjust the algorithm to the case where no upper bound on $\Delta$ is needed, at the cost of the algorithm using $(1+\eps)\Delta$ colours instead.
We will follow the approach sketched in the informal overview. We begin by showing some theoretical and structural results, which form the basis of the algorithm:

\subsection{Structural results}
\input{structure.tex}

\subsection{Contents of the data structures}
We will describe the data structures in terms of the following graph: 
\begin{definition} \label{def:skeleton}
    Let $G$ be a graph with maximum degree $\Delta \leq \Delta_{\max}$, let $c$ be a proper (partial) $(\Delta_{\max}+1)$-edge-colouring of $G$, and let $\ell$ be a parameter. The \emph{bichromatic skeleton} of $G$ is the directed graph $\mathcal{B}(G,c,\ell)$ with:
    \begin{itemize}
        \item Vertex-set $V(\mathcal{B}(G,c,\ell)) = X_1 \cup X_2$, where every maximal bichromatic path $P$ of $G$ under $c$ is represented by a node in $X_1$ and every node in $V(G)$ is represented by a node in $X_2$. Note that we allow $P$ to consist of an isolated vertex. 
        \item Edge-set $E(\mathcal{B}(G,c,\ell)) = Y_1 \cup Y_2 \cup Y_3$. Here:
        \item $(x_{2}x_{1}) \in Y_1$ with $x_1 \in X_1$ and $x_2 \in X_2$ if $x_2$ is an endpoint of the bichromatic path represented by $x_1$.
        \item $(xy) \in Y_2$ with $x,y \in X_2$ if the nodes in $G$ represented by $x,y$ in $\mathcal{B}(G,c)$ are neighbours via a coloured edge in $G$. 
        \item $(x_{1}x_{2}) \in Y_3$ with $x_1 \in X_1$ and $x_2 \in X_2$ if $x_2$ belongs to the vertex set of the bichromatic path represented by $x_1$, and, furthermore, the distance in the bichromatic path, represented by $x_1$, from $x_2$ to an endpoint of the bichromatic path is at most $\ell-1$.
    \end{itemize}
\end{definition}
If $c$ and $\ell$ are clear from the context, we will often leave the arguments out and simply write $\mathcal{B}(G)$.
Let $\mathcal{P}(G, c, \ell)$ or simply $\mathcal{P}$, if the arguments are clear from the context, be the set containing all bichromatic paths $P$ that are contained in some maximal bichromatic path $P'$ of $G$ under $c$ such that $P$ has length at most $\ell$ and such that $P \cap P'$ contains an endpoint of $P'$.
Then there is a canonical map $r: \mathcal{P} \cup V(G) \rightarrow V(\mathcal{B}(G))$ that maps a vertex in $V(G)$ to its representative in $X_2$, and similarly a bichromatic path $P$ to the representative in $X_1$ of the unique maximal bichromatic path $P'$ containing $P$.
Below we state some very basic and straight-forward properties of this type of graphs:
\begin{claim} \label{clm:standard}
    Let $G$, $\Delta \leq \Delta_{\max}$, $c$, and $\ell$ be as in Definition~\ref{def:skeleton}. Then for $\mathcal{B}(G)$ the following holds:
    \begin{itemize}
        \item $\mathcal{B}(G)$ contains no parallel edges pointing in the same direction.
        \item The maximum in-degree of $\mathcal{B}(G,c)$ is upper bounded by $2(\Delta_{\max}+1)^2$.
        \item The maximum out-degree of $v \in X_2$ is upper bounded by $2(\Delta_{\max}+1)^2$.
    \end{itemize}
\end{claim}
\begin{proof}
    It is straight-forward to check that for any $u,v \in V(\mathcal{B}(G))$ at most one edge $(uv)$ can be inserted, as $G$ is simple. 

    The second item follows from the fact that at most $(\Delta_{\max}+1)^{2}$ bichromatic paths can go through any vertex of $G$. More specifically, if $v \in X_1$ then at most $2$ in-edges can be incident to $v$, and if $v \in X_2$, then at most $(\Delta_{\max}+1)^{2}$ in-edges of $v$ can belong to $Y_3$ and at most $(\Delta_{\max}+1)$ in-edges of $v$ can belong to $Y_2$. In this case, no in-edges belong to $Y_1$. 

    Similarly, for the third item, we note that for $v \in X_2$ at most $(\Delta_{\max}+1)$ out-edges of $v$ belongs to $Y_2$ and at most $(\Delta_{\max}+1)^{2}$ out-edges of $v$ belongs to $Y_1$, since $v$ can be the endpoint of at most $(\Delta_{\max}+1)^{2}$ maximal bichromatic paths. No out-edges of $v$ belongs to $Y_3$.
\end{proof}
Finally, we note that any consistent $i$-step process with step-length at most $\ell$ corresponds to a directed path in $\mathcal{B}(G)$ in a canonical way: simply map $w_1+y_1+P_1 + \dots + w_{i} + y_{i}+P_{i}$ to the path $(r(w_1),r(y_1), r(P_1), \dots, r(w_i), r(y_i), r(P_i))$. 
Note that we will never revisit a representative of a path, since different steps are not allowed to be part of the same maximal bichromatic path.

Again let $G$, $\Delta \leq \Delta_{\max}$, $c$, and $\ell$ be as in Definition~\ref{def:skeleton}. We will maintain data structures containing the following information: 
\begin{itemize}
    \item We explicitly maintain the graph $\mathcal{B}(G)$.
    \item For each edge $e = (x_2,x_1) \in Y_1$, we maintain for all $i \leq \sqrt{\log n}$ and for all neighbours $w$ of $y_1 \in r^{-1}(x_2)$, the following information in $\mathcal{D}_{i}(e,w)$: 
    either a non-empty set $\mathcal{S}^{i}[e,w]$ certifying that $w$ is $(\tfrac{a}{2},i,P,y_1)$-spread (with $P$ maximal bichromatic path containing all paths in $r^{-1}(x_1)$) or an empty set $\mathcal{S}^{i}[e,w]$. 
    Importantly, we require that if $w$ is $(i,a,\emptyset, \emptyset, x_1, y_1)$-heavy, then the set $\mathcal{S}^{i}[e,w]$ is non-empty.
    In particular, it follows from Theorem~\ref{thm:main} that all sets of the form $\mathcal{S}^{\sqrt{\log n}}[e,w]$ are non-empty. 
\end{itemize} 
In the next section, we will describe how to efficiently maintain the above information under colouring and un-colouring of edges. This is sufficient to construct a dynamic algorithm, since the above data structures readily contain information about which edges to recolour in order to extend the maintained colouring. 

\subsection{Maintaining the data structures} \label{sec:algos}
In this section, we will consider two main operations: $\texttt{colour}(e,\kappa)$ and $\texttt{un-colour}(e)$. Here $\texttt{colour}(e,\kappa)$ will extend the maintained colouring $c$ by giving the edge $e$ the colour $\kappa$. We require that $\kappa$ is free at both endpoints of $e$ for $\texttt{colour}(e,\kappa)$ to be called. 
Oppositely, the operation $\texttt{un-colour}(e)$ will reduce the colouring $c$ by assigning $e$ a blank colour $\neg$. 
In addition to these operations, we will also support $\texttt{add}(e)$ which inserts a new edge $e$ initially coloured $\neg$, and $\texttt{remove}(e)$ which removes a blankly coloured edge $e$ from the graph. 
These last two operations are at this point straight-forward to maintain, as all the work can be pushed to the two main operations. 
As mentioned earlier, we will again assume that at all times $\Delta(G) \leq \Delta_{\max}$. 
We first give an informal overview of the data structures, and then we provide formal description afterwards. 

\paragraph{Maintaining the first item:} In order to explicitly maintain $\mathcal{B}(G)$, we will consider the ${\Delta_{max+1} \choose 2}$ graphs induced in $G$ by any pair of distinct colours $\kappa_1, \kappa_2$. 
Since the colouring $c$ is at all times proper, these graphs consist only of isolated points, paths, and even cycles.
In order to dynamically maintain the representation, we maintain each component in a Top-tree~\cite{HolmEtAL}. 
In the case of even cycles, we cannot store the component directly in a dynamic tree data structure, but in that case we store the entire component except for one arbitrary edge in a Top-tree. 
As ancillary information, we store the number of degree $1$ vertices in the component. Note that there can be only $0$, $1$, or $2$ vertices of degree $1$. 
We use this information to explicitly store $\mathcal{B}(G)$:
whenever a maximal bichromatic path is changed, we remove the corresponding node and all its related edges. Next we insert representatives of any un-represented maximal bichromatic paths, and add the edges incident to these representatives as specified by Definition~\ref{def:skeleton}. 

\paragraph{Maintaining the second item:} We first consider the operation $\texttt{un-colour}(e)$. 
Whenever $e$ is un-coloured some of the information stored in the data structures is no longer valid. 
In particular, for every non-empty maximal bichromatic path $P$ that previously contained $e$, any $i$-step process that used $P$ might no longer be valid.  
In a similar vein, some extension close to $e$ could stop being consistent.
We can mark every data structure that might be affected by this as follows: any vertex that can reach $P$ via a $j$-step process, could contain faulty information in their data structures for $i \geq j $. 
We will later bound this number using Claim~\ref{clm:standard} as well as the fact that $e$ belongs to at most $(\Delta_{\max}+1)^2$ such $P$'s.

Next, we will cleanse the marked data structures and recompute their contents from scratch. We will re-compute the marked data structures using the invariant that all data structures of the form $\mathcal{D}_{i-1}$ should be recomputed before the first data structure of the form $\mathcal{D}_{i}$ is recomputed. 
Performing the re-computation in the case where $i = 1$ is straight-forward as there only is one choice for the $1$-step process. 
Recomputing structures of the form $\mathcal{D}_{i}$ assuming that all structures of the form $\mathcal{D}_{j}$ with $j \leq i-1$ are correct, now amounts to making Lemma~\ref{lma:fuse} algorithmical. 
In order to update the contents of $\mathcal{D}_{i}(e,w)$ with $e = (x_2x_1)$, we consider the bichromatic path $x_1$. 
If $x_1$ is short enough (i.e.\ it has length at most $\ell + 1$), we simply set $\mathcal{S}[e,w]$ to be the set containing the augmenting stepping process $w+y_1+x_1$ for $y_1 \in r^{-1}(x_2)$.
Otherwise, for the $\ell$ first points along $x_1$, beginning from $y_1 \in r^{-1}(x_2)$, we calculate the correct extension according to Definition~\ref{def:sc} item $1)$, look up the relevant set $\tilde{\mathcal{S}}$ that is either empty or certifies that the extension point is $(\tfrac{a}{2},i-1,\tilde{P})$-spread for some $\tilde{P}$.
Then we determine if the point is blocked according to the definition from Lemma~\ref{lma:fuse}.
Finally, we count the number of non-blocked extensions with non-empty sets $\tilde{\mathcal{S}}$ along the first $\ell$ points of $x_1$ beginning at $y_1$. 
If we have at least $\tfrac{a}{2}$ such points, we arbitrarily pick $\tfrac{a}{2}$ of them. For each picked point, as explained in the proof of Lemma~\ref{lma:fuse}, we pick some non-blocked augmenting and semi-consistent process of length at most $i-1$, and extend it to an augmenting and semi-consistent process of length at most $i$. 
These processes together will form the non-empty set $\mathcal{S}[e,w]$. 
In the case, where there are less than $\tfrac{a}{2}$ non-blocked points with non-empty associated sets, we simply leave $\mathcal{S}[e,w]$ empty. 
The key property of this update step is that Lemma~\ref{lma:fuse} and the subsequent discussion ensures that $\mathcal{S}[e,w]$ is non-empty if $w$ is heavy. 

Performing $\texttt{colour}(e,\kappa)$ is now completely symmetric to the above. 
Assigning $e$ a colour $\kappa$ also causes some data structures to, potentially, be out-dated. These can now be fixed similarly to before. 

Next we give a more formal description of the data structures:

\paragraph{Maintaining a bichromatic skeleton:} 
We will let $F_{\kappa_1,\kappa_2}$ denote the pseudo-forest induced in $G$ by all edges coloured $\kappa_1$ or $\kappa_2$. For some vertex $v \in G$, we let $PT_{\kappa_1,\kappa_2}(v)$ be the pseudo-tree in $F_{\kappa_1,\kappa_2}$ containing $v$. 
We store every $PT_{\kappa_1,\kappa_2}$ as a tree $T_{\kappa_1,\kappa_2}$ stored as a Top-tree~\cite{HolmEtAL} and possibly one extra edge stored in a special slot $S_{\kappa_1,\kappa_2}$.
We denote by $T_{\kappa_1,\kappa_2}(v)$ the Top-tree currently containing $v$. Note that we will maintain the invariant that $S_{\kappa_1,\kappa_2} = xy$ only if $x,y$ already belong to the same Top-tree i.e.\ if $T_{\kappa_1,\kappa_2}(v) = T_{\kappa_1,\kappa_2}(u)$.
To perform \texttt{BCS-colour}$(uv, \kappa)$, we do as follows:
\begin{itemize}
    \item For all $\kappa' \neq \kappa$, we check if $T_{\kappa,\kappa'}(v) = T_{\kappa,\kappa'}(u)$. 
    If this is not the case, then we simply call $\texttt{link}(u,v)$ on the Top-tree data structures.
    Otherwise, we store $uv$ in the special slot $S_{\kappa,\kappa'}$.
    \item For each Top-tree we store as ancillary information the number of vertices with degree $1$ or $0$ contained in each cluster, not counting the boundary vertices. It is straight-forward to maintain this information under \texttt{join} and \texttt{split}. 
    \item After each modification to the Top-tree, we perform a non-local search for each leaf and check if they are incident to the edge stored in $S_{\kappa,\kappa'}$ (if any such edge). If so, the vertices' degrees are changed accordingly, and they might no longer be classified as vertices with degree $1$ or $0$. 
    \item If the Top-tree has either $1$ or $2$ such vertices, then it is classified as a bichromatic path. 
    \item Next, we remove every vertex (and incident edges) representing a bichromatic path whose Top-tree, or related special slot, was modified in the above steps. 
    \item Finally, we add representatives (and incident edges) of all bichromatic paths whose Top-tree, or related special slot, was modified in the above steps. 
    \item To add the incident edges, we access the leaves of the Top-tree to find their neighbours in $G$ as well as the first $\ell$ vertices (from either direction) on the leaf-to-leaf path in the Top-tree. Finally, we add the edges as prescribed by Definition~\ref{def:skeleton}. 
    \item Finally, we add parallel, oppositely directed edges between $r(u)$ and $r(v)$.
\end{itemize}
To perform \texttt{BCS-uncolour}$(uv)$ we only have to modify the first and the last step. The rest of the steps are exactly as before. 
Instead of calling $\texttt{link}(u,v)$, we call $\texttt{cut}(u,v)$ (or remove $uv$ from $S_{\kappa,\kappa'}$ if it stored there). 
Finally, if $S_{\kappa,\kappa'} = xy$ is non-empty, we call $\texttt{link}(x,y)$, and set $S_{\kappa,\kappa'} = \emptyset$. 
After this, we do the same clean-up step as we did in \texttt{BCS-colour}, except that instead of adding parallel, oppositely directed edges between $r(u)$ and $r(v)$, we remove any such edges. \\

\paragraph{Maintaining the $\mathcal{D}_{i}$'s:}
To perform \texttt{un-colour}$(uv)$ we do as follows: 
\begin{itemize}
    \item First we un-colour $uv$ in $G$ and update $\mathcal{B}(G)$ as described earlier. 
    \item Next, we mark some of the edges of $\mathcal{B}(G)$ with the following scheme: an edge $e$ is marked $j$ if the data structures $\mathcal{D}_{i}(e,w)$ for $j \leq i \leq \sqrt{\log n}$ needs to be updated. An edge might be marked multiple times, but we need only keep the lowest mark, it receives. 
    \item All new edges in $\mathcal{B}(G)$ from step $1$ are marked $1$. 
    \item We let $R$ be the set containing representatives of the $2$-hop neighbourhood of $u$ and $v$ in $G$, all representatives of endpoints of new bichromatic paths of $\mathcal{B}(G)$, and all representatives of former endpoints of no longer-existing bichromatic paths.
    \item All edges in $Y_1$ incident to $R$ are marked $1$. 
    \item We set $M_1 = R$. We will mark all in-edges of vertices belonging to the set $N_{Y_3, \text{in}}(N_{Y_2, \text{in}}(M_1))$ also belonging to $Y_1$ i.e.\ the set $E_{Y_1, \text{in}}(N_{Y_3, \text{in}}(N_{Y_2, \text{in}}(M_1)))$ with $2$. Finally, we set $M_2 = N_{Y_1, \text{in}}(N_{Y_3, \text{in}}(N_{Y_2, \text{in}}(M_1)))$. 
    \item For all $1<i\leq \sqrt{\log n}$, we mark all in-edges of vertices belonging to the set $N_{Y_3, \text{in}}(N_{Y_2, \text{in}}(M_i))$ belonging to $Y_1$ with $i+1$. Finally, we set $M_{i+1} = N_{Y_1, \text{in}}(N_{Y_3, \text{in}}(N_{Y_2, \text{in}}(M_i)))$.
    \item Having updated all data structures at level $<i$, we update $\mathcal{D}_{i}(e,w)$ as follows: 
    To update $\mathcal{D}_{i}(x_1x_2,w)$ for $i = 1$, check if the $1$-step process starting at $w_1 = r^{-1}(w)$ with $y_1 = r^{-1}(x_2)$ and $P_1 \in r^{-1}(x_1)$ is augmenting. If so, store the set $\mathcal{S}[e,w,i]$ containing this augmenting process.
    Note that it is crucial that this information can be efficiently computed without access to any other similar data structures. 

    For $i > 1$, we ensure that all of the data structures $\mathcal{D}_{j}(e,w)$ have been (re)computed for all $j < i$. 
    To update $\mathcal{D}_{i}(x_2x_1,w)$ we let $w_1 = r^{-1}(w)$ and $y_1 = r^{-1}(x_2)$. 
    We proceed as described earlier. If $w_1 + y_1 + x_1$ has length at most $\ell$, we conclude that it is both short and augmenting, and we simply store this process in $\mathcal{S}[e,w,i]$.
    
    Otherwise, we consider all valid choices of $P_1 \in r^{-1}(x_1)$ of length at most $\ell$.
    For each valid choice of $P_1$, there is a unique extension point $x = x(P_1)$ corresponding to the last vertex of $P_1$ visited, if one traverses $P_1$ starting at $y_1$. Note that $x$ is among the first $\ell$ vertices of the maximal bichromatic path containing $P_1$. For each choice of $P_1$ (or equivalently $x$), we do as follows:
    \begin{enumerate}
        \item Construct a fan consistent with Lemma~\ref{lma:fan2}. We let $P_2(x)$ resp.\ $y_2(x)$ be the be bichromatic path resp.\ the endpoint of $P$ consistent with Lemma~\ref{lma:fan2}. 
        \item Next we check if this extension is blocked in the sense of Lemma~\ref{lma:fuse}. 
        This is done by simply following every $(i-1)$-step process in $\tilde{\mathcal{S}}$ stored in $D_{i-1}((r(y_2(x)r(P_2(x))),x)$ checking the conditions.
        \item We repeat the above step for each choice of $x$, noting which extensions are blocked. 
        \item We pick $\tfrac{a}{2}$ arbitrary unblocked choices of $x$ with non empty associated sets $\tilde{\mathcal{S}}$. For each choice, we pick one non-blocked stepping process (represented by a string in the set stored for $x$), say $x+y_2+P_2+ \dots + w_j + y_j + P_j$, increase it to $w+y+P_1 + x+y_2+P_2+ \dots + w_j + y_j + P_j$, and represent it as a string. 
        We collect all such strings together to form $\mathcal{S}[i,e,w]$. 
        \item Finally, we un-mark every marked edges.
    \end{enumerate}
\end{itemize}
In order to perform \texttt{colour}$(uv, \kappa)$, we do almost the same as above. 
The only change is that instead of un-colouring $uv$ in $G$, we colour it in $G$ with the colour $\kappa$. Then we update $\mathcal{B}(G)$ accordingly, and then proceed exactly as before. 

\subsection{Analysing the data structures}
We start out by analysing the data structure for maintaining a bichromatic skeleton:
\begin{lemma} \label{lma:skeletonCorrect}
    Let $G$ be a dynamic graph with maximum degree $\Delta \leq \Delta_{\max}$ at all times, and let $\ell$ be a parameter. Let $c$ be a $(\Delta_{\max}+1)$-edge-colouring of $G$. Then we can maintain the bichromatic skeleton $\mathcal{B}(G,c,\ell)$ under the operations \emph{\texttt{add}}$(e)$, \emph{\texttt{remove}}$(e)$, \emph{\texttt{un-colour}}$(e)$, and \emph{\texttt{colour}}$(e,\kappa)$ in $\tilde{O}(\Delta_{\max}^2+\Delta_{\max}\cdot{} \ell \cdot{} \log(n))$ worst-case update time per operation. 
\end{lemma}
\begin{proof}
    We have already described the algorithm in Section~\ref{sec:algos}, so what remains is to analyse it and prove that it is correct. To show correctness, we have to prove two things. First of all that the maintained bichromatic skeleton has the required vertex set, and second of all that it has the required edge-set. 
    Since $X_2$ cannot be updated, we need only check that $X_1$ is correctly maintained. 
    Since $\texttt{add}(e)$ and $\texttt{remove}(e)$ does not alter $X_1$, we need only check that $\texttt{un-colour}(e)$ and $\texttt{colour}(e, \kappa)$ correctly re-establishes $X_1$. 
    We will argue that $\texttt{colour}(e, \kappa)$ does so, and note that the case for $\texttt{un-colour}(e)$ is completely symmetrical. 

    After colouring $e$, the only new vertices of $X_1$ that need to be added are representatives of bichromatic paths containing $e$. 
    Any bichromatic component $C$ containing $e$ can only contain edges with colours $\kappa$ and $\kappa'$ for some $\kappa' \neq \kappa$.
    Since the colouring maintained is forced to be proper, it follows that $C$ must be either an even cycle or a path (possibly of length 0). 
    In particular, $C$ can always be represented by a pseudo-tree or a tree with one additional special edge. Therefore, the algorithm always correctly identifies the number of of vertices in $C$ of degree $1$ or $0$. 
    Hence, the algorithm will always identify if $C$ is a (possibly empty) bichromatic path, and correctly insert a representative of $C$ into $\mathcal{B}(G)$. So $X_1$ contains a representative of all current bichromatic paths. 
    The algorithm also deletes any bichromatic path whose Top-tree representation was modified during the process above. Since $e$ is inserted into the representation of $G[\kappa, \kappa']$ for all $\kappa'$, either the Top-tree or the related special slot representing the component of $e$ in $G[\kappa, \kappa']$ will be modified. 
    In particular, either the Top-tree or the related special slot representing bichromatic paths no longer present will be modified and hence deleted. Therefore, $X_1$ is maintained correctly. 

    To see that the edge set is also correctly maintained, we note that the algorithm inserts the new edges in $Y_2$ as well as any relevant edges incident to new bichromatic paths. These are the only new edges that need to be added. It also removes any edges going to representatives no longer present. 

    To recover the update times, we note that any operation to the Top-trees can be performed in $O(\log n)$ time (see~\cite{HolmEtAL}), and that any non-local search for edges or degree $1$ or $0$ vertices can be done in $O(\log n)$ time~\cite{HolmEtAL} per edge. 
    It is also standard that the number of such vertices can be maintained with only $O(1)$ overhead. 
    Since we update or search only $O(\Delta_{\max})$ Top-trees, the total time spend searching in the Top-trees is in $O(\Delta_{\max}\cdot{}\ell\cdot{}\log n)$. 
    Furthermore, for each new (and old) bichromatic path, of which there are at most $O(\Delta_{\max})$, we have to add (and delete) $2(\Delta_{\max}+1)$ edges to (from) $Y_1$ and $2\ell$ edges to (from) $Y_3$. Each edge is easy to insert (or delete) in constant time after performing the searches in the Top-trees already accounted for. Hence this takes additional time at most $O(\Delta_{\max}^2+\Delta_{\max}\cdot{} \ell \cdot{} \log n)$ and the lemma follows. 
\end{proof}
\begin{remark}
    We remark that it is straight-forward to allow insertion and deletion of isolated vertices into $G$ in $O(1)$ update time per insertion/deletion. 
\end{remark}
Next we analyse the data structure for maintaining the $\mathcal{D}_{i}$'s. We will show the following Theorem: 
\begin{theorem} \label{thm:colourStructures}
    Let $G$ be a dynamic graph, and let $\Delta_{\max}$ be a known upper-bound on the maximum degree throughout the entire update sequence. 
    Then, we can maintain the data structures $\mathcal{D}_{i}(\cdot{}, \cdot{})$ in $(\Delta_{\max}+1)^{\tilde{O}(\sqrt{\log n})}$ worst-case update time per operation. 
\end{theorem}
We have already described how to update the data structure in Section~\ref{sec:algos}. In the following, we show correctness and analyse the update times: 

\paragraph{Correctness: }
The algorithm needs two parameters $\ell$ and $a$, which we will specify later. The only important thing is that they satisfy the requirements of earlier presented lemmas. We will assume this is the case for now, and then fix the parameters later. 

First we show that any data structure whose contents are no longer valid, will indeed be properly marked:
\begin{lemma} \label{lma:markCorrectly}
    Assuming that all data structures $\mathcal{D}_{i}$ are initially valid, then after the marking process of a call to $\emph{\texttt{colour}}(e,\kappa)$ or $\emph{\texttt{uncolour}}(e)$ concludes, we have that: if $f$ is any edge and $w$ any vertex with data structure $\mathcal{D}_{i}(f,w)$ no longer valid, then $f$ is marked with $j\leq i$.
\end{lemma}
\begin{proof}
    Let $f = x_2x_1$, $w$ and $i$ be resp.\ an edge, a vertex and index such that $\mathcal{D}_{i}(f,w)$ is no longer correct. We will show that then $f$ is marked with $j \leq i$. 

    There are two reasons why $\mathcal{D}_{i}(f,w)$ might no longer be correct. 
    The first reason $A)$ is that it contained an $i$-step process with step-length $\leq \ell$ which used a changed bichromatic path and so is no longer valid. 
    The second reason $B)$ is that changes to the colouring $c$ causes some invocations of Lemma~\ref{lma:fan2} to be altered and so some of the stored processes are no longer semi-consistent. Of course both reasons could be present at the same time. Observe that it is sufficient to show that if any $i$-step process with step-length $\leq \ell$ beginning with $w+x_2+x_1$ can reach either 1) a changed bichromatic path in $(i-1)$-steps or less or $2)$ any vertex in the one-hop neighbourhood of $e$ in $(i-1)$-steps or less, then $f$ is marked with $j \leq i$. 
    Indeed, if $1)$ does not hold, then the reason $A)$ can never occur. Similarly, if $2)$ does not hold, then reason $B)$ cannot occur. 

    Thus it is sufficient to show that if $1)$ or $2)$ holds then $f \in E_{Y_1, \text{in}}(N_{Y_3, \text{in}}(N_{Y_2, \text{in}}(M_{j-1})))$ for some $2 \leq j \leq i$ or $f$ is incident to a vertex in $R$. 
    Indeed, then $f$ is marked with $j \leq i$.

    We begin by assuming that $1)$ holds.
    This means that after the changes to $c$, there is an $i$-step process $w_1+y_1+P_1+\dots + w_i + y_i + P_i$ with $w_1 = w$, $y_1 = r^{-1}(x_2)$, and $P_1 \in r^{-1}(x_1)$ using a changed bichromatic path $P$ during the $j^{\text{th}}$ step i.e.\ $P_j$ is a sub-path of $P$. 
    Let $j \leq i$ be the smallest number for which the above is true.

    This means that before the change to $c$, $w_1+y_1+P_1+\dots + w_j + y_j$ was also an initial segment of some $j$-step process.
    Hence, this initial segment corresponds to a directed path in $\mathcal{B}(G)$ of the form $r(w_1),r(y_1),r(P_1), \dots, r(w_j),r(y_j)$ present both before and after the update to $c$. This follows from the minimality of $j$ and the discussion directly after Claim~\ref{clm:standard}. 
    Observe that then $y_j \in R$, $(r(w_j)r(y_j)) \in Y_2$, $(r(P_{j-1})r(w_j)) \in Y_3$. 
    Hence, if $j = 1$, then $f = (x_2x_1) = (r(y_1)r(P_1))$ is incident to $R$, and so we would be done. 
    More generally note that for all $k \leq j-1$, we have that $(r(y_k)r(P_k)) \in Y_1$, $(r(w_k)r(y_k)) \in Y_2$, and if $k>1$ that $(r(P_{k-1})r(w_k)) \in Y_3$. 
    Thus if $\psi: \mathcal{P}(V(\mathcal{B}(G))) \rightarrow \mathcal{P}(V(\mathcal{B}(G)))$ maps a set $A \in \mathcal{P}(V(\mathcal{B}(G)))$ to the set $N_{Y_1, \text{in}}(N_{Y_3, \text{in}}(N_{Y_2, \text{in}}(A)))$, then $r(y_{k-1}) \in \psi(\{r(y_{k})\})$ for $1 < k \leq j$. 
    In particular, it follows by induction that $y_{1} \in \psi^{j-1}(\{y_j\})$. 
    So if $y_j \in R$, then $y_{1} \in M_{j}$ and therefore $f = (r(y_{1})r(P_{1})) \in E_{Y_1, \text{in}}(N_{Y_3, \text{in}}(N_{Y_2, \text{in}}(M_{j-1})))$. 

    On the other hand, assume that $2)$ holds. Then there is an $i$-step process $w_1+y_1+P_1+\dots + w_i + y_i + P_i$ with $w_1 = w$, $y_1 = r^{-1}(x_2)$, and $P_1 \in r^{-1}(x_1)$ such that for some $j$ we have $w_j = z$ for some $z \in N^{1}(e)$ in the coloured neighbourhood of $e$. I.e.\ we only consider neighbours of $z$ along edges that are coloured. 
    Indeed, otherwise the application of Lemma~\ref{lma:fan2} cannot be changed at $z$, as it depends only on the colours of edges incident to vertices in the $1$-hop neighbourhood of $z$. 
    WLOG we can assume that $j$ is the smallest such index among all choices.
    
    If $j = 1$, then we are now done as $y_1 \in N^{1}(z) \subset N^2(e)$, and therefore $f$ is incident to a vertex in $R$.
    For $j > 1$ we can conclude, similarly to case $1)$, that $r(y_{1}) \in \psi^{j-1}(R)$ and therfore that $f \in E_{Y_1, \text{in}}(N_{Y_3, \text{in}}(N_{Y_2, \text{in}}(M_{j-1})))$ and thusly marked with $j$. This concludes the proof of the lemma.
\end{proof} 
The algorithm maintain sets $S^{i}_{w,y,P}$ for all $i \leq \sqrt{\log n}$ in accordance with the description before Corollary~\ref{cor:setfuse}, so since the lemma above together with induction shows that these sets are properly updated, we can apply Corollary~\ref{cor:setfuse}. 
Hence, we conclude that for any valid choice of $w,y$, and $P$, we find that the set $S^{\sqrt{\log n}}_{w,y,P}$ will be non-empty. Finally, Lemma~\ref{lma:fuse} and the subsequent discussion proves that if $w$ is $(a,i, \emptyset, \emptyset, P)$-heavy, then the set $S^{i}_{w,y,P}$ is non-empty. This concludes correctness of the algorithm. 

In the following paragraph, we provide implementations of the data structures used by the algorithm and analyse the update time. 

\paragraph{Implementation and analysis:}
We first briefly comment on implementations certain data structures. For each BBST, we can use Red-blacks trees~\cite{cormen2022introduction} to support basic operations in $O(\log n)$ worst-case time. For Top-trees we use the version from~\cite{HolmEtAL}.
As for the analysis, we first show that not too many edges are marked: 
\begin{lemma} \label{lma:markedCount}
    In each invocation of the marking scheme, at most $2^{3i}(\Delta_{\max}+1)^{6i}|R|$ edges are marked with $i$. Hence $|M_{i}| \leq 2^{3i}(\Delta_{\max}+1)^{6i}|R|$.
\end{lemma}
\begin{proof}
    The proof is by induction on $i$. 
    For $i = 1$, we note that by item two and item three of Claim~\ref{clm:standard} at most $4(\Delta_{\max}+1)^2\cdot{} |R|$ edges are initially marked $1$. Also $|M_1| = |R|$.
    For the induction step, we assume the statement for all $j \leq i-1$.  We note that applying item two of Claim~\ref{clm:standard} thrice shows that $M_{i} \leq 2^3(\Delta_{\max}+1)^6|M_{i-1}|$ and that the number of edges marked with $i$ is also at most $2^3(\Delta_{\max}+1)^6|M_{i-1}|$. Applying the induction hypothesis now yields the result. 
\end{proof}
Next we show that $R$ is not too big:
\begin{lemma}\label{lma:RCount}
    In each invocation of the marking scheme, we have that $|R| \leq 2+2(\Delta_{\max}+1)^2+6(\Delta_{\max}^2+1)^2 \leq 10(\Delta_{\max}^2+1)^2$.
\end{lemma}
\begin{proof}
    We note that at most $3(\Delta_{\max}^2+1)^2$ bichromatic paths are altered after a removal or insertion of a coloured edge: namely 3 for each pair of colours; the bichromatic path that used to contain the edge and the two new bichromatic paths that this path was split into (for insertion it is of course opposite; one new bichromatic path that is the fusion of two old). 
    There are at most $(\Delta_{\max}+1)^2$ choices for such colour pairs. Note that each path has at most $2$ endpoints. 
    Finally, there are at most $2(\Delta_{\max}+1)^2$ vertices in the $2$-hop neighbourhood of $u$ and $v$.
    Adding these quantities gives the bound.
\end{proof}
These two lemmas allow us to analyse the time it takes to mark all edges:
\begin{lemma}
    It takes $\tilde{O}(2^{3(\sqrt{\log n}+4)}(\Delta_{\max}+1)^{6(\sqrt{\log n}+2)})$ time to mark all edges. 
\end{lemma}
\begin{proof}
    We can mark each edge with constant overhead by performing a backwards BFS search along the relevant edges (assuming we store different types of edges in different BBST's). 
    Hence, the Lemma follows by combining Lemma~\ref{lma:markedCount} with Lemma~\ref{lma:RCount}, and summing over all choices of $1 \leq i \leq \sqrt{\log n}$:
    \begin{align*}
        10(\Delta_{\max}+1)^{2}\sum \limits_{i = 1}^{\sqrt{\log n}} 2^{3i}(\Delta_{\max}+1)^{6i} &\leq 10(\Delta_{\max}+1)^{2}\cdot{}2^{3(\sqrt{\log n} + 1)}(\Delta_{\max}+1)^{6(\sqrt{\log n} + 1)} \\
        &\leq 2^{3(\sqrt{\log n} + 4)}(\Delta_{\max}+1)^{6(\sqrt{\log n }+ 2)}
    \end{align*}
\end{proof}
Next, we bound the time it takes to check whether an extension is overlapping:
\begin{lemma} \label{lma:checkExtension}
    Given a $1$-step process $w_1+y_1+P_1$, with $w_2$ the last vertex of $P_1$,
    and a semi-consistent and augmenting $(i-1)$-step process $w_2+y_2+P_2+\dots + w_{i}+y_{i}+P_{i}$, we can check if $C = w_1+y_1+P_1+\dots + w_{i}+y_{i}+P_{i}$ is a semi-consistent and augmenting $i$-step process in $\tilde{O}((i+1)\cdot{}\ell\cdot{}(\Delta_{\max}+1)^3)$ time. 
\end{lemma}
\begin{proof}
    We check whether $w_2, y_2$ and $P_2$ is consistent with Lemma~\ref{lma:fan2} in $\tilde{O}(\Delta_{\max}^{3})$ time. 
    Next we can can trace and temporarily mark the two-hop neighbourhood of every vertex in $w_1+y_1+P_1$, as well as the edges in $P_1$. 
    Finally, we can follow along the $(i-1)$-step process and check the conditions of Definition~\ref{def:sc} by using the temporary marks. 
    It takes $\tilde{O}((i+1) \ell)$ time to trace the process and $\tilde{O}(\ell (\Delta_{\max}+1)^2)$ to temporarily mark all of the edges and vertices. 
    We note here that any update time required to correctly locate edges or vertices in memory is captured by the $\tilde{O}$-notation which accounts for look-ups in the BBST. 
    We can use the Top-trees to check whether a subpath from a previously used maximal bichromatic path is present. 
\end{proof}
In particular, the lemma above essentially says that we can check if a process is blocked in $\tilde{O}((i+1)\cdot{}\ell\cdot{}(\Delta_{\max}+1)^3)$ time. 
Indeed, checks not accounted for above can be done in a similar fashion. Hence, we have the following lemma:
\begin{lemma} \label{lma:checkblock}
    The algorithm can check if a particular extension point $x$ is blocked in $\tilde{O}(\sqrt{\log n}\cdot{}\ell\cdot{}a^2\cdot{}(\Delta_{\max}+1)^3)$ time. In addition, it can mark any non-blocked strings with $O(1)$ overhead. 
\end{lemma}
\begin{proof}
    For any choice of $x$, we need only check if it is blocked at 3 different levels. In total at most $O(a^{2})$ process represented by a string from the relevant set needs to be checked. 
    Each process has at most $\sqrt{\log n}$ steps, and the lemma follows by applying Lemma~\ref{lma:checkExtension}.
\end{proof}
Finally, we analyse the time it takes to update the data structures:
\begin{lemma} \label{lma:finalUpdTime}
    Let $f$ be marked with $j \leq i$. Then for any $w$, we can update $\mathcal{D}_{i}(f,w)$ in $\tilde{O}( 3^{\sqrt{\log n}}(\Delta_{\max}+1)^{3\sqrt{\log n}+4}\cdot{}\ell \cdot{} a^2 \cdot{} \log^2{n})$ time. 
\end{lemma}
\begin{proof}
    Note firstly that the $\tilde{O}$-notation takes care of any $\log n$ factors arising from inserting, searching or deleting things from BBST's. 

    We need to check at most $\ell$ extensions $x$. Each can be checked in $\tilde{O}(\sqrt{\log n}\cdot{}\ell\cdot{}a^2\cdot{}(\Delta_{\max}+1)^3)$ time by Lemma~\ref{lma:checkblock}. 
    After checking these extensions, we can create the set of strings to store in $\mathcal{D}_{i}(f,w)$ in $O(\tfrac{a}{2})$ time by storing at least one non-blocked string for each non-blocked extension. 
\end{proof}
Picking $\ell = \ellval$ (any slightly larger value also works) and $\avallb \leq a \leq \aval$, and combining all of the above finalizes the proof of Theorem~\ref{thm:colourStructures}. 

\subsection{The proof of Theorem~\ref{thm:colAlg}}
In this section, we present a proof of Theorem~\ref{thm:colAlg} (restated below for convenience). 
\begin{theorem}[Identical to Theorem~\ref{thm:colAlg}]
    Let $G$ be a dynamic graph, and let $\Delta_{\max}$ be a known upper-bound on the maximum degree throughout the entire update sequence. 
    Then, we can maintain a proper $(\Delta_{\max}+1)$-edge-colouring of $G$ in $(\Delta_{\max}+1)^{\tilde{O}(\sqrt{\log n})}$ worst-case update time per operation.
\end{theorem}

Firstly, fix parameters $\ell = \ellval$ and any value of $a$ satisfying that\\ $\avallb \leq a \leq \aval$ so that we may apply the lemmas from the earlier parts of the section. 

\paragraph{Description of the algorithm: }
Maintaining the data structures from Section~\ref{sec:algos} represents most of the work required to show Theorem~\ref{thm:colAlg}. 
The idea is the following. We will indeed maintain these data structures and ensure that after every update terminates, the colouring $c$ maintained by the algorithm of Theorem~\ref{thm:colourStructures} is a proper $(\Delta_{\max}+1)$-edge-colouring.
Given a call to $\texttt{delete}(e)$, we first call $\texttt{un-colour}(e)$ and then subsequently $\texttt{remove}(e)$ in the previous data structures and in our representation of $G$. 
Given a call to $\texttt{insert}(e)$, we do the following steps: 
\begin{itemize}
    \item Call $\texttt{add}(e)$. 
    \item Then we pick an arbitrary vertex $u \in e$, and construct a consistent fan according to Lemma~\ref{lma:fan1}. The fan ends at a vertex $y_1$ and outputs a bichromatic path $P$. 
    \item If $|P| \leq \ell+1$, we have discovered an augmenting chain, and we can immediately jump to step $5)$. 
    Otherwise, we look in $\mathcal{D}_{\sqrt{\log n}}(r(y_1)r(P),u)$ to find an augmenting and consistent stepping process. By Theorem~\ref{thm:colourStructures} this set is non-empty. 
    \item Then we convert this string into the stepping process it presents and then into a Vizing chain in its pre-image. By Lemma~\ref{lma:processToChain} this Vizing chain is augmenting. 
    \item We shift the colouring $c$ according to this chain using the operations $\texttt{un-colour}(\cdot{})$ and $\texttt{colour}(\cdot{},\cdot{})$. 
    Since the colouring is shifted along a chain, we ensure that $c$ is at all times proper, and so all of the required assumptions hold. 
    \item Finally, since the chain is augmenting, after shifting it, the remaining un-coloured edge $f$ has a free colour $\kappa$. We colour it using $\texttt{colour}(f,\kappa)$.
\end{itemize}

\paragraph{Correctness:} Theorem~\ref{thm:main} guarantees the existence of a consistent and augmenting Vizing chain of the form, we seek. 
Corollary~\ref{cor:setfuse} guarantees that the data structure $\mathcal{D}_{\sqrt{\log n}}(r(y_1)r(P),u)$ contains at least one string of the claimed form. 
Next, applying Lemma~\ref{lma:processToChain} ensures that this process guarantees a consistent and augmenting Vizing chain.
Finally, correctness follows by Theorem~\ref{thm:colourStructures}, as the Theorem guarantees the correctness of $\texttt{un-colour}(\cdot{})$ and $\texttt{colour}(\cdot{},\cdot{})$. 

\paragraph{Analysis:} The running time of $\texttt{delete}(e)$ follows immediately from Theorem~\ref{thm:colourStructures}. 
As for $\texttt{insert}(e)$, we note that since the shiftable-chain has length $\tilde{O}(\sqrt{\log n}\ell)$, we need only call $\texttt{un-colour}(\cdot{})$ and $\texttt{colour}(\cdot{},\cdot{})$ $\tilde{O}(\sqrt{\log n}(\ell + (\Delta_{\max}))$ times. 
This part of the process dominates the update time, as we can trace the process in $\tilde{O}(\sqrt{\log n} \cdot{}\ell)$ time, and compute the at most $\sqrt{\log n}$ fans in $\tilde{O}(\Delta^{3}_{\max})$ time each.
This concludes the proof.

\section{Putting it all together} \label{sct:overall}
In this Section, we combine the results of the previous section with a degree-scheduler from~\cite{duan} to show Theorem~\ref{thm:FullColour} 
\begin{theorem}[Identical to Theorem~\ref{thm:FullColour}]
    Let $G$ be a dynamic graph, and $0<\eps<1$ a given parameter.
    We can maintain a proper $(1+\eps)\Delta$-edge-colouring of $G$ in $2^{\tilde{O}_{\log \eps^{-1}}(\sqrt{\log n})}$ amortised update time per operation.
\end{theorem}

\subsection{The final algorithm} \label{sct:combined}
For now, we will still assume that we have access to some upper bound $\Delta_{\max}$ on the maximum degree of $G$ throughout the entire update sequence. We will then get rid of this assumption in the next subsection by using a degree-scheduler due to Duan, He, and Zhang~\cite{duan}. 

The algorithm is now the following: 
If $\Delta_{\max} \leq \tfrac{10^7 \log^5 n}{\eps^{3}}$: we simply apply Theorem~\ref{thm:colAlg} directly. Since $\Delta_{\max}+1 \leq \ceil{(1+\eps)\Delta_{\max}}$ for any $\eps > 0$, we are immediately done by Theorem~\ref{thm:colAlg}. 

Otherwise, if $\Delta_{\max} > \tfrac{10^7 \log^5 n}{\eps^{3}}$, we first apply Theorem~\ref{thm:splitHier} to get a degree-splitting hierarchy of height $0<h \leq 4\sqrt{\log n}$. Then we apply Theorem~\ref{thm:colAlg} with the upper bound on the maximum degree equal to $\Delta'_{\max}$, where $\Delta'_{\max} = \hat{\Delta}_{h}$, to each graph of $\mathcal{G}_{h}$, using a distinct palette to colour each graph. 
Correctness of the algorithm then follows immediately from Theorem~\ref{thm:splitHier} and Theorem~\ref{thm:colAlg}. 

As for the update time of the algorithm, we note that by Theorem~\ref{thm:splitHier} each sequence of $\sigma$ updates to $G$ results in at most $\sigma\cdot{}2^{\tilde{O}_{\log \eps^{-1}}(\sqrt{\log n})}$ updates to graphs in $\mathcal{G}_{h}$ and at most $2^{\tilde{O}_{\log \eps^{-1}}(\sqrt{\log n})}$ total update time. 
We can handle each such update to a graph in $\mathcal{G}_{h}$ in $2^{\tilde{O}_{\log \eps^{-1}}(\sqrt{\log n})}$ time by Theorem~\ref{thm:colAlg}, so in total, we process the $\sigma$ updates in $\sigma\cdot{}2^{\tilde{O}_{\log \eps^{-1}}(\sqrt{\log n})}$ time, thus resulting in an amortised update time of $2^{\tilde{O}_{\log \eps^{-1}}(\sqrt{\log n})}$ time, as claimed. 

Finally, we analyse the number of colours used. We note that by Theorem~\ref{thm:splitHier}, there exists non-negative integers $i,j$ with $i+j = h$ such that:
\[
|\mathcal{G}_{h}| \leq t_{1}^{i}t_{2}^{j}
\]
and 
\[
\Delta'_{\max} = \hat{\Delta}_{h} \leq (1+\tfrac{\eps}{16})\frac{\Delta_{\max}}{t_{1}^{i}t_{2}^{j}}
\]
In particular, each graph in $\mathcal{G}_{h}$ is coloured using at most $(1+\tfrac{\eps}{16})\frac{\Delta_{\max}}{t_{1}^{i}t_{2}^{j}}+1$ colours. 
Therefore, the total number of colours used is upper bounded by:
\[
|\mathcal{G}_{h}| \cdot{}((1+\tfrac{\eps}{16})\frac{\Delta_{\max}}{t_{1}^{i}t_{2}^{j}}+1) \leq (1+\tfrac{\eps}{2})\Delta_{\max} + t_{1}^{i}t_{2}^{j} \leq (1+\eps)\Delta_{\max}
\]
Here we used the fact that Theorem~\ref{thm:splitHier} terminates with a hierarchy for which 
$\Delta_{\max} \geq \tfrac{2\cdot{}t_{1}^{i}t_{2}^{j}}{\eps}$.

\subsection{Dynamic $\Delta$} \label{sec:dynDelta}
In this subsection, we finish the proof of Theorem~\ref{thm:FullColour}. This last part contains no new ideas, as we can use the degree-scheduler of Duan, He and Chang~\cite{duan} essentially ad verbatim. For the sake of completeness, we briefly sketch their approach below.
 
Let $\mu$ be some constant not specified until later (we will ultimately set $\mu = \Theta(\eps)$). 
Consider the exponentially increasing intervals $I_{i}:= [(1+\mu)^{i-1},(1+\mu)^{i}]$ for $i\in \{ 1, \dots, s\}$ with $s$ minimum subject to $(1+\mu)^{s} \geq n$. For each $i$, we will maintain a subgraph of $G$, say $G_{i}$ with the guarantees that if $\Delta(G) \in I_{i}$, then $G_{i} = G$, and that for all $i$, we have $\Delta(G_i) \leq (1+\mu)^{i+1}$ always. 
Furthermore, we will at all times maintain some index $j$ for which $\Delta(G) \in I_{j}$. 

Given the above, one can run any $(1+\eps')\Delta_{\max}$-edge-colouring algorithm on all of the $G_{i}$'s, and then let the colouring of $G_{j}$ be the currently valid edge-colouring. If one sets $\Delta_{\max,i} = (1+\mu)^{i+1}$ in the instance of the edge-colouring algorithm on $G_i$, one has the guarantee that at all times the algorithm maintains a $(1+\mu)^{2}(1+\eps')\Delta$-edge-colouring. 
Indeed, at all times in the copy currently indexed by $j$ will have $\Delta\cdot{}(1+\mu)^{2} \geq \Delta_{\max,j}$. 
Setting $\eps', \mu \leq \tfrac{\eps}{7}$ is then sufficient to achieve the correct number of colours.

In order to maintain this information, Duan, He and Chang do as follows: let $\delta(e)$ be the maximum degree of an endpoint of $e$. 
If $\delta(e) \leq (1+\mu)^{i}$, we add $e$ to $G_{i}$. Now $e$ stays in $G_{i}$ until $\delta(e) > (1+\mu)^{i+1}$, at which point in time it is removed, and not added back until $\delta(e) \leq (1+\mu)^{i}$ again. 
It is now easy to check that that every time some vertex $w$ crosses the threshold $(1+\mu)^{k}$ for some $k$, it performs at most $deg(w)$ insertions into $G_{k}$ or at most $deg(w)$ deletions from $G_{k-1}$. Between any set of insertions into some $G_k$, $deg(w)$ must have increased beyond $(1+\mu)^{k+1}$ and back down to $(1+\mu)^{k}$. All of the deletions incident to $w$ that decreased the degree from $(1+\mu)^{k+1}$ to $(1+\mu)^{k}$ can be charged the cost of performing the $deg(w)$ insertions into $G_k$. 
Deletions can be analysed in a symmetric way.
In total, the amortised recourse of performing these changes is upper bounded by $O(\mu^{-1})$ for each update. 
Since there are at most $O(\tfrac{\log n}{\mu})$ choices of $i$, the total amortised recourse and update time becomes $O(\tfrac{\log n}{\mu})$. 
One can easily keep track of $j$ or directly calculate $j$ using knowledge of the current maximum degree $\Delta$. Therefore picking $\mu, \eps' = \theta(\eps)$ with $\mu, \eps' \leq \tfrac{\eps}{7}$ yields the corollary. 


\section{Acknowledgements}

The author would like to thank Eva Rotenberg for the support, help, and encouragements.

\newpage

\bibliographystyle{alpha}
\bibliography{main}

\end{document}

%% file: structure.tex
In this section, we will show some theoretical results about stepping processes, which will form the basis for the algorithm used for edge-colouring graphs with small maximum degree. 
We will fix some graph $G$ and a proper (partial) colouring $c$ of the edges of $G$. 
We assume that $G$ has maximum degree at most $\Delta_{\max}$, and hence we allow the case where $\Delta(G) < \Delta_{\max}$. 
It will be assumed in all statements that this is the case, and we will not specify it further in this section.
We will also be working with parameters $\ell$, $a$, and $\gamma$. These parameters often need to satisfy different requirements, which will be specified in the required lemmas. 
Towards the end of the section, we will give explicit values for these parameters which will be sufficient for the needed lemmas.
We begin with some basic definitions which we will use to make precise statements.

\subsubsection{Representing strings and sets}
Given a stepping process $C = w_1+y_1+P_1 + \dots + w_i + y_i+P_i$, we can represent it as a string $s(C) = P_1P_2\dots P_i$ and a starting node $w(C) = w_1$. Here, we encode each $P_j$ as its own unique character in the alphabet over which the string is constructed. As such if $P$ and $P'$ differ in at least one way, the character they represent in the string are different. Even if $P$ and $P'$ represent the same maximal bichromatic path, but with different starting points, they will be considered different.
This means that the length of the string $s(C)$ is $i$ if $C$ takes $i$ steps. 
We emphasize again that the order in which $P_j$ is listed matters, so that $y_j$ coincides with the first vertex of $P_j$.
Furthermore, we apologize for the ambiguity of the notation $P_j$ -- in strings it is viewed as a single character, outside of strings as an ordered list of vertices forming a path.

Given a set $\mathcal{S}$ of strings representing stepping processes, we let 
$$\mathcal{S}(p) := \{s | s\in \mathcal{S}, \text{lcp}(p,s) = p \}$$
i.e.\ we include only the strings of $\mathcal{S}$ which have $p$ as a prefix. 
We let 
$$\mathcal{S}_{i,j}:= \{s| s = \bar{s}[i\dots j] \text{ for some } \bar{s}\in S\}$$
i.e.\ the set containing every string from $\mathcal{S}$ restricted to its substring from the $i^{\text{th}}$ to the $j^{\text{th}}$ character. 

For two strings $s$ and $\bar{s}$, we say that they agree until $i$ if the length of their longest common prefix is $i$ and that they differ at $i$ if the length of their longest common prefix is exactly $i-1$. 

We are interested in sets of strings which do not share prefixes, as they encode stepping process that are somewhat 'spread out'. 
In particular, we care about the following property:
\begin{definition} \label{def:agood}
    Given a set $\mathcal{S}$ of strings representing stepping processes all with the same starting node $w$, we say that $\mathcal{S}$ is $(a,i)$-good if the compressed trie $T$, over all strings in $\mathcal{S}$, satisfies that all nodes of depth $\leq i-1$ is either a leaf or an internal node with at least $a$ children i.e.\ degree at least $a+1$ (except for the root which could have degree $a$). 
\end{definition}

We will work with \emph{graded} sets of the form $A = \bigcup \limits_{k = 0}^{\infty} A_k$. 
For an element $\alpha \in A$, it must be the case that there exists an $i$ for which $\alpha \in A_i$, but it need not be the case for more than one such index. 
Given two graded sets $A,B$, we let $A \cup B$ be the graded set defined as $(A \cup B)_{i} = A_{i} \cup B_{i}$.

Typically, we want the sets of strings that we are working with to satisfy some additional regularity requirements. These are outlined in the definition below: 
\begin{definition} \label{def:restricted}
    Given a set of $\mathcal{S}$ of strings representing stepping processes all with the same starting node $w$, we say that $\mathcal{S}$ is $(i,\ell,A,B,P,y)$-restricted for neighbour $y$ and maximal bichromatic path $P$ beginning at $y$, if for all $s \in \mathcal{S}$, we have: \\
    If $s$ represents the stepping process $w_1+y_1+P_1 + \dots + w_j + y_j+P_j$, then: 
    \begin{enumerate}
        \item For all $k$, the path represented by $P_k$ has length at most $\ell$.
        \item $j \leq i$
        \item For all $k$, we have $\paren{\bigcup_{t=k}^{j}P_t} \cap A_{k-1} = \emptyset$
        \item For all $j$, we have $\paren{\bigcup_{t=k}^{j}\{w_t\}} \cap B_{k-1} = \emptyset$
        \item $P_1 \subset P$
        \item $y_1 = y$
    \end{enumerate}
\end{definition}
\begin{remark}
    The inclusion of $y$ has the sole purpose of avoiding that one can include stepping processes with $y_1 \neq y$. 
    In order to not clutter the notation more than absolutely necessary, we will abuse notation slightly and simply write $(i,\ell,A,B,P)$-restricted. It is then implicit that one has to choose an endpoint of $P$ as $y$. 
    In the case of semi-consistent extensions the choice of $y$ is fixed by the extension process. 
    As a result of the omission of $y$, some of the proceeding proofs will not comment on item $6)$, but it follows trivially from the context that it must hold whenever necessary. 
\end{remark}

The goal of the section is to show that no matter the choice of $w$, $y$, and maximal bichromatic path $P$, one can find a large set of augmenting and semi-consistent stepping process that are in some sense supported independently. 
We make this idea precise in the following definition:
\begin{definition} \label{def:heavy}
    We say a node $w$ with associated bichromatic path $P$ is $(a,i, \ell, A,B, P)$-heavy if there exists a set $\mathcal{S}$ of strings representing augmenting, semi-consistent stepping processes beginning at $w$ such that $\mathcal{S}$ is both $(i,\ell,A,B,P)$-restricted and $(a,i)$-good. 
    Furthermore, every string in $\mathcal{S}$ represents a stepping process $w+y_1+P_1 + \dots + w_j + y_j + P_j$ where $P_j \subset P$ and $y_1 = y$ is identical for all of the processes.
\end{definition}
Observe that being $(a,i, \ell, A, B, P)$-heavy implies that one is also $(a,i, \ell, A', B', P)$-heavy for any $A' \subset A$ and $B'\subset B$. In particular, one is also $(a,i, \ell, \emptyset, \emptyset, P)$-heavy. 

What we will ultimately end up showing is the following theorem: 
\begin{theorem} \label{thm:main}
    Let $w$ be any node, $y$ any neighbour of $w$ and $P$ any maximal bichromatic path beginning at $y$. Then $w$ is $(a,\sqrt{\log n}, \ell, \emptyset,\emptyset, P)$-heavy. 
\end{theorem}
Hence, the theorem above shows that we can find a large set of augmenting and semi-consistent stepping process which are independent in the way that they do not share prefixes. 
This can be exploited algorithmically to store smaller sets of well-chosen stepping processes that are 'spread out' in the sense that they cannot overlap the same parts of the graph too much. This in turns allows an algorithm to only have to consider a small amount of stepping process, which is required for efficiency.

Before we show the theorem, we will first make the above connection more explicit. 
The main idea of the algorithm is to maintain such "spread out" or "indpendent" sets of stepping processes, and then iteratively use them to build similar sets for longer and longer stepping process. 
At some point the stepping processes become so long that no matter where we begin, we can find augmenting stepping processes. The algorithm can then use such process to locate augmenting Vizing chains during updates. 

It is essential that the stepping processes used do not all rely too heavily on the same parts of the graph. The algorithm will build processes front-to-back, and one part of the graph might later need to be avoided, to make space for earlier steps of the process.
In such a case, the process might not be semi-consistent, in which case we have no guarantee that it accurately models a Vizing chain. 
In some sense, one can view the set of stepping processes as a \emph{sparse skeleton} that model in which directions augmenting and semi-consistent stepping processes could go.  
Before we get into the details of how such extension are done, we begin with some simple lemmas that we will use extensively later on.
We begin by showing that heaviness is, in some sense, a heriditary property: 
\begin{lemma} \label{lma:hereditary}
    Suppose $w$ is $(a,i, \ell, A, B, P)$-heavy, but that the length of $P$ is more than $\ell$. Then there exists at least $a$ $\bar{w} \in P$ that are $(a,i-1, \ell, \bar{A}, \bar{B}, \bar{P}_{\bar{w}})$-heavy for respective bichromatic paths $\bar{P}_{\bar{w}}$ chosen in a semi-consistent manner, and for graded sets $\bar{A}$ and $\bar{B}$ with $\bar{A}_0 = A_0 \cup A_1$ and $\bar{A}_{i} = A_{i+1}$ for all $i \geq 1$ and similarly $\bar{B}_0 = B_0 \cup B_1$ and $\bar{B}_{i} = B_{i+1}$ for all $i \geq 1$.
\end{lemma}
\begin{proof}
    Since $w$ is $(a,i, \ell, A, B, P)$-heavy with $|P| > \ell$, it follows that it has a witness set of augmenting semi-consistent stepping processes $\mathcal{S}$ of size at least $a$. 
    Hence in the compressed trie $T_{\mathcal{S}}$ over $\mathcal{S}$, the root has at least $a$ children and every child $\bar{w}$ is the root of a subtree $T_{\mathcal{S}}(\bar{w})$ satisfying that until depth $i-2$ every node is either a root or has internal degree at least $a+1$. 
    But for any such choice of $\bar{w}$, we note that this subtree coincides with the compressed trie of the set of strings $\mathcal{S}_{\bar{w}}$ obtained as follows: 
    Let $P' \subset P$ be the bichromatic path leading to extension point $\bar{w}$. Then set $\mathcal{S}_{\bar{w}} = (\mathcal{S}(P'))_{2\dots i}$. 
    We claim that this set certifies that $\bar{w}$ is $(a,i-1, \ell, \bar{A}, \bar{B}, \bar{P}_{\bar{w}})$-heavy for respective bichromatic paths $\bar{P}_{\bar{w}}$ chosen in a semi-consistent manner. 
    Indeed, above we showed that $\mathcal{S}_{\bar{w}}$ is $(a,i-1)$-good, so we need only argue that it is $(i-1,\ell,\bar{A}, \bar{B},\bar{P}_{\bar{w}})$-restricted. Since $S$ is $i$ restricted, and we have truncated the first step away, it is clear that we satisfy items $1)$ through $4)$ of the definition. To see $5)$, we note that any extension through $\bar{w}$ has only one unique choice for maximal bichromatic path $P'$ to extend through by Lemma~\ref{lma:fan2} and the proceeding discussion, since every stepping process in $S$ is semi-consistent by asssumption. 
\end{proof}

Next, we show that a sort of 'reverse' heriditary property also holds; the lemma below shows that if you have many heavy extension points, you must be heavy yourself:
\begin{lemma} \label{lma:reversheriditary}
    Let $w$ be a vertex, $y$ a neighbour of $w$, $P$ a maximal bichromatic path ending at $y$ with length more than $\ell$, and let the parameters $a, i, A, B$ be given such that $A_{0} \cap P[1 \dots \ell] = \emptyset$ and $w \notin B_0$. Suppose, furthermore, that $a$ points $\bar{w} \in P[1\dots \ell] - B_{1}$ are $(i-1, \ell, \bar{A}, \bar{B}, \bar{P})$-heavy for some semi-consistent extension path $\bar{P}$ (obtained by truncating a stepping process of the form $w+y+P'$ for $P' \subset P$) and graded sets given by $\bar{A}_0 = A_0 \cup A_{1} \cup N^{2}(w)$, $\bar{A}_1 = A_2 \cup N^{2}(P')$, and $\bar{A}_{k} = A_{k+1}$ otherwise, and $\bar{B}_0 = B_0 \cup B_{1} \cup N^{2}(w)$, $\bar{B}_1 = B_2 \cup N^{2}(P')$,  and $\bar{B}_k = B_{k+1}$ otherwise. Here $P'$ is the subpath of $P$ from $y$ to the next vertex after $\bar{w}$.

    Then $w$ is $(a,i, \ell, A, B, P)$-heavy. 
\end{lemma}
\begin{proof} 
    Let $\mathcal{S}_{\bar{w}}$ be a set of prefixes certifying that $\bar{w}$ is $(i-1, \ell, \bar{A}, \bar{B}, \bar{P})$-heavy with parameters as outlined in the statement of the lemma. By assumption, there are at least $a$ choices of $\bar{w}$ where such a set exists. 
    Now, for any such set $\mathcal{S}_{\bar{w}}$, we let $\tilde{\mathcal{S}}_{\bar{w}}$ be the set containing prefixes of the form $P' + s$ for $s \in \tilde{\mathcal{S}}_{\bar{w}}$, where $P' \subset P$ is the subpath of $P$ from $y$ to $\bar{w}$.

    Now we claim that for any $(i-1, \ell, \bar{A}, \bar{B}, \bar{P})$-heavy choice of $\bar{w}$, we have that $\tilde{\mathcal{S}}_{\bar{w}}$ is $(i,\ell,A,B,P)$-restricted. 
    Indeed, by choice of $\bar{w}$,  $P'$ has length at most $\ell$. Since $\bar{w}$ is $(i-1, \ell, \bar{A}, \bar{B}, \bar{P})$-heavy, it is also $(i-1, \ell, \bar{A}, \bar{B}, \bar{P})$-restricted. As such, items $1)$ and $2)$ of Definition~\ref{def:restricted} are true for $\tilde{\mathcal{S}}_{\bar{w}}$. 
    Items $3)$ and $4)$ are readily checked: For $k = 0$, by assumption $P' \cap A_0 = \emptyset$, so for $k = 0$, item $3)$ holds. For $k \geq 1$ we have $A_{k+1} \subset \bar{A}_{k}$, so, by assumption, item $3)$ also holds in these cases. 
    Item $4)$ is similar. 
    Finally, item $5)$ holds by construction. 

    Next we claim that every string $s \in \tilde{\mathcal{S}}_{\bar{w}}$ presents a stepping process that is augmenting and semi-consistent. It is sufficient to show that they are semi-consistent, since then they must be augmenting, as every string in $\mathcal{S}_{\bar{w}}$ represents an augmenting stepping process. 

    To see that they are also semi-consistent we check as follows: First we let $w+y+P'+\bar{w}+\bar{y}+\bar{P}+ w_{3}+y_3+P_3 \dots + w_{j}+y_j+P_j$ with $j \leq i-1$ be any stepping process represented by a string in $\tilde{\mathcal{S}}_{\bar{w}}$. 
    Since $\bar{P}$ is a semi-consistent extension path, it cannot be identical to $P$. As $\bar{w}+\bar{y}+\bar{P}+ \dots + w_{j}+y_j+P_j$ is semi-consistent by assumption, it follows that the first item of Definition~\ref{def:sc} is true. 
    Since $N^{2}(w) \subset \bar{A}_0 \cap \bar{B}_0 $ and $N^{2}(P') \subset \bar{A}_1 \cap \bar{B}_1$, it follows by the semi-consistency of $\bar{w}+\bar{y}+\bar{P}+ \dots + w_{j}+y_j+P_j$ that item $2)$ and $3)$ holds. 

    Now let $\mathcal{S}$ be the union of $\tilde{\mathcal{S}}_{\bar{w}}$ for all valid choices of $\bar{w}$. 
    We claim that then $\mathcal{S}$ certifies that $w$ is $(a,i, \ell, A, B, P)$-heavy. Indeed, since any set $\tilde{\mathcal{S}}_{\bar{w}}$ is $(i,\ell,A,B,P)$-restricted and only contains strings representing semi-consistent and augmenting stepping processes, it follows that $\mathcal{S}$ inherits these properties. Hence, we need only show that $\mathcal{S}$ is $(a,i)$-good. 
    To this end consider the compressed trie $T$ over $\mathcal{S}$. We node that the root node has degree at least $a$ in $T$, since by assumption we have at least $a$ valid choices of $\bar{w}$. Since all of the sets $\tilde{\mathcal{S}}_{\bar{w}}$ are $(a,i-1)$-good, we find that in fact every internal node of $T$ of depth at most $i-1$ is either a leaf or has at least $a$ children. 

    Finally, we conclude that $\mathcal{S}$ certifies that $w$ is $(a,i, \ell, A, B, P)$-heavy. 
\end{proof}

In the lemma, we wish to show why it is important to be able to store stepping processes that are somewhat spread out. To be more precise, we show that if a triplet $(w,y,P)$ is not heavy, then it can be \emph{suppressed} by a small set of vertices:
\begin{lemma} \label{lma:suppression}
    Let $w$, a vertex, and $P$, a maximal bichromatic path ending at a neighbour of $w$, be given. Furthermore, let parameters $a, i, A, B$ be given. Suppose that $w$ is \emph{not} $(a,i, \ell, A, B, P)$-heavy. Suppose $P$ has length more than $\ell$. 

    Then one can augment $B$ by a graded set $D$ with $|D_k| \leq \ell^{k-1} \cdot{} a$ for all $k \geq 1$ such that for graded set $\hat{B} = B \cup D$, no non-empty, $(i,\ell,A,\hat{B},P)$-restricted set of strings representing augmenting and semi-consistent stepping processes exists. 
    
    In such a case, we say that $D$ $(a,i, \ell, A, B, P)$-\emph{suppresses} $w$.
\end{lemma}
\begin{proof}
    The proof is by induction on $i$. For $i = 1$, the statement is clear. Indeed, no $1$-step stepping process can be both augmenting and $(i,\ell,A,\hat{B},P)$-restricted under the assumptions of the lemma. 

    For the step, we assume that the statement is true for $i-1$ and any valid choice of $w$ and $P$, and we want to show that it is true for $i$ and any valid choice of $w$ and $P$. 
    To this end, we assume we are given valid $w$ and $P$. 
    We note by Lemma~\ref{lma:reversheriditary} that there can be at most $a$ $(a,i-1, \ell, \bar{A}, \bar{B}, \bar{P}_{\bar{w}})$-heavy semi-consistent extensions along $P[1 \dots \ell]$ for some semi-consistent extension $\bar{P}_{\bar{w}}$. Here  $\bar{A}$ and $\bar{B}$ are as defined in the lemma. 
    Indeed, otherwise $w$ is $(a,i, \ell, A, B, P)$-heavy. 
    We let $D_1$ contain exactly these extension vertices. 
    Now there are at most $\ell$ semi-consistent extension along $P[1 \dots \ell]$ that are not $(a,i-1, \ell, \bar{A}, \bar{B}, \bar{P}_{\bar{w}})$-heavy for their respective semi-consistent extension path $\bar{P}_{\bar{w}}$. 

    By induction, any such extension point $\bar{w}$ can be $(a,i-1, \ell, \bar{A}, \bar{B} \cup D[\bar{w}], \bar{P}_{\bar{w}})$-suppressed by some graded set $D[\bar{w}]$ with $|D[\bar{w}]_i| \leq \ell^{i-1} \cdot{} a$. 
    For $k \geq 2$, we let $D_{k} = \bigcup \limits_{\bar{w}} D[\bar{w}]_{k-1}$.
    Observe that $|D_k| \leq \ell \cdot \ell^{k-2} \cdot{} a$ by the induction assumption. 

    Now, we claim that $D$ $(a,i, \ell, A, B, P)$-suppresses $w$. Indeed, suppose not, and let $w+y_1+P_1 + w_2 + y_2 + P_2 + \dots w_j + y_j + P_{j}$ be some augmenting and semi-consistent stepping process whose string forms a $(i,\ell,A,B \cup D,P)$-restricted set. 
    Note that then $w_2 + y_2 + P_2 + \dots w_j + y_j + P_{j}$ is also semi-consistent and augmenting. 
    Since $w_2 \notin D_1$, we note that $D[w_2]$ does not suppress $w_2 + y_2 + P_2 + \dots w_j + y_j + P_{j}$, which must mean that $w_2 + y_2 + P_2 + \dots w_j + y_j + P_{j}$ is not $(i-1, \ell, \bar{A}, \bar{B}, \bar{P}_{\bar{w}})$-restricted. 
    Indeed, if $D[w_2]$ suppressed $w_2 + y_2 + P_2 + \dots w_j + y_j + P_{j}$, then by construction $D$ would suppress $w+y_1+P_1 + w_2 + y_2 + P_2 + \dots w_j + y_j + P_{j}$. 

    However, this is a contradiction. To see why, observe that $w+y_1+P_1 + w_2 + y_2 + P_2 + \dots w_j + y_j + P_{j}$ is $(i,\ell,A,B \cup D,P)$-restricted. This implies that $w_2 + y_2 + P_2 + \dots w_j + y_j + P_{j}$ satisfies items $1), 2)$ and $5)$ in Definition~\ref{def:restricted}. 
    If item $3)$ fails, it must be because $\exists r,t$ such that $P_{r} \cap \bar{A}_{t} \neq \emptyset$ and $r \geq t + 2$. If $t\geq 2$ this implies that $P_{r} \cap A_{t+1} \neq \emptyset$, which is a contradiction. For $t = 0$, we find a contradiction, since if $P_{r}$ for $r\geq 2$ intersects $\bar{A}_{0}$, then either $P_{r} \cap (A_0 \cup A_1) \neq \emptyset$ or $P_r \cap N^{2}(w)$ is non-empty. Both are contradictions as $r \geq 2$ and the stepping process is semi-consistent.
    
    For $t = 1$, we must have $P_r \cap \bar{A}_{1} \neq \emptyset$ with $r \geq 3$, but then either $P_r \cap A_2 \neq \emptyset$ or $P_r \cap N^{2}(P')$ is non-empty. Again both are contradictions due to the stepping process being $(i,\ell,A,B \cup D,P)$-restricted and semi-consistent. 

    Similarly to above, one can show that Item $4)$ also holds. We include this for completeness below, but first we recall that this would imply that $w_2 + y_2 + P_2 + \dots w_j + y_j + P_{j}$ is in fact $(i-1, \ell, \bar{A}, \bar{B}, \bar{P}_{\bar{w}})$-restricted, and thus we ultimately reach a contradiction to the statement that $D$ does not $(a,i, \ell, A, B, P)$-suppress $w$, and the lemma follows. 

    To see that item $4)$ holds, assume that it does not. Then there must exist $r$ and $t$ with $r \geq t+2$ such that $w_{r} \cap \bar{B}_{t}$ is non-empty. 
    Like before, if $t \geq 2$ this would imply that $w_{r} \cap B_{t+1}$ is non-empty which contradicts the fact that $w+y_1+P_1 + w_2 + y_2 + P_2 + \dots w_j + y_j + P_{j}$ is $(i,\ell,A,B \cup D,P)$-restricted. 
    If $t = 0$, either $w_r \cap (B_0 \cup B_1)$ or $w_r \cap N^2(w)$ is non-empty. Both are contradictions to respectively the restrictedness of the process or the semi-consistency of the process. 
    If $t = 1$, either $w_r \cap B_2$ or $w_2 \cap N^{2}(P')$ is non-empty, which is a contradiction similarly to above. 
\end{proof}

So far, we have been working with sets of strings representing a very precise set of processes. In this next part, we wish to relax this notion a bit, and introduce less restrictive requirements that could be realised by an efficient algorithm. 
In particular, we are interested in sets of strings of the following form:
\begin{definition} \label{def:spread}
    We say a node $w$ with associated maximal bichromatic path $P$ and neighbour $y$ is $(a,i,P)$-spread if there exists a set $\mathcal{S}$ of strings representing augmenting, semi-consistent stepping processes with at most $i$ steps, all beginning at $w$, such that $\mathcal{S}$ is both $(a,1)$-good and $(\tfrac{a}{2},2)$-good and every string in $\mathcal{S}$ represents a stepping process whose first bichromatic path is a subpath of $P$ beginning at $y$. 
\end{definition}
The following lemma is an immediate consequence high-lighting the fact that the new requirement is less restrictive:
\begin{lemma}
    Suppose $w$ with neighbour $y$ is $(a,i, \ell, A, B, P)$-heavy for some choice of $A, B$. Then $w$ is $(a,i,P)$-spread. 
\end{lemma}
\begin{proof}
    If $i \geq 2$ the statement is clear, as in this case being $(a,i, \ell, A, B, P)$-heavy is strictly more restrictive than being $(a,i,P)$-spread. If $i = 1$ the same is also true, as in this case $w+y+P$ must necessarily be augmenting, if $w$ is $(a,i, \ell, A, B, P)$-heavy.
\end{proof}

Next, we wish to make precise why it is helpful that the sets contain strings which do not share prefixes. 
To make this formal, we will consider strings which contain slightly more information: by a \emph{complete string} of a stepping process $w_1 + y_1 + P_1 + \dots w_j+y_j+P_j$, we mean the string $w_1$ concatenated with the usual string representing the process, namely $P_1\dots P_j$. A \emph{complete prefix} is a prefix of a complete string representing a semi-consistent process. 

The following lemma now shows that if one only considers prefix-free strings, then the stepping process that they represent must be mostly spread out. The proof of this lemma is rather standard (see for example~\cite{Christiansen}). 
\begin{lemma} \label{lma:prefixcount}
    Let $w$ be any vertex. Let $Z_i$\footnote{Sometimes denoted by $Z_i[w]$ if the choice of $w$ is not clear from the contrast.} be the set of complete strings of semi-consistent, $i$-step stepping processes visiting $w$ for the first time during the $i^{\text{th}}$ step. Then the size of $(Z_i)_{1,i}$ is at most $(3(\Delta_{\max}+1)^3)^{i}$. 
    In other words, there exists a set of strings $X_i$ of size at most $(3(\Delta_{\max}+1)^3)^{i}$ such that the complete prefix of length $i$ of every string in $Z_i$ belongs to $X_i$. 
\end{lemma}
\begin{proof}
    The proof is by strong induction on $i$.
    For $i = 1$, we observe that at most $(\Delta_{max}+1)^2$ bichromatic paths go through $w$. Any $1$-step process using one such bichromatic path can have at most $2(\Delta_{\max}+1)$ different complete prefixes of length $1$, namely one for each choice of neighbour of the endpoints of the path. Furthermore, any neighbour of $w$ can also visit $w$ during a $1$-step process, and so we conclude that letting $X_1$ contain strings for all of these at most $2(\Delta_{\max}+1)^3 + (\Delta_{\max}+1) \leq 3(\Delta_{\max}+1)^3$ complete prefixes of length $1$ is a valid choice. 

    For the induction step, we assume the statement is true for $1, 2, \dots i-1$ and any choice of $w$. Now any semi-consistent stepping process visiting $w$ for the first time during the $i^{\text{th}}$ step, must visit one of the vertices represented by a string in $X_1$ for the first time during the $(i-1)^{\text{st}}$ step (if it does not visit the vertex for the first time, one cannot extend through it in a semi-consistent manner). 
    But for every $\bar{w}$ represented by a string in $X_1$, the induction hypothesis guarantees that at most $(3(\Delta_{\max}+1)^3)^{i-1}$ complete prefixes can be a complete prefix of length $i-1$ for a stepping process in $Z_{i-1}[\bar{w}]$. Each such semi-consistent process can only extend in one way through $\bar{w}$ and so any process represented by string $s$ in $Z_{i}$ must be represented by a string from $X_{i-1}$ concatenated with a new character $P_{i-1}$. 
    As noted before, semi-completenes implies that there is only one such choice of $P_{i-1}$ for each process represented by a string from $X_{i-1}[\bar{w}]$, namely truncating the bichromatic path at exactly $\bar{w}$. 

    Hence, $X_i[w]$ is a subset of the set containing all such complete prefixes, of which there are at most $\sum \limits_{\bar{w} \in X_{1}[w]} |X_{i-1}[\bar{w}]| \leq  (3(\Delta_{\max}+1)^3)^{i}$, by the induction hypothesis.
\end{proof}
We will sometimes use the notation introduced in the lemma before. I.e.\ we let $Z_{i}[w]$ refer to the set of complete strings of semi-consistent, $i$-step stepping processes visiting $w$ for the first time during the $i^{\text{th}}$ step, and we let $X_{i}[w] = (Z_{i}[w])_{1,i}$.

Given a set of strings $\mathcal{S}$, we let the $i^{\text{th}}$-\emph{span} of $\mathcal{S}$ denote the set of complete prefixes of length $i$ belonging exactly to strings in $\mathcal{S}$. We typically say that $\mathcal{S}$ span these prefixes. 
The next lemma shows how one can fuse sets representing spread-out processes in order to form a new set representing longer processes that are still spread-out. This type of operation is key to the algorithm that we will develop later on in this section.
\begin{lemma} \label{lma:fuse}
    Let $a \geq \avallb$, $\ell \geq \ellval$, and $i \leq \sqrt{\log n}$ be given parameters. 
    Suppose that every $(a,i-1, \ell, \emptyset, \emptyset, \bar{P})$-heavy node $\bar{w}$ with neighbour $\bar{y}$ and bichromatic path $\bar{P}$ is assigned an (arbitrarily chosen) set $\mathcal{S}_{\bar{w},\bar{y},\bar{P}}$ that certifies that $\bar{w}$ is $(\tfrac{a}{2},i-1,\ell)$-spread.
    Suppose, furthermore, that some arbitrary selection of the remaining $(\tfrac{a}{2},i-1,\ell)$-spread nodes $\bar{w}$ also are assigned such sets (the selection could be empty). 
    For all remaining choices of $\bar{w},\bar{y}$, and $\bar{P}$, set $\mathcal{S}_{\bar{w},\bar{y},\bar{P}} = \emptyset$.

    Let $w$ be a vertex, $y$ a neighbour of $w$, and $P$ a bichromatic path beginning at $y$. Suppose that the length of $P$ is greater than $\ell $. 
    Let $\bar{\mathcal{S}}$ be the union of sets $\mathcal{S}_{\bar{w},\bar{y},\bar{P}}$ for every extension point $\bar{w}$, with associated neighbour $\bar{y}$ and bichromatic extension path $\bar{P}$ picked according to Item $1)$ in Definition~\ref{def:sc}, along $P$, when beginning a stepping process from $w$. Furthermore, let $\tilde{S}$ contain strings representing processes of the form $P'+s$, with $s \in \mathcal{S}_{w}$ and $P' \subset P$.
    
   Then if $w$ is $(a,i, \ell, \emptyset, \emptyset, P)$-heavy, there exists a set $\mathcal{S}_{w}$ contained in $\tilde{\mathcal{S}}$ of size at most $\tfrac{a^2}{8}$ that certifies that $w$ is $(\tfrac{a}{2},i-1,\ell)$-spread.
\end{lemma}
\begin{proof} 
    Take any augmenting and semi-consistent stepping process represented by a string $s$ from $\bar{\mathcal{S}}$. Say it has the form $w_2+y_2+P_2 + \dots + w_j + y_j + P_j$. 
    Since $w_2$, $y_2$, and $P_2$ are picked according to Item $1)$ in Definition~\ref{def:sc}, we note that there exists a subpath $P_1 \subset P$ such that $w+y+P_1+w_2+y_2+P_2+ \dots +w_j + y_j + P_j$ is a stepping process that satisfies Item $1)$ of Definition~\ref{def:sc}. 

    We will say that any such process that is overlapping or violate Item $2)$ or $3)$ of Definition~\ref{def:sc} is \emph{blocked}.
    We will also refer to the string $P_1+s$ representing it as \emph{blocked}.
    Since these type of extensions are unique, we will sometimes simply refer to the string $s$ or the process $w_2+y_2+P_2+ \dots +w_j + y_j + P_j$ as being blocked. 

    There are a few ways in which this can happen. First of all, Item $2)$ can be violated only if there is some index $k \geq 3$ for which $w_k$ overlaps with $N^{2}(\{w_1\}) \cup N^{2}(P_1)$. 
    There are at most $|N^{2}(\{w_1\}) \cup N^{2}(P_1)| \leq (\Delta+1)^{2} \cdot{} (\ell + 1)$ such vertices. 
    Each such point can be reached by at most $(3(\Delta_{\max}+1)^3)^{i}$ complete prefixes of length $i$ by Lemma~\ref{lma:prefixcount}. 
    Hence, at most $i \cdot{} (3(\Delta_{\max}+1)^3)^{i} \cdot{} (\Delta_{\max}+1)^{2} \cdot{} (\ell + 1)$ of the processes represented in $\bar{\mathcal{S}}$ can be violating in this way. 
    Of these, at most $(3(\Delta_{\max}+1)^3) \cdot{} (\Delta_{\max}+1)^{2} \cdot{} (\ell + 1)$ are violating at $w_3$.

    We can count the number of processes blocked by Item $3)$ in a similar fashion. In fact, we will get the same result, as we simply counted the number of such processes that reached points in $N^{2}(\{w_1\}) \cup N^{2}(P_1)$ during the second step or later. 
    To this, we add the at most $3(\Delta_{\max}+1)^3 \cdot{} (\Delta_{\max} + 1)$ choices of $w_2$ which could yield a $P_2$ reaching a vertex in $N^{2}(\{w_1\})$, and the at most $2 \cdot{} (3(\Delta_{\max}+1)^3)^{2} \cdot{} (\Delta_{\max}+1)^{2} \cdot{} (\ell + 1)$ choices of $w_3$ which could yield a $P_3$ reaching a vertex in $N^{2}(\{w_1\}) \cup N^{2}(P_1))$.

    Finally, we note that the above covers all overlapping cases as well, as the case where $P_1$ and $P_2$ belong to the same maximal bichromatic path is already accounted for in the above. 
    Hence if $P_1+s$ is not blocked, then it represents an augmenting and semi-consistent stepping process. 

    We say a string is blocked at level one, if $P_2$ reaches a vertex in $N^{2}(\{w_1\})$. It is blocked at level two if $w_3$ violates Item $3)$ or if $w_3$ yields a path $P_3$ which could reach a vertex in $N^{2}(\{w_1\}) \cup N^{2}(P_1))$. 
    In all other cases, a string is blocked at level three. 

    We say an extension point $\bar{w}$ is \emph{blocked} if its associated set $\mathcal{S}_{\bar{w}, \bar{y}, \bar{P}}$ represents at least one process blocked at level one, or at least $\tfrac{a}{8}$ processes blocked at level two, or at least $\tfrac{a^2}{32}$ processes blocked at level three. 

    If an extension point $\bar{w}$ is not blocked, then the set $\mathcal{S}_{\bar{w}, \bar{y}, \bar{P}}$ contains either only one string since $\bar{w} + \bar{y} + \bar{P}$ is augmenting, or at least $\tfrac{a}{4}$ non-blocked strings representing augmenting and consistent stepping processes such that the first character of each string is different. 
    In order to see this, we assume that $|\mathcal{S}_{\bar{w}, \bar{y}, \bar{P}}| > 1$, so that we must be in the latter case. Now, we claim that if one remove all blocked strings from the set, it would still span at least $\tfrac{a}{4}$ different complete prefixes of length 2. 
    Indeed, by definition at most $\tfrac{a}{8}$ processes are blocked at level two, so at most $\tfrac{a}{8}$ choices of $w_3$ are blocked. This removes at most $\tfrac{a}{8}$ complete prefixes of length 2 from the second span of the remaining strings.
    Finally, a process in $\mathcal{S}_{\bar{w}, \bar{y}, \bar{P}}$ cannot be blocked at level 3 and only have two steps, since any violation here occurs at some $w_j$ or $P_j$ with $j>3$. 
    Hence, it takes $\tfrac{a}{4}$ blocked processes at level 3 to block all processes through any choice of $w_3$. Hence blocked processes at level 3 remove at most $\tfrac{a}{8}$ complete prefixes of length 2 from the second span of the remaining strings.
    Hence, the claim follows. 

    If at most $\tfrac{a}{2}$ choices of $w_2$, which is also assigned a non-empty set, are blocked, we arbitrarily pick $\tfrac{a}{2}$ non-blocked $w_2$. This is possible, as $w$ is assumed to be $(a,i, \ell, \emptyset, \emptyset, P)$-heavy, and so has at least $a$ extension points that are $(a,i-1, \ell, \emptyset, \emptyset, \bar{P})$-heavy for some appropriate choice of $\bar{P}$ by Lemma~\ref{lma:reversheriditary}, and hence are assigned a non-empty set by assumption.
    
    For each choice of $w_2$ we pick either one non-blocked and augmenting string of the form $P_2\in \mathcal{S}_{w_2, y_2, P_2}$ (in the case where $\mathcal{S}_{w_2, y_2, P_2} = \{P_2\}$) or $\tfrac{a}{4}$ non-blocked strings from $\mathcal{S}_{w_2, y_2, P_2}$ so that none of the picked strings in $\mathcal{S}_{w_2, y_2, P_2}$ share a complete prefix of length two. 
    By above such a choice exists, and subject to the conditions above the choice can be made arbitrarily.
    These strings can be extended by adding an appropriate sub-path $P_1 \subset P$ to obtain strings representing stepping processes beginning at $w$.

    The collection of all of these extended strings will form $\mathcal{S}_{w}$. 
    We claim that the set has the desired properties. Observe firstly that it has the claimed size and contains string of the appropriate form by construction. 
    Indeed, the strings are of the form $P_1+s$, and they represent processes that are augmenting and semi-consistent. 
    Secondly, note that by construction $\mathcal{S}_{w}$ is both $(\tfrac{a}{2},1)$-good and $(\tfrac{a}{4},2)$-good. 
    Hence, the set certifies that $w$ is $(\tfrac{a}{2},i-1,\ell)$-spread.

    The lemma therefore follows, if we can show that at most $\tfrac{a}{2}$ valid extension vertices are blocked. 
    Hereto, we note that we have already counted the number of vertices blocked at the different levels. Namely, there are at most $3(\Delta_{\max}+1)^3 \cdot{} (\Delta_{\max} + 1)$ points blocked at level one.

    There are at most $ 3 \cdot{} (3(\Delta_{\max}+1)^3)^{2} \cdot{} (\Delta_{\max}+1)^{2} \cdot{} (\ell + 1)$ processes blocked at level two. These can in turn cause at most 
    $$\frac{3 \cdot{} (3(\Delta_{\max}+1)^3)^{2} \cdot{} (\Delta_{\max}+1)^{2} \cdot{} (\ell + 1)}{\tfrac{a}{8}} \leq (\Delta_{\max}+1)^{9 \sqrt{\log n}} \leq \tfrac{a}{6}$$ 
    vertices to be blocked. 
    Finally, we have at most $\sqrt{\log n}^2 \cdot{} (3(\Delta_{\max}+1)^3)^{\sqrt{\log n}} \cdot{} (\Delta_{\max}+1)^{2} \cdot{} (\ell + 1)$ vertices blocked at level three, which can cause at most 
    $$ \frac{\log n \cdot{} (3(\Delta_{\max}+1)^3)^{\sqrt{\log n}} \cdot{} (\Delta_{\max}+1)^{2} \cdot{} (\ell + 1)}{\tfrac{a^2}{32}} \leq \tfrac{a}{6}$$ 
    vertices to be blocked.
    Thus we find that at most $\tfrac{a}{2}$ choices of $w_2$ can be blocked, and so our earlier assumption always holds true, and the lemma follows.  
\end{proof}

Observe that as long as one avoids blocked extension points and processes, the remaining choices in the proof of Lemma~\ref{lma:fuse} can be made arbitrarily. This will be heavily exploited in the final algorithm, where we will compute the above sets in an iterative fashion, using the strategy outlined in the proof above. 

We will essentially maintain sets that form a valid solution to the following iterative procedure. Suppose initially that for any choice of $w,y,$ and $P$ such that $w$ is $(a,1, \ell, \emptyset, \emptyset, P)$-heavy, we set $\mathcal{S}_{w,y,P}^1 = \{P\}$ i.e.\ the set contains only the string representing the augmenting process $w+y+P$. For all other choices of $w,y,$ and $P$, we leave the set $\mathcal{S}_{w,y,P}^1$ empty.
Suppose that for all $i \leq \sqrt{\log n}$ and for all choices of $w,y,P$ such that $w$ is $(a,i, \ell, \emptyset, \emptyset, P)$-heavy, we inductively construct $\mathcal{S}_{w,y,P}^{i}$ using the approach from above. Suppose, furthermore, that some arbitrary selection of the remaining $(\tfrac{a}{2},i,\ell)$-spread nodes $\bar{w}$ also are assigned sets $\mathcal{S}_{w,y,P}^{i}$ certifying that they are $(\tfrac{a}{2},i,\ell)$-spread. 
Suppose finally, that all remaining choices of $\bar{w},\bar{y}$, and $\bar{P}$, are assigned the empty set i.e.\ here $\mathcal{S}_{\bar{w},\bar{y},\bar{P}}^{i} = \emptyset$.

Then it follows by induction on $i$ together with Lemma~\ref{lma:fuse} that any choice of $w,y,$ and $P$ such that $w$ is $(a,\sqrt{\log n}, \ell, \emptyset, \emptyset, P)$-heavy will be assigned a non-empty set $\mathcal{S}_{\bar{w},\bar{y},\bar{P}}^{\sqrt{\log n}}$. 
It follows by Theorem~\ref{thm:main} that any such triplet indeed causes $w$ to be $(a,\sqrt{\log n}, \ell, \emptyset, \emptyset, P)$-heavy, and so, crucially, any algorithm maintaining sets in alignment with the above procedure, will produce a non-empty set $\mathcal{S}_{\bar{w},\bar{y},\bar{P}}^{\sqrt{\log n}}$ for any valid choice of $w,y,$ and $P$.
Hence, we have shown the following corollary:
\begin{corollary} \label{cor:setfuse}
    Any algorithm maintaining sets $\mathcal{S}_{w,y,P}^{i}$ as described above for all $i\leq \sqrt{\log n}$ will produce a non-empty set $\mathcal{S}_{\bar{w},\bar{y},\bar{P}}^{\sqrt{\log n}}$ for any choice of vertex $w$, neighbour $y$, and associated bichromatic path $P$.
    Furthermore, the set $\mathcal{S}_{\bar{w},\bar{y},\bar{P}}^{\sqrt{\log n}}$ certifies the fact that $w$ is $(\tfrac{a}{2},\sqrt{\log n}, P)$-spread.
\end{corollary}
We are now ready to show Theorem~\ref{thm:main}. To do so, we first need an important building block which we will develop in the next subsection.

\subsubsection{Proof of Proposition}

In this subsection, we show Proposition~\ref{prop:main} introduced below. 
Our setup and choice of constants is different compared to~\cite{Christiansen}, but we will follow a proof approach similar to the one introduced there. 

Throughout the section, we fix parameters $a$, $\ell$, $\Delta_{\max}$, and $\gamma$ and a graph $G$ endowed with a proper (partial) $(\Delta_{\max} + 1)$-colouring $c$. 

We will consider a game between two players, Alice and Bob. 
In this game, we are given as input a starting vertex $w$, a neighbour $y$ of $w$, and a maximal bichromatic path $P$ beginning at $y$. 
The goal of Alice is to find an augmenting, semi-consistent stepping process with step length at most $\ell$ such that the first step of the process takes the form $w+y+P'$ for some subpath $P' \subset P$.
The goal of Bob is to deny Alice in achieving their goal. 

The game proceeds in rounds. In the first round, we let $R_1$ denote all of the notes along $P$ that can be visited by a $1$-step process of the form Alice seeks, though not necessarily augmenting. 
Alice then gets to pick an arbitrary subset $N_1 \subset R_1$ along with one semi-consistent process that gets to extend through $N_1$. That is, for any vertex $u \in N_1$, Alice gets to pick exactly one semi-consistent $1$-step stepping process ending up at $u$, and then Alice gets to extend semi-consistently through $u$. Note that Alice is allowed to pick any valid choice of $P_2$, but they must all extend semi-consistently through $u$ using the process that Alice picked to get to $u$. In other words, all valid choices of $P_2$ belong to the same maximal bichromatic path. 
After Alice chooses $N_1$ and the valid extensions, Bob gets to ban a set of vertices $B_1$ of size at most $\gamma |N_i|$. 
Once Bob has specified their set of banned vertices, $R_2$ is determined to be be any point reachable from a $2$-step process $w+y+P'+w_2+y_2+P_2$, of the form Alice seeks (not necessarily augmenting), for which $w_2 \in N_1 \setminus B_1$ and $P_2$ is consistent with the extension Alice chose. In other words, the $2$-step process is represented by a string which forms a $(2,\ell, \emptyset, \hat{B}, P)$-restricted set with $\hat{B}_1 = B_1$ and $\hat{B}_0 = \emptyset$. 

The $i^{\text{th}}$ round proceeds completely synchronous to above. 
Alice picks set $N_i \subset R_i$ together with one extension type per point. That is for $u \in N_i$, Alice picks any semi-consistent $i$-step process of the form $w+y+P'+w_2+y_2+P_2+ \dots + w_{i}+y_{i}+P_{i}$, such that for all $j \leq i$, we have $w_j \in N_{j-1} \setminus B_{j-1}$ and $P_j$ consistent with the extension of Alice's choosing, that ends at $u$. 
We refer to the process $w+y+P'+ \dots + w_{i}+y_{i}+P_{i}$ as the \emph{trace} of $u$.

Alice can then semi-consistently extend this process through $u$ via neighbour $y_u$ as long as each extension adds length at most $\ell$ to the length of the process, and as long as all the extension paths begin at $y_u$ and are subpaths of the same maximal bichromatic path $P_u$.

Bob then bans a set of vertices $B_i$ of size at most $\gamma |N_i|$, and we let $R_{i+1}$ be the set of vertices reachable by an $i$-step process $w+y+P'+ \dots + w_{i+1}+y_{i+1}+P_{i+1}$ of the form Alice seeks (though not necessarily augmenting) such that for all $j \leq i+1$, we have $w_j \in N_{j-1} \setminus B_{j-1}$ and $P_j$ consistent with the extension of Alice's choosing. 
Observe again that the stepping processes considered are represented by strings which are $(2,\ell, \emptyset, \hat{B}, P)$-restricted with $\hat{B}_j = B_j$ and $\hat{B}_0 = \emptyset$.
We let $\hat{N}_i = N_i - B_i$

The proposition, we wish to prove, can now be stated as follows:
\begin{proposition} \label{prop:main}
    Suppose that $\ell \geq \ellval $, and that $\aval \geq a \geq \avallb$, $\gamma \leq \gammaval$. Then there exists a strategy for Alice, so that they win within $\sqrt{\log n}$ rounds. 
\end{proposition}

In order to prove Proposition~\ref{prop:main}, we first describe the strategy that Alice will employ. 
Alice will play as follows: Firstly, Alice will fix some arbitrary total ordering $\prec$ on the set of strings representing stepping process beginning at $w$.
In the $i^{\text{th}}$ round, they will ban a subset of vertices from $R_i$, which could potentially lead to an extension, which is not semi-consistent. 
To be more precise, Alice will ban any vertex $u$ from $R_i$ such that there exists a $1$-step stepping process, beginning at $u$, which visits a point in the $2$-hop neighbourhood of $\{w\} \cup \paren{\bigcup \limits_{j=1}^{i-1} R_{j}}$. Additionally, any vertex $u$ belonging to the $2$-hop neighbourhood of $\{w\} \cup \paren{\bigcup \limits_{j=1}^{i-1} R_{j}}$ is also banned. 

They will then let the remaining points form $N_i$. To get the extension path through a point $u \in N_i$, they will extend the semi-consistent stepping process, which reaches $u$ for the first time during the $i^{\text{th}}$ step, that is represented by the string that is smallest according to $\prec$. 
Naturally, if Alice finds any extension which is augmenting with an extension path of length at most $\ell$, the game is done, and Alice wins. 

At this point it holds by assumption that any point in $N_i$ has at least one process to choose as trace, however we wish to show that in fact any extension consistent with Alice' choice of extension path will be a semi-consistent stepping process. 
This will help us show that the sets $N_j$ and $R_j$ grow rapidly as $j$ increases.
This in turn allows us to prove that the the process must terminate rather quickly, which under the above circumstances, can only happen if Alice wins.
We begin by showing that any extension consistent with Alice's choice of extension path will result in a semi-consistent stepping process:
\begin{lemma} \label{lma:extension good}
    When playing the strategy described above, for any vertex $u \in N_i$ with trace $w+y+P'+ \dots + w_{i}+y_{i}+P_{i}$ and any subpath $P_{i+1} \subset P_u$ beginning at $y_u$, we have that the process $w+y+P'+ \dots + w_{i}+y_{i}+P_{i}+u+y_u+P_{i+1}$ is semi-consistent. 
\end{lemma}
\begin{proof}
    Suppose that it is not semi-consistent. By assumption, we know that the trace is semi-consistent and that the extension is performed in a semi-consistent manner, so it must be the case that either $u$ overlaps with the $2$-hop neighbourhood of $w+y+P'+ \dots + w_{i}$, that $P_{i+1}$ overlaps with the $2$-hop neighbourhood of $w+y+P'+ \dots + w_{i}$, or that $P_i$ is a subpath of $P_u$. 

    Observe, however, that for any of these cases to occur, it must be the case that either $u$ belongs to the $2$-hop neighbourhood of $\bigcup \limits_{j=1}^{i-1} R_{i-1}$ or that there exists a $1$-step process beginning at $u$ that reaches a point in the $2$-hop neighbourhood of $\bigcup \limits_{j=1}^{i-1} R_{i-1}$.

    These cases cannot occur, as any such vertices are banned by Alice, and the lemma follows.
\end{proof} 
The above lemma essentially says that any sufficiently short and valid extension through a point in $N_i$ will end at a vertex belonging to $R_{i+1}$. As we will see shortly, this ensures that $R_{i+1}$ is a lot bigger than $N_i$. 
Before showing this, we first show that Alice will not ban too many vertices by themselves:
\begin{lemma} \label{lma:Aban}
    When playing the strategy above, Alice bans at most $|\{w\} \cup \paren{\bigcup \limits_{j=1}^{i-1} R_{j}}|\cdot{} 4(\Delta_{\max}+1)^5$ vertices from $R_i$. 
\end{lemma}
\begin{proof}
    We first count the number of vertices in $G$ for which there exists a $1$-step stepping process which visits a point in the $2$-hop neighbourhood of $\bigcup \limits_{j=1}^{i-1} R_{i-1}$. Let $u$ be such a vertex visiting the set via stepping process $C$, and let $v \in \{w\} \cup \paren{\bigcup \limits_{j=1}^{i-1} R_{j}}$ be the first vertex that $C$ visits. 
    Then $X_1[v]$ will contain a complete string of length at most $1$ representing $C$, i.e.\ $u \in X_1[v]$ by Lemma~\ref{lma:prefixcount}. Observe that any $1$-step stepping process is semi-consistent, and so the lemma can be applied. 
    Hence, we can bound the number of possible choice of $u$ as $\sum \limits_{v \in F} |X_1[v]| \leq |F| \cdot{} 3(\Delta_{\max}+1)^3$, where $F$ denotes the $2$-hop neighbourhood of $\{w\} \cup \paren{\bigcup \limits_{j=1}^{i-1} R_{j}}$. 
    By additionally banning any vertices belonging to $F$ and noting that $|F| \leq (\Delta_{\max}+1)^2 \cdot{} |\{w\} \cup \paren{\bigcup \limits_{j=1}^{i-1} R_{j}}|$, the lemma follows. 
\end{proof}

The following lemma is the core of the proof. It shows that if Alice has not won by the $i^{\text{th}}$ round, then $N_i$ and $R_i$ must experience exponential growth. The approach applied in the proof is similar to an approach applied by Christiansen in~\cite{Christiansen}.
\begin{lemma} \label{lma:game}
    Let $\ell \geq \ellval$, $\aval \geq a \geq \avallb $, $\gamma \leq \gammaval$.
    Suppose Alice plays the strategy described above, but that they did not win in the first $i$ rounds. Then, regardless of Bob's strategy, the following holds \emph{after} Bob has banned at most $\gamma|N_i|$ vertices: 
    \begin{enumerate}
        \item $|\hat{N}_i| \geq \tfrac{\ell}{8(\Delta_{\max}+1)^{3}}\paren{1 + \sum \limits_{j = 1}^{i-1} |\hat{N}_{j}|} $
        \item $|\hat{N}_i| \geq \tfrac{|R_i|}{2} $
    \end{enumerate}
\end{lemma}
\begin{proof}
    The proof is by strong induction on $i$. For $i = 1$, we first observe that $|P| > \ell$, since Alice did not win in the first round. Hence, we conclude that $|R_1| = \ell $.
    Next, applying Lemma~\ref{lma:Aban} yields that after Alice have banned their choice of vertices there are at least $\ell - 4(\Delta_{\max}+1)^5$ valid choices left in $R_1$ for $N_1$. Since Alice picks all such vertices, we find that $|\hat{N}_1| \geq (1-\gamma)(\ell - 4(\Delta_{\max}+1)^5)$ after Bob has banned vertices, and so we find that the base case holds.

    Next, we move to the induction step, so assume that the statement is true for $1, 2, \dots, i-1$. 
    Observe that any coloured edge is part of at most $(\Delta_{\max}+1)$ maximal bichromatic paths, namely one for each choice of second colour. 
    By assumption, Alice does not win in round $i$, and so every extension through a point in $N_{i-1}$ uses a bichromatic path of length $> \ell$. Hence, we can lower bound the size of $R_{i}$ as follows: 
    Let $H$ be the subgraph of $G$ spanned by the first $\ell$ edges of every extension path picked by Alice through every vertex in $\hat{N}_{i-1}$. 
    By the observation above, an edge can belong to at most $2(\Delta_{\max}+1) \cdot{} (\Delta_{\max}+1)$ of these extension paths -- once for every endpoint of maximal bichromatic path going through the edge. 
    Note than then $|E(H)| \geq \tfrac{\ell |\hat{N}_{i-1}|}{2(\Delta_{\max}+1)^{2}}$ and $|V(H)| = |R_i|$. 
    By the hand-shaking lemma, we have that $\tfrac{|E(H)|}{|V(H)|} \leq \tfrac{(\Delta_{\max} + 1)}{2}$, and so we find that $|R_{i}| \geq \tfrac{\ell |\hat{N}_{i-1}|}{(\Delta_{\max}+1)^{3}}$. 

    It follows by Lemma~\ref{lma:Aban} that Alice bans at most $\{w\} \cup \paren{\bigcup \limits_{j=1}^{i-1} R_{j}}$ of these points. That leaves at least 
    $$|N_i| = |R_i| -  |\{w\} \cup \paren{\bigcup \limits_{j=1}^{i-1} |R_{j}|}|\cdot{} 4(\Delta_{\max}+1)^5 \geq \tfrac{\ell |\hat{N}_{i-1}|}{(\Delta_{\max}+1)^{3}} -|\{w\} \cup \paren{\bigcup \limits_{j=1}^{i-1} |R_{j}|}|\cdot{} 4(\Delta_{\max}+1)^5 $$ 
    valid extension points for Alice. By the induction hypothesis, we can estimate as follows: 
    \begin{align*}
        |\{w\} \cup \paren{\bigcup \limits_{j=1}^{i-1} R_{j}}|\cdot{} 4(\Delta_{\max}+1)^5 &\leq 4(1+2(1+\sum \limits_{j = 1}^{i-1} |\hat{N}_{j}|))(\Delta_{\max}+1)^5 \\
        & \leq 12(1+\sum \limits_{j = 1}^{i-1} |\hat{N}_{j}|)(\Delta_{\max}+1)^5
    \end{align*}
    and 
    \begin{align*}
        \frac{\ell |\hat{N}_{i-1}|}{(\Delta_{\max}+1)^{3}} \geq \frac{\ell}{2(\Delta_{\max}+1)^{3}} \paren{1+\sum \limits_{j = 1}^{i-1} |\hat{N}_{j}|}
    \end{align*}
    Hence, it follows that if $\ell \geq \ellval$, then we have 
    \begin{align*}
        |\{w\} \cup \paren{\bigcup \limits_{j=1}^{i-1} R_{j}}|\cdot{} 4(\Delta_{\max}+1)^5 &\leq \tfrac{|R_{i}|}{100}
    \end{align*}
    and that
    \begin{align*}
        |N_i| &\geq \paren{\tfrac{\ell}{2(\Delta_{\max}+1)^{3}} - 12(\Delta_{\max}+1)^5} \paren{1+\sum \limits_{j = 1}^{i-1} |\hat{N}_{j}|} \\
        &\geq \tfrac{\ell}{4(\Delta_{\max}+1)^{3}}\paren{1+\sum \limits_{j = 1}^{i-1} |\hat{N}_{j}|} 
    \end{align*}
    The lemma now follows by observing that even if Bob bans $\gamma|N_i|$ of these points, we still find: 
    \begin{align*}
        |\hat{N}_i| \geq \tfrac{99}{100}(1-\gamma)|R_i| \geq \tfrac{|R_{i}|}{2}
    \end{align*}
    and 
    \begin{align*}
        |\hat{N}_i| \geq \tfrac{\ell}{8(\Delta_{\max}+1)^{3}}\paren{1+\sum \limits_{j = 1}^{i-1} |\hat{N}_{j}|} 
    \end{align*}    
\end{proof}
Now we show the proposition:
\begin{proof}[Proof of Proposition~\ref{prop:main}]
    Suppose that Alice has not won after round $\sqrt{\log n}$ concludes. Then by Lemma~\ref{lma:game} and induction, it follows that $\hat{N}_{\sqrt{\log n}+1}$ has size at least $\paren{\tfrac{\ell}{8(\Delta_{\max}+1)^{3}}}^{\sqrt{\log n}+1} > n$, which is a contradiction as $\hat{N}_{\sqrt{\log n}+1} \subset V(G)$ is not a multi-set. 
\end{proof}

\subsubsection{Theorem}
In this subsection, we finally deduce Theorem~\ref{thm:main}, which we restate below for the readers convenience: 
\begin{theorem}
    Let $w$ be any node, $y$ any neighbour of $w$ and $P$ any maximal bichromatic path beginning at $y$. Then $w$ is $(a,\sqrt{\log n}, \ell, \emptyset,\emptyset, P)$-heavy. 
\end{theorem}
\begin{proof}
    Suppose for contradiction that $w$ is not $(a,\sqrt{\log n}, \ell, \emptyset,\emptyset, P)$-heavy. Then by Lemma~\ref{lma:suppression}, there exists a graded set $D$ with $D_{i} \leq \ell^{i-1} \cdot{} a$ which $(a,\sqrt{\log n}, \ell, \emptyset, \emptyset, P)$-\emph{suppresses} $w$. 

    We now claim that Bob has a winning strategy even if the conditions of Proposition~\ref{prop:main} is met. This would yield the desired contradiction. 
    Hereto, we describe the strategy of Bob. 
    Bob will simply play the strategy $B_i = D_i$. Since the graded set $D$ $(a,\sqrt{\log n}, \ell, \emptyset, \emptyset, P)$-\emph{suppresses} $w$, it follows that there exist no non-empty $(\sqrt{\log n}, \ell, \emptyset, \emptyset, P)$-restricted sets containing strings representing semi-consistent and augmenting stepping processes. 
    As such, if we can show that the strategy is valid with respect to $\gamma$, we achieve the desired contradiction. 
    To this end, observe that by Lemma~\ref{lma:game} as long as Alice has not won yet, we must have $|N_i| \geq (\tfrac{\ell}{8(\Delta_{\max}+1)^{3}})^{i} \geq \ell^{i-1} \tfrac{a}{\gamma}$ for all $i \leq \sqrt{\log n}$, and so Bob's strategy is indeed valid.
\end{proof}
The idea behind the dynamic algorithm and the associated data structures presented in the rest of the section, is to exploit the structural results derived in this section. 
In particular, we will make Corollary~\ref{cor:setfuse} algorithmic by dynamically and explicitly maintaining sets of the form outlined in the corollary. 
The idea is that we can use the approach of Lemma~\ref{lma:fuse} to splice together relevant parts of sets of lower order along the bichromatic path to form a higher order set. 
By carefully repairing and maintaining such sets, we can maintain augmenting stepping processes for all starting points. 
This allows us to efficiently perform colouring updates, as one can build a Vizing fan around an endpoint of an inserted edge, and then use the stored stepping processes to guide the search for an augmenting Vizing chain. 
Thus the key part of the algorithm becomes the data structure for maintaining these sets. To do so, we heavily rely on the fact that any colouring or un-colouring of an edge cannot affect too many data structures, as not too many points can actually reach such an edge via semi-consistent stepping processes. 
Therefore, the data structure can afford to explicitly locate and fix any data structures affected by such a change. 
We expand upon this in the following subsections, but first we begin by outlining the contents of the data structures that we will maintain.

%% file: main.bib
@article{2BERNSHTEYN,
	author = {Anton Bernshteyn},
	doi = {https://doi.org/10.1016/j.jctb.2021.10.004},
	issn = {0095-8956},
	journal = {Journal of Combinatorial Theory, Series B},
	pages = {319-352},
	title = {A {F}ast {D}istributed {A}lgorithm for $({\Delta}+ 1)$-{E}dge-{C}oloring},
	url = {https://www.sciencedirect.com/science/article/pii/S009589562100085X},
	volume = {152},
	year = {2022}
}

@inproceedings{pettie,
  title={The complexity of distributed edge coloring with small palettes},
  author={Chang, Yi-Jun and He, Qizheng and Li, Wenzheng and Pettie, Seth and Uitto, Jara},
  booktitle={Proceedings of the Twenty-Ninth Annual ACM-SIAM Symposium on Discrete Algorithms},
  pages={2633--2652},
  year={2018},
  organization={SIAM}
}

@inproceedings{ghaffari2018deterministic,
  title={Deterministic {D}istributed {E}dge-{C}oloring with {F}ewer {C}olors},
  author={Ghaffari, Mohsen and Kuhn, Fabian and Maus, Yannic and Uitto, Jara},
  booktitle={Proceedings of the 50th Annual ACM SIGACT Symposium on Theory of Computing},
  pages={418--430},
  year={2018}
}

@article{grebik2020measurable,
  title={Measurable versions of {V}izing's theorem},
  author={Greb{\'\i}k, Jan and Pikhurko, Oleg},
  journal={Advances in Mathematics},
  volume={374},
  pages={107378},
  year={2020},
  publisher={Elsevier}
}

@article{vizing1964estimate,
  title={On an estimate of the chromatic class of a p-graph},
  author={Vizing, Vadim G},
  journal={Discret Analiz},
  volume={3},
  pages={25--30},
  year={1964}
}

@inproceedings{su2019towards,
  title={Towards the {L}ocality of {V}izing’s {T}heorem},
  author={Su, Hsin-Hao and Vu, Hoa T},
  booktitle={Proceedings of the 51st Annual ACM SIGACT Symposium on Theory of Computing},
  pages={355--364},
  year={2019}
}

@inproceedings{duan,
  author    = {Ran Duan and
               Haoqing He and
               Tianyi Zhang},
  editor    = {Timothy M. Chan},
  title     = {Dynamic {E}dge Coloring with {I}mproved {A}pproximation},
  booktitle = {Proceedings of the Thirtieth Annual {ACM-SIAM} Symposium on Discrete
               Algorithms, {SODA} 2019, San Diego, California, USA, January 6-9,
               2019},
  pages     = {1937--1945},
  publisher = {{SIAM}},
  year      = {2019},
  url       = {https://doi.org/10.1137/1.9781611975482.117},
  doi       = {10.1137/1.9781611975482.117},
  bibsource = {dblp computer science bibliography, https://dblp.org}
}

@article{alon2003simple,
  title={A simple algorithm for edge-coloring bipartite multigraphs},
  author={Alon, Noga},
  journal={Information Processing Letters},
  volume={85},
  number={6},
  pages={301--302},
  year={2003},
  publisher={Elsevier Science}
}

@article{barenboim2017fully,
  title={Fully-{D}ynamic {G}raph {A}lgorithms with {S}ublinear {T}ime inspired by {D}istributed {C}omputing},
  author={Barenboim, Leonid and Maimon, Tzalik},
  journal={Procedia Computer Science},
  volume={108},
  pages={89--98},
  year={2017},
  publisher={Elsevier}
}

@inproceedings{bhattacharya2018dynamic,
  title={Dynamic {A}lgorithms for {G}raph {C}oloring},
  author={Bhattacharya, Sayan and Chakrabarty, Deeparnab and Henzinger, Monika and Nanongkai, Danupon},
  booktitle={Proceedings of the Twenty-Ninth Annual ACM-SIAM Symposium on Discrete Algorithms},
  pages={1--20},
  year={2018},
  organization={SIAM}
}

@article{cole2001edge,
  title={Edge-{C}oloring {B}ipartite {M}ultigraphs in ${O}({E} \log {D})$ {T}ime},
  author={Cole, Richard and Ost, Kirstin and Schirra, Stefan},
  journal={Combinatorica},
  volume={21},
  number={1},
  pages={5--12},
  year={2001},
  publisher={Springer-Verlag GmbH}
}

@inbook{Gabow,
  title={Algorithms for {E}dge-{C}olouring {G}raphs},
  author={Gabow, Harold and Nishizeki, Takao and Kariv, Oded and Leven, Daniel and Tereda, Osamu},
  journal={Technical Report},
  year={1985}
}

@article{Ian,
  title={The {NP}-completeness of edge-coloring},
  author={Holyer, Ian},
  journal={SIAM Journal on computing},
  volume={10},
  number={4},
  pages={718--720},
  year={1981},
  publisher={SIAM}
}

@article{karloff1987efficient,
  title={Efficient {P}arallel {A}lgorithms for {E}dge {C}oloring {P}roblems},
  author={Karloff, Howard J and Shmoys, David B},
  journal={Journal of Algorithms},
  volume={8},
  number={1},
  pages={39--52},
  year={1987},
  publisher={Elsevier}
}

@inproceedings{Christiansen,
  author       = {Aleksander Bj{\o}rn Grodt Christiansen},
  editor       = {Barna Saha and
                  Rocco A. Servedio},
  title        = {The Power of Multi-step Vizing Chains},
  booktitle    = {Proceedings of the 55th Annual {ACM} Symposium on Theory of Computing,
                  {STOC} 2023, Orlando, FL, USA, June 20-23, 2023},
  pages        = {1013--1026},
  publisher    = {{ACM}},
  year         = {2023},
  url          = {https://doi.org/10.1145/3564246.3585105},
  doi          = {10.1145/3564246.3585105},
  bibsource    = {dblp computer science bibliography, https://dblp.org}
}

@article{BernDhawan,
  author       = {Anton Bernshteyn and
                  Abhishek Dhawan},
  title        = {Fast algorithms for Vizing's theorem on bounded degree graphs},
  journal      = {CoRR},
  volume       = {abs/2303.05408},
  year         = {2023},
  url          = {https://doi.org/10.48550/arXiv.2303.05408},
  doi          = {10.48550/ARXIV.2303.05408},
  eprinttype    = {arXiv},
  eprint       = {2303.05408},
  bibsource    = {dblp computer science bibliography, https://dblp.org}
}

@book{cormen2022introduction,
  title={Introduction to algorithms},
  author={Cormen, Thomas H and Leiserson, Charles E and Rivest, Ronald L and Stein, Clifford},
  year={2022},
  publisher={MIT press}
}

@article{HolmEtAL,
  author       = {Stephen Alstrup and
                  Jacob Holm and
                  Kristian de Lichtenberg and
                  Mikkel Thorup},
  title        = {Maintaining information in fully dynamic trees with top trees},
  journal      = {{ACM} Trans. Algorithms},
  volume       = {1},
  number       = {2},
  pages        = {243--264},
  year         = {2005},
  url          = {https://doi.org/10.1145/1103963.1103966},
  doi          = {10.1145/1103963.1103966},
  bibsource    = {dblp computer science bibliography, https://dblp.org}
}

@inproceedings{DBLP:conf/soda/BhattacharyaGW21,
  author       = {Sayan Bhattacharya and
                  Fabrizio Grandoni and
                  David Wajc},
  editor       = {D{\'{a}}niel Marx},
  title        = {Online Edge Coloring Algorithms via the Nibble Method},
  booktitle    = {Proceedings of the 2021 {ACM-SIAM} Symposium on Discrete Algorithms,
                  {SODA} 2021, Virtual Conference, January 10 - 13, 2021},
  pages        = {2830--2842},
  publisher    = {{SIAM}},
  year         = {2021},
  url          = {https://doi.org/10.1137/1.9781611976465.168},
  doi          = {10.1137/1.9781611976465.168},
  bibsource    = {dblp computer science bibliography, https://dblp.org}
}

@inbook{CostaEtAl,
author = {Sayan Bhattacharya and Martín Costa and Nadav Panski and Shay Solomon},
title = {Nibbling at Long Cycles: Dynamic (and Static) Edge Coloring in Optimal Time},
booktitle = {Proceedings of the 2024 Annual ACM-SIAM Symposium on Discrete Algorithms (SODA)},
chapter = {},
year = {2024},
pages = {3393-3440},
doi = {10.1137/1.9781611977912.122},
URL = {https://epubs.siam.org/doi/abs/10.1137/1.9781611977912.122},
eprint = {https://epubs.siam.org/doi/pdf/10.1137/1.9781611977912.122}
}

@article{DBLP:journals/corr/abs-2311-08367,
  author       = {Sayan Bhattacharya and
                  Mart{\'{\i}}n Costa and
                  Nadav Panski and
                  Shay Solomon},
  title        = {Arboricity-Dependent Algorithms for Edge Coloring},
  journal      = {CoRR},
  volume       = {abs/2311.08367},
  year         = {2023},
  url          = {https://doi.org/10.48550/arXiv.2311.08367},
  doi          = {10.48550/ARXIV.2311.08367},
  eprinttype    = {arXiv},
  eprint       = {2311.08367},
  bibsource    = {dblp computer science bibliography, https://dblp.org}
}

@article{Vlieghe,
  author       = {Aleksander B. J. Christiansen and
                  Eva Rotenberg and
                  Juliette Vlieghe},
  title        = {Sparsity-Parameterised Dynamic Edge Colouring},
  journal      = {CoRR},
  volume       = {abs/2311.10616},
  year         = {2023},
  url          = {https://doi.org/10.48550/arXiv.2311.10616},
  doi          = {10.48550/ARXIV.2311.10616},
  eprinttype    = {arXiv},
  eprint       = {2311.10616},
  bibsource    = {dblp computer science bibliography, https://dblp.org}
}

@article{Costa2,
  author       = {Sayan Bhattacharya and
                  Mart{\'{\i}}n Costa and
                  Nadav Panski and
                  Shay Solomon},
  title        = {Density-Sensitive Algorithms for ({\(\Delta\)}+1)-Edge Coloring},
  journal      = {CoRR},
  volume       = {abs/2307.02415},
  year         = {2023},
  url          = {https://doi.org/10.48550/arXiv.2307.02415},
  doi          = {10.48550/ARXIV.2307.02415},
  eprinttype    = {arXiv},
  eprint       = {2307.02415},
  bibsource    = {dblp computer science bibliography, https://dblp.org}
}

@article{Sinnamon,
  author       = {Corwin Sinnamon},
  title        = {A Randomized Algorithm for Edge-Colouring Graphs in O(m{\(\surd\)}n)
                  Time},
  journal      = {CoRR},
  volume       = {abs/1907.03201},
  year         = {2019},
  url          = {http://arxiv.org/abs/1907.03201},
  eprinttype    = {arXiv},
  eprint       = {1907.03201},
  bibsource    = {dblp computer science bibliography, https://dblp.org}
}

@article{split1,
  author       = {Mohsen Ghaffari and
                  Juho Hirvonen and
                  Fabian Kuhn and
                  Yannic Maus and
                  Jukka Suomela and
                  Jara Uitto},
  title        = {Improved distributed degree splitting and edge coloring},
  journal      = {Distributed Comput.},
  volume       = {33},
  number       = {3-4},
  pages        = {293--310},
  year         = {2020},
  url          = {https://doi.org/10.1007/s00446-018-00346-8},
  doi          = {10.1007/S00446-018-00346-8},
  bibsource    = {dblp computer science bibliography, https://dblp.org}
}

@inproceedings{split2,
  author       = {Mohsen Ghaffari and
                  Hsin{-}Hao Su},
  editor       = {Philip N. Klein},
  title        = {Distributed Degree Splitting, Edge Coloring, and Orientations},
  booktitle    = {Proceedings of the Twenty-Eighth Annual {ACM-SIAM} Symposium on Discrete
                  Algorithms, {SODA} 2017, Barcelona, Spain, Hotel Porta Fira, January
                  16-19},
  pages        = {2505--2523},
  publisher    = {{SIAM}},
  year         = {2017},
  url          = {https://doi.org/10.1137/1.9781611974782.166},
  doi          = {10.1137/1.9781611974782.166},
  bibsource    = {dblp computer science bibliography, https://dblp.org}
}

@article{DBLP:journals/corr/abs-2401-10538,
  author       = {Michael Elkin and
                  Ariel Khuzman},
  title        = {Deterministic Simple (1+{\textdegree}){\(\Delta\)}-Edge-Coloring in
                  Near-Linear Time},
  journal      = {CoRR},
  volume       = {abs/2401.10538},
  year         = {2024},
  url          = {https://doi.org/10.48550/arXiv.2401.10538},
  doi          = {10.48550/ARXIV.2401.10538},
  eprinttype    = {arXiv},
  eprint       = {2401.10538},
  bibsource    = {dblp computer science bibliography, https://dblp.org}
}

@inproceedings{sawlani2020near,
  title={Near-optimal fully dynamic densest subgraph},
  author={Sawlani, Saurabh and Wang, Junxing},
  booktitle={Proceedings of the 52nd Annual ACM SIGACT Symposium on Theory of Computing},
  pages={181--193},
  year={2020}
}

@inbook{Chekurietal,
author = {Chandra Chekuri and Aleksander Bjørn Christiansen and Jacob Holm and Ivor van der Hoog and Kent Quanrud and Eva Rotenberg and Chris Schwiegelshohn},
title = {Adaptive Out-Orientations with Applications},
booktitle = {Proceedings of the 2024 Annual ACM-SIAM Symposium on Discrete Algorithms (SODA)},
chapter = {},
pages = {3062-3088},
doi = {10.1137/1.9781611977912.110},
URL = {https://epubs.siam.org/doi/abs/10.1137/1.9781611977912.110},
eprint = {https://epubs.siam.org/doi/pdf/10.1137/1.9781611977912.110}
}

@inproceedings{New1,
  author       = {Sepehr Assadi},
  editor       = {Yossi Azar and
                  Debmalya Panigrahi},
  title        = {Faster Vizing and Near-Vizing Edge Coloring Algorithms},
  booktitle    = {Proceedings of the 2025 Annual {ACM-SIAM} Symposium on Discrete Algorithms,
                  {SODA} 2025, New Orleans, LA, USA, January 12-15, 2025},
  pages        = {4861--4898},
  publisher    = {{SIAM}},
  year         = {2025},
  url          = {https://doi.org/10.1137/1.9781611978322.165},
  doi          = {10.1137/1.9781611978322.165},
  bibsource    = {dblp computer science bibliography, https://dblp.org}
}

@inproceedings{New2,
  author       = {Sayan Bhattacharya and
                  Din Carmon and
                  Mart{\'{\i}}n Costa and
                  Shay Solomon and
                  Tianyi Zhang},
  title        = {Faster ({\(\Delta\)}+1)-Edge Coloring: Breaking the m{\(\surd\)}n
                  Time Barrier},
  booktitle    = {65th {IEEE} Annual Symposium on Foundations of Computer Science, {FOCS}
                  2024, Chicago, IL, USA, October 27-30, 2024},
  pages        = {2186--2201},
  publisher    = {{IEEE}},
  year         = {2024},
  url          = {https://doi.org/10.1109/FOCS61266.2024.00128},
  doi          = {10.1109/FOCS61266.2024.00128},
  bibsource    = {dblp computer science bibliography, https://dblp.org}
}

@inproceedings{New3,
  author       = {Sepehr Assadi and
                  Soheil Behnezhad and
                  Sayan Bhattacharya and
                  Mart{\'{\i}}n Costa and
                  Shay Solomon and
                  Tianyi Zhang},
  editor       = {Michal Kouck{\'{y}} and
                  Nikhil Bansal},
  title        = {Vizing's Theorem in Near-Linear Time},
  booktitle    = {Proceedings of the 57th Annual {ACM} Symposium on Theory of Computing,
                  {STOC} 2025, Prague, Czechia, June 23-27, 2025},
  pages        = {24--35},
  publisher    = {{ACM}},
  year         = {2025},
  url          = {https://doi.org/10.1145/3717823.3718265},
  doi          = {10.1145/3717823.3718265},
  bibsource    = {dblp computer science bibliography, https://dblp.org}
}

@inproceedings{ONew1,
  author       = {Joakim Blikstad and
                  Ola Svensson and
                  Radu Vintan and
                  David Wajc},
  editor       = {Bojan Mohar and
                  Igor Shinkar and
                  Ryan O'Donnell},
  title        = {Online Edge Coloring Is (Nearly) as Easy as Offline},
  booktitle    = {Proceedings of the 56th Annual {ACM} Symposium on Theory of Computing,
                  {STOC} 2024, Vancouver, BC, Canada, June 24-28, 2024},
  pages        = {36--46},
  publisher    = {{ACM}},
  year         = {2024},
  url          = {https://doi.org/10.1145/3618260.3649741},
  doi          = {10.1145/3618260.3649741},
  bibsource    = {dblp computer science bibliography, https://dblp.org}
}

@inproceedings{ONew2,
  author       = {Joakim Blikstad and
                  Ola Svensson and
                  Radu Vintan and
                  David Wajc},
  editor       = {Yossi Azar and
                  Debmalya Panigrahi},
  title        = {Deterministic Online Bipartite Edge Coloring},
  booktitle    = {Proceedings of the 2025 Annual {ACM-SIAM} Symposium on Discrete Algorithms,
                  {SODA} 2025, New Orleans, LA, USA, January 12-15, 2025},
  pages        = {1593--1606},
  publisher    = {{SIAM}},
  year         = {2025},
  url          = {https://doi.org/10.1137/1.9781611978322.49},
  doi          = {10.1137/1.9781611978322.49},
  bibsource    = {dblp computer science bibliography, https://dblp.org}
}

@inproceedings{ONew3,
  author       = {Aditi Dudeja and
                  Rashmika Goswami and
                  Michael Saks},
  editor       = {Yossi Azar and
                  Debmalya Panigrahi},
  title        = {Randomized Greedy Online Edge Coloring Succeeds for Dense and Randomly-Ordered
                  Graphs},
  booktitle    = {Proceedings of the 2025 Annual {ACM-SIAM} Symposium on Discrete Algorithms,
                  {SODA} 2025, New Orleans, LA, USA, January 12-15, 2025},
  pages        = {4948--4982},
  publisher    = {{SIAM}},
  year         = {2025},
  url          = {https://doi.org/10.1137/1.9781611978322.168},
  doi          = {10.1137/1.9781611978322.168},
  bibsource    = {dblp computer science bibliography, https://dblp.org}
}
